\documentclass[letterpaper,twocolumn,twoside]{IEEEtran}
\usepackage{amsmath}
\usepackage{mathtools}
\usepackage{amsmath}
\usepackage{amssymb}
\usepackage{cite}
\usepackage{array}
\usepackage{amsthm}
\usepackage{amsmath,dsfont}
\usepackage{tabulary}
\usepackage{multirow}
\usepackage{cleveref}
\usepackage{subfig}
\usepackage[font=small]{caption}
\usepackage{color}
\usepackage{pifont}
\usepackage{epsfig}
\usepackage{graphicx}
\usepackage{url}
\usepackage{amssymb}
\IEEEoverridecommandlockouts
\graphicspath{ {./figs/} }

\newcommand{\s}{\mathsf{s}}
\newcommand{\PPP}{$\mathsf{PPP}$}

\newcommand{\MBSs}{\Phi_m}
\newcommand{\SBSs}{\Phi_s}

\newcommand{\F}{\mathrm{F}}

\newcommand{\power}[1]{\mathrm{P}_{#1}}
\newcommand{\res}{\mathrm{W}}

\newcommand{\BW}{\mathrm{W}}

\newcommand{\UEdensity}{\lambda_{u}}
\newcommand{\BSprocess}{\Phi_b}
\newcommand{\UEprocess}{\Phi_{u}}
\newcommand{\ULUEprocess}{\Phi_{ul}}
\newcommand{\DLUEprocess}{\Phi_{dl}}
\newcommand{\DLfrac}{\eta}

\newcommand{\SINR}{\mathsf{SINR}}

\newcommand{\dlos}{\mathrm{D}_\mathrm{LOS}}

\newcommand{\plos}{p_\mathrm{LOS}}

\newcommand{\bwidth}{\mathrm{W}}

\newtheorem{thm}{{\bf Theorem}}
\newtheorem{cor}{Corollary}
\newtheorem{rem}{Remark}
\newtheorem{lem}{Lemma}
\newtheorem{assumption}{Assumption}
\newtheorem{approximation}{Approximation}

\theoremstyle{definition}
\newtheorem{definition}{Definition}

\begin{document}

\title{Performance of Dynamic and Static TDD in Self-backhauled Millimeter Wave Cellular Networks}
\author{Mandar N. Kulkarni, Jeffrey G. Andrews and Amitava Ghosh\thanks{Email:
$\mathtt{\{mandar.kulkarni@,jandrews@ece.\}utexas.edu,amitava.}$ $\mathtt{ghosh@nokia-bell-labs.com}$. M. Kulkarni and J. Andrews are with the University of Texas at Austin, TX. A. Ghosh is with Nokia Bell Labs, IL.}}
\maketitle

\begin{abstract}
Initial deployments of millimeter wave (mmWave) cellular networks are likely to be enabled with self-backhauling. In this work, we propose a random spatial model to analyze uplink (UL) and downlink (DL) SINR distribution and mean rates corresponding to different access-backhaul and UL-DL resource allocation schemes in a self-backhauled mmWave cellular network with Poisson point process (PPP) deployment of users and base stations. In particular, we focus on heuristic implementations of static and dynamic time division duplexing (TDD) for access links with synchronized or unsynchronized access-backhaul (SAB or UAB) time splits. We propose PPP approximations to characterize the distribution of the new types of interference encountered with dynamic TDD and UAB. These schemes offer better resource utilization than static TDD and SAB, however potentially higher interference makes their choice non-trivial and the offered gains sensitive to different network parameters, including UL/DL traffic asymmetry, user load per BS or number of slave BSs per master BS. One can harness notable gains from UAB and/or dynamic TDD only if backhaul links are designed to have much larger throughput than the access links. 
\end{abstract}

\section{Introduction}
Self-backhauling offers a simple cost-saving strategy to enable dense millimeter wave cellular networks\cite{Taori15,Ran14,SinJSAC14}. A self-backhauled network has two types of base stations (BSs) -- master BSs (MBSs) and slave BSs (SBSs). SBSs wirelessly backhaul users' data to/from the fiber backhauled MBSs through either a direct wireless connection or over multiple SBS-SBS hops, sharing the spectrum with access links\cite{selfbackpatent06}. A fundamental problem for designing a self-backhauled network is to split the available time-frequency resources between uplink (UL) and downlink (DL) and for the access and backhaul links. In this work, we develop a generic random spatial model for studying the resource allocation problem in two hop self-backhauled mmWave cellular networks, with a focus on comparing static and dynamic time division duplexing (TDD) with synchronized or unsynchronized access-backhaul (SAB or UAB). 

\subsection{Dynamic TDD with unsynchronized access-backhaul:-- motivation and prior work}
Conventionally, a network-wide static split of resources is done between UL and DL, meaning that every BS follows a common UL-DL split of time-frequency resources. Such a static split can be very inefficient in dense networks wherein the load per base station is highly variable, as shown in Fig.~\ref{fig:DTDDmotiv}. Although the network has overall $50\%$ UL users, the fraction of UL users per BS varies from $16\%$ to $100\%$, and thus a network wide $50-50$ split between UL and DL resources is wasteful. Dynamic TDD is a class of scheduling schemes wherein every BS is free to choose its own UL-DL split\cite{Li00,Shen12}. Widespread use of this TDD scheme was challenging for sub-6GHz networks owing to cross-interference between UL transmissions in one cell and DL transmissions in neighboring cells\cite{Li00,Shen12}. Since DL transmissions generally have more power than UL, dynamic TDD generally hurts UL signal to interference plus noise ratio (SINR). At mmWave frequencies, however, dynamic TDD is expected to perform much better given the likely noise-limited behaviour due to directionality and large bandwidth \cite{Ran14,Rois15,GuptaKul2016}. Furthermore, the significance of enabling dynamic TDD in future cellular networks is predicted to be even higher for meeting ultra low latency and high throughput requirements of the future wireless technologies\cite{NokiaTDD, Eric16}. Stochastic geometry has been used to quantify the cross UL-DL interference effects through calculating the {\em SINR distribution} in sub-6GHz cellular \cite{Yu15}, device-to-device enhanced networks \cite{Sun15} and UL mmWave cellular networks\cite{GuptaKul2016} but there is no comprehensive UL-DL {\em rate} analysis with dynamic TDD. In this work, we characterize the gains with dynamic TDD in mmWave cellular networks for UL and DL through explicit mean rate formulas as a function of network parameters and a simple interference mitigation scheme.

Incorporating relays in cellular networks was an afterthought, primarily for coverage enhancement, in current deployments of cellular networks. Two-hop relaying was introduced in 3GPP release 10\cite[Ch. 18]{Dahlman14}. However, mmWave cellular networks are expected to have dense deployments right from the start to provide sufficient coverage overcoming the enhanced blockage effects and to meet the desired 5G data rates for enabling extreme mobile broadband applications\cite{Ghosh14,BaiHea14}. Thus, a simple cost saving strategy to enable flexible deployments is to have a fraction of BSs wirelessly backhauling data to the rest which have fiber backhaul connectivity, that motivates self-backhauled mmWave cellular networks. Traditionally, in-band implementation of relay networks is restricted to synchronized access-backhaul (SAB), wherein the access and backhaul links are active on non-overlapping time slots\cite[Ch. 18]{Dahlman14}. However, from resource allocation perspective, an MBS needs more backhaul slots than SBSs in self-backhauled networks. This is not possible with the conventional SAB implementation. An example is shown in Fig.~\ref{fig:UABmotiv} wherein there are 2 SBSs connected to an MBS. With SAB, the second SBS is silent in a backhaul slot when first SBS is scheduled by the MBS. In fact, the second SBS could have utilized the unscheduled backhaul slots for communicating with its UEs. This issue will be magnified if there are tens or hundreds of SBSs connected to an MBS. An SBS {\em poaching} the unscheduled backhaul slots for access is said to employ an unsynchronized access-backhaul (UAB) strategy, wherein access and backhaul links need not be scheduled on orthogonal resource blocks. Introducing the above mentioned implementation of UAB, however, comes at a cost of increasing interference on the backhaul links which makes it non-trivial to choose UAB over SAB. Again, the subdued interference effects at mmWave make UAB attractive for practical implementations. UAB has been implicitly incorporated in algorithmic solutions to the resource allocation problem in sub-6GHz relay networks\cite{Viswanathan05} and more recently in mmWave self-backhauled networks\cite{Rois15,Nazmul17}\footnote{The term ``integrated access-backhaul" coined in \cite{Nazmul17,GhoshATT16} by Qualcomm and AT$\&$T is same as ``UAB" in this work, although our heuristic implementation has a more specific form described in Section~\ref{sec:sysmodel}.} In this work, we capture the tradeoff between increasing interference and better resource allocation with UAB through our random spatial model, and the analysis can be used to compute optimal {\em poaching probabilities} (defined in Section~\ref{sec:scheduling}) to strike a balance. In \cite{SinJSAC14}, UAB was implicitly employed, although the focus was on noise-limited mmWave cellular networks. Previous stochastic geometry analysis of relay networks, like in \cite{Lu15,Tabassum16,Sharma16}, did not incorporate UAB and also was focused on sub-6GHz cellular networks. 
\begin{figure*}
  \centering
\subfloat[Dynamic TDD: Varying fraction of UL users per BS. Triangles are BSs and $(x,y)$ are number of UL and DL UEs.]
{\label{fig:DTDDmotiv}{\includegraphics[width=0.9\columnwidth]{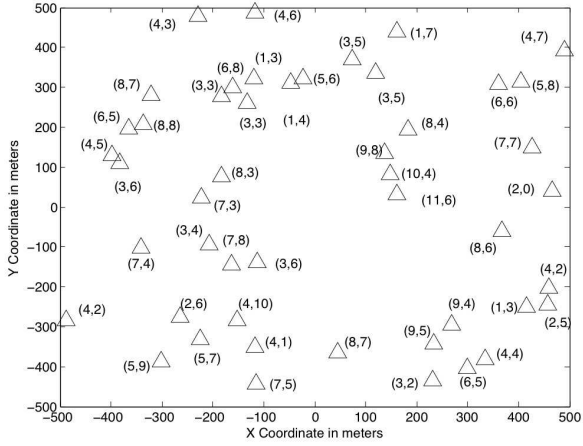}}}
\subfloat[UAB: Need more backhaul slots at MBS than SBS.]
{\label{fig:UABmotiv}{\includegraphics[width= 0.9\columnwidth]{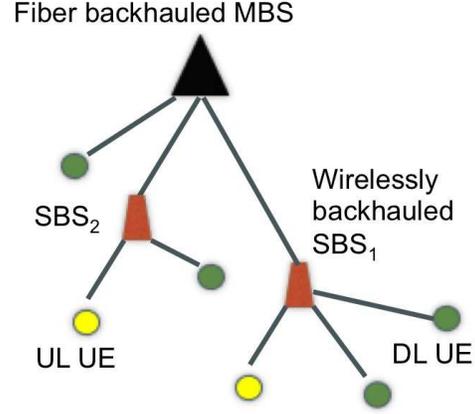}}}
\caption{Motivation for dynamic TDD and UAB.}
 \label{fig:motiv}
\end{figure*}
\subsection{Contributions}
\label{sec:contributions}
{\bf UL and DL analysis of dynamic TDD in mmWave cellular networks.}
This is the first work to our knowledge to analyze UL and DL SINR distribution and mean rates in dynamic TDD enabled mmWave cellular networks. We consider a time-slotted system and prioritize all initial slots in a typical frame for DL scheduling and later slots for UL scheduling. Such a prioritization is shown to have inherent UL interference mitigation and the variation of SINR across time slots can be as large as $10-15$ dB. This translates to some gain in mean rate as well, but is more crucial for decreasing UL SINR outage probabilities. PPP deployment for users and base stations was assumed for the analysis. 

{\bf UL and DL analysis of mmWave self-backhauled cellular networks with unsynchronized access-backhaul.}
We compare the achievable mean rates with SAB and UAB in self-backhauled mmWave cellular networks. The optimal number of slots to be exclusively allocated for access is shown to be non-increasing with UAB as compared to SAB. A PPP approximation is proposed and validated for characterizing the interference distribution with UAB, which we believe can have a variety of applications as mentioned in Section~\ref{sec:conclusion}. 

{\bf Engineering insights.}
The comparison of resource allocation schemes considered in this paper is fundamentally dependent on more than ten system parameters, and thus it is not possible to enumerate concrete regimes wherein one strategy will outperform another. Also, dynamic TDD may be the preferred choice over static TDD for DL users but not for UL users, and UAB with no exclusive access slots may be desirable for the users connected to SBSs but not for those connected to MBSs. The analytical formulae provided in this paper provides a transparent approach to compare the resource allocation schemes for different networks and propagation settings in terms of mean rates and SINR distributions of a typical UL and DL user in the network. Dynamic TDD and UAB {\em usually} outperform or at least provide similar performance to load aware static TDD and SAB in terms of mean rate of a typical user in millimeter wave cellular networks operating with large bandwidths (order of GHz). The gains of dynamic TDD over static TDD are larger for low load, and asymmetric traffic scenarios. Load aware static TDD can still be preferable over dynamic TDD in interference-limited highly loaded scenarios with symmetric UL and DL traffic requests on an average per BS. We further find that there is no need for asymmetric traffic or low UE load for gains with UAB over SAB and we just need sufficiently large number of SBS per MBS.  Self-backhauling is indeed a low cost coverage solution that can enable flexible deployments, but not particularly useful to enhance mean rates if the same antenna array is used by SBSs for both access and backhaul links. Employing higher spectral efficiency backhaul links is important to harvest the benefits of dynamic TDD and UAB. 

\section{System Model}
\label{sec:sysmodel}
\subsection{Spatial distribution of base stations and users}
Let $\MBSs$ and $\SBSs$ denote independent PPPs on $\mathbb{R}^2$ of MBSs and SBSs with density $\lambda_m$ and $\lambda_s$ BSs/km$^2$. Let $\BSprocess$ denote the superposition of the two BS PPPs and $\lambda_b = \lambda_m+\lambda_s$ denote its density. User equipments (UEs) are distributed as an independent homogeneous PPP $\;\UEprocess$ with density $\UEdensity$ UEs/km$^2$  on $\mathbb{R}^2$. A fraction $\DLfrac$ of UEs have DL requests and the rest of them UL. $\ULUEprocess$ and $\DLUEprocess$ denote the UL and DL UE point processes with densities  $ (1-\eta)\lambda_u$ and $\eta\lambda_u$, respectively. UEs always have data to transmit/receive. All devices are half duplex. 
\begin{table}
\caption{Notation summary and default numerical parameters}
\centering
\label{tab:notation}
\begin{tabulary}{\columnwidth}{|C | C| C|}
\hline
{\bf Notation} & {\bf Parameter(s)} & {\bf Value(s) if applicable} \\\hline
$\Phi_u$,$\BSprocess$, $\MBSs$,$\SBSs$ & UE, BS, MBS and SBS PPP on $\mathbb{R}^2$ & --\\\hline
$\lambda_u, \lambda_b,$ $\lambda_m, \lambda_s$ & Density of UE, BS, MBS and SBS PPP & 200, 100, 20, 80 (per km$^2$)  \\\hline
$N_{u},N_{d},$ $N_{s}$ & Number of UL UEs, DL UEs and SBSs. Add subscript $X$ for BS at $X$ & --\\\hline
$X^*$, $X^{**}$ & $X^*$ is BS serving UE at origin and $X^{**}$ is MBS serving $X^*\in\SBSs$&--\\\hline
$\power{m},\power{s},$ $\power{u}$ & Transmit powers & 30, 30, 20 dBm\cite{Ghosh14} \\\hline
$\Delta_m,\Delta_s,$ $\Delta_u$ & Half power beamwidth & $10^{o},10^{o},60^o$  \cite{Ghosh14,BaiHea14}
\\\hline
$G_{m},G_{s},$ $G_{u}$ & Main lobe gain & $24, 24, 6$ dB \cite{Akd14,Taori15,BaiHea14} \\\hline 
$g_{m},g_{s},$ $g_{u}$ & Side lobe gain & $-4, -4, -14$ dB \cite{BaiHea14}\\\hline
$B_\nu$, $\mathcal{A}_\nu$ & Association bias and probability towards BS of tier $\nu\in\{m,s\}$ & $B_s = B_m = 0$ dB \\\hline
$f_c$, $\res$  & Carrier frequency and bandwidth & 28 GHz, 200 MHz \\\hline
$\plos$, $\dlos$ & Blockage parameters & $0.3,\; 200$ m \cite{AndBai16} \\\hline
$\alpha_l,\alpha_n$ & LOS, NLOS path loss exponents & 2.1, 3.4\cite{Ghosh14}\\\hline
$\mathrm{C}_0$ & 1m reference distance omnidirectional path loss & $\left(3\times 10^8/4\pi f_c\right)^2$\\\hline
$\sigma^2$ & thermal noise (in dBm) & $-174 + 10\log_{10}(\res)  + 5$\\\hline
$\eta, \delta, \F$ & Fraction of DL UEs, fraction of access slots, frame size & 0.5, 0.5, 1\\\hline
$\ell,\mu$ & Access/backhaul or LOS/NLOS link & $\ell\in\{a,b\}$, $\mu\in\{l,n\}$\\\hline 
$t$ & Tier of BS PPP & $t\in\{m,s\}$\\\hline
$i$ & Slot index & $1\leq i\leq \F$\\\hline
$w_a, w_b$ & Resource allocation scheme in access and backhaul subframe &  $w_a\in\{S,D\}$, $w_b\in\{\mathrm{UAB},\mathrm{SAB}\}$ \\\hline
\end{tabulary}
\end{table}
\subsection{TDD frames and scheduling}
\label{sec:scheduling}
In the following discussion, UL denotes UE to BS links for access and SBS to MBS links for backhaul. Similarly, DL denotes the BS to UE links for access and MBS to SBS for backhaul.   

Fig.~\ref{fig:tddframe}(a) shows the TDD frame structure. Each frame consists of 4 subframes for DL access, UL access, DL backhaul, and UL backhaul. There are $\F_{ad}$, $\F_{au}$, $\F_{bd}$, $\F_{bu}$ slots, each of duration $\mathrm{T}$, in the 4 subframes. We denote by $\F_{a} = \F_{ad}+\F_{au}$, $\F_{b} = \F_{bd}+\F_{bu}$, and $\F  = \F_{a}+\F_{b}$.  We add a subscript $X$ to each of these to denote the sub-frame size for BS at $X\in\BSprocess$. The terminology $i^{\text{th}}$ slot would refer to the $i^{\text{th}}$ slot starting from the beginning of the TDD frame and $i$ varies from $1$ to $\F$. We neglect the slots allocated for control signals and subframe switching\cite{Ghosh16}, although this can be incorporated by scaling the rate estimates in this work by a constant factor. 

All BSs allocate $\delta$ fraction of $\F$ for access and rest for backhaul. If $\delta\F <1$ then in every time slot a coin is flipped with this probability to decide whether the slot is for access or backhaul, which is synchronously adopted by all BSs. Optimization over $\delta$ is  done numerically based on mean rate analysis in Section~\ref{sec:results}. Allowing different BSs to have a different $\delta$ is possible but for analytical tractability we do not consider such a scenario. Thus, $\F_a = \lceil\delta \F\rceil$ with probability $\delta \F-\lfloor\delta \F\rfloor$, and $\F_a = \lfloor\delta \F\rfloor$ otherwise.

Let $\gamma_{\ell,w,X}$ denote the fraction of slots allocated for DL transmissions in subframe of type $\ell\in\{a,b\}$ by BS at location $X$,  $w\in\{S,D\}$ denote static and dynamic TDD schemes when $\ell=a$, and $w\in\{\text{SAB},\text{UAB}\}$ denote synchronized and unsynchronized access-backhaul schemes when $\ell = b$. More on these schemes is discussed in the following text. The above notation implies that $\F_{ad,X} = \lceil \F_a \gamma_{a,w,X} \rceil$ with probability $\F_a \gamma_{a,w,X}-\lfloor\F_a \gamma_{a,w,X}\rfloor$, and $\F_{ad,X} = \lfloor\F_a \gamma_{a,w,X}\rfloor$ otherwise. Similarly for $\F_{bd,X}$ by replacing $\gamma_{a,w,X}$ with $\gamma_{b,w,X}$ and $\F_a$ with $\F - \F_a$.  
\subsubsection{Scheduling in access subframes}
We consider the following schemes for choosing $\gamma_{a,w,X}$. In each slot, a BS randomly schedules an UL/DL UE uniformly from the set of connected UEs. 
\begin{itemize}
\item {\bf Static TDD.} Here, $\gamma_{a,S,X} = \gamma_{a}$, which is a fixed constant independent of $X\in\BSprocess$. This can be a completely load unaware scheme if $\gamma_a$ is irrespective of $\eta$, and a load aware scheme if $\gamma_a$ is dependent on $\eta$. We focus on a load aware scheme wherein $\gamma_a = \eta$.
\item {\bf Dynamic TDD.} Now, we let $\gamma_{a,D,X}$ to be dependent on the BS location $X$ so that every BS can make their own choice of UL/DL time split in an access subframe. We focus on $\gamma_{a,D,X} =\mathds{1}(N_{d,X}>0) \frac{N_{d,X}}{N_{u,X}+N_{d,X}}$, where $N_{u,X}$ and $N_{d,X}$ are the number of UL and DL users connected to the BS at $X$. Several variations of this policy are possible, such as adding a different optimized exponent $n$ to $N_{u,X},N_{d,X}$ or incorporating other network parameters to capture the disparity of the UL/DL service rate. These variations are left to future work.
\end{itemize}
\subsubsection{Scheduling in backhaul subframes}
Like the access subframe, it is possible to have static and dynamic TDD schemes for deciding the fraction of DL slots in a backhaul subframe. However, for analytical simplicity we assume $\gamma_{b,w,X} = \eta$, which is fixed for all $X\in\BSprocess$. Hierarchical scheduling is assumed in the backhaul subframe. First the MBSs make a decision of scheduling available SBSs with at least one UL/DL UE in a UL/DL backhaul subframe with uniformly random SBS selection for each slot. A SBS has to adhere to the slots allocated by its serving MBS for backhauling. Let the set $\mathcal{F}$ represent sub-frame lengths that are fixed across all BSs irrespective of the scheduling strategies. $\F_a$ and $\F_{bd}$ are two permanent members of $\mathcal{F}$. Further, $\F_{ad}$ is also an element  of $\mathcal{F}$ under static TDD scheme. 
\begin{figure*}
\centering
\subfloat[A TDD Frame.]{\includegraphics[width = \columnwidth]{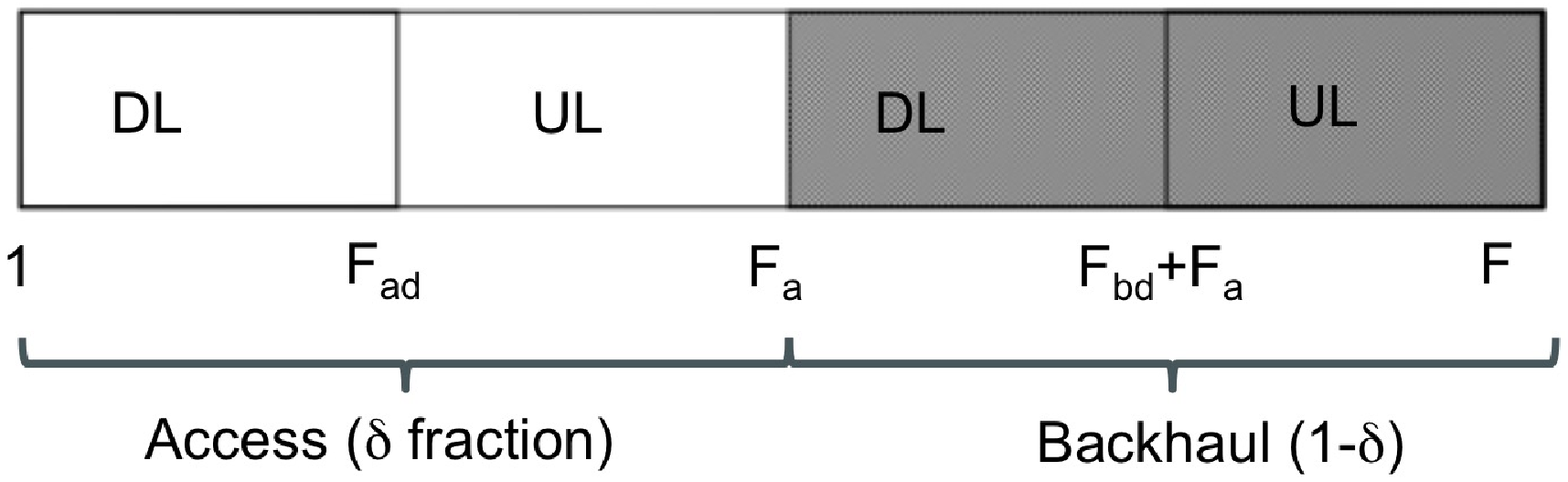}}\\
\subfloat[ Figure shows (i) Heirarchical scheduling in backhaul subframe with UAB or SAB. (ii) Dynamic TDD can lead to different DL subframe sizes in access subframe. ]{\includegraphics[width = 1.7\columnwidth]{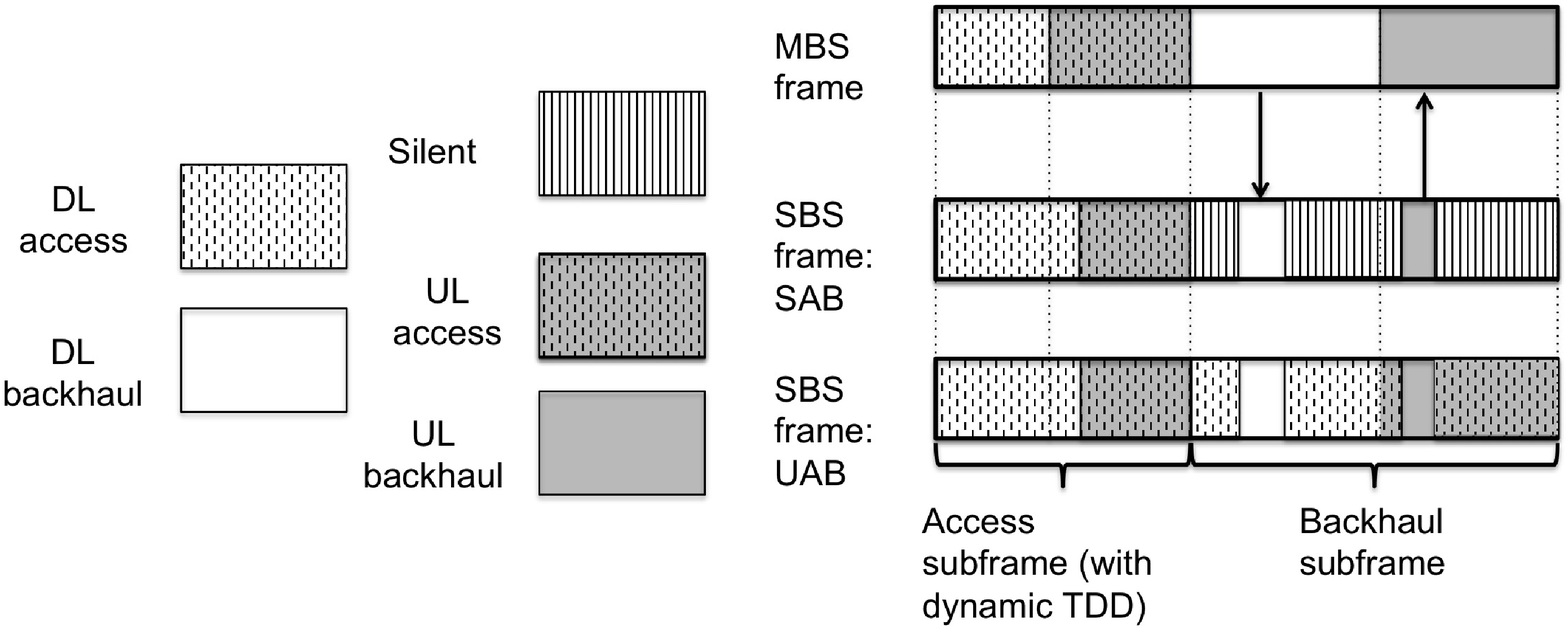}}
\caption{TDD frame structure.}
\label{fig:tddframe}
\end{figure*}
Although $\F_{bd}$ is fixed, a version of dynamic TDD is employed through UAB. 
\begin{itemize}
\item {\bf Synchronized access-backhaul (SAB).} SBS remains silent in unscheduled backhaul slots.
\item  {\bf Unsynchronized access-backhaul (UAB) or poaching.} SBS schedules an UL/DL access link in the unscheduled backhaul slots. We focus here on a simple policy wherein UL access {\em poaches}
only UL backhaul slots and similarly for DL. We assume that the SBS schedules an UL UE independently with probability $p_{ul}$ in an unscheduled backhaul UL slot and stays silent otherwise. $p_{dl}$ is the probability of scheduling a DL UE in a backhaul DL slot. 
\end{itemize}
\begin{rem}The analysis of in-band backhauling in this paper follows for out-of-band backhauling as well. In this case, a fraction $\delta$ of total bandwidth is allocated to access. 
\end{rem}
\subsection{Received signal power model}
The received signal at $X\in\BSprocess\cup \UEprocess$ from $Y\in\BSprocess\cup \UEprocess$ with $X\neq Y$ in the $i^\text{th}$ time slot of a typical TDD frame is given by
$P_r(X,Y) = \mathrm{C}_0 \mathrm{P}_Y h_{i,X,Y} G_{i,X,Y} L(X,Y)^{-1}$,
where $\mathrm{C}_0$ is the reference distance omnidirectional path loss at 1 meter, $\mathrm{P}_Y$ is the transmit power and is equal to either $\power{m}$, $\power{s}$ or $\power{u}$ depending on whether $Y\in\MBSs$, $Y\in\SBSs$ or $Y\in\UEprocess$. $h_{i,X,Y}$ is the small scale fading, $G_{i,X,Y}$ is the product of transmit and receive antenna gains and $L(X,Y) = ||X-Y||^{\alpha_{X,Y}}$ is path loss between $X$ and $Y$. Here, $\alpha_{X,Y}$ is $\alpha_l$ with probability $p_l(||X-Y||)$ and $\alpha_n$ otherwise. There are several models proposed for $p_l(d)$ to incorporate blockage effects\cite{BaiHea14,SinJSAC14,Renzo15,Akd14}. The generalized LOS ball model proposed in \cite{SinJSAC14} and validated in \cite{Kul14,AndBai16} is used in this work. As per this model, $p_l(d) = \plos$ if  $d\leq \dlos$ and $p_l(d) = 0$ otherwise. Let $p_n(d) = 1-p_l(d)$.

Here, $h_{i,X,Y}$ are independent and identically distributed (i.i.d.) to an exponential random variable with unit mean for all $X,Y\in\BSprocess\cup \UEprocess$. However, $h_{i,X,Y}$ can be arbitrarily correlated across time slots $i$. If the access link under consideration is a desired signal link, $G_{i,X,Y} = G_t G_u$, where $G_{(.)}$ denotes main lobe gain and $t\in\{m,s\}$. Similarly, $G_{i,X,Y} = G_m G_s$ for the backhaul desired signal link. An interfering link has antenna gain distribution as follows\cite{BaiHea14}, 
\begin{equation*}
G_{i,X,Y} \stackrel{d}{=} \begin{cases}
\Psi_{t_1,t_2} & \text{ if }X\in \Phi_{t_1}, Y\in\Phi_{t_2}\\&
 \text{ with }t_1, t_2\in\{m,s,u\} \text{ and }t_1\neq t_2, \\
\Psi_{t,t} & \text{ if }X, Y \in\Phi_t\text{ with }t\in\{m,s,u\} , \\
\end{cases}
\end{equation*}
where $\stackrel{d}{=}$ denotes equality in distribution. Further, $G_{i,X,Y}$ is independently distributed with $G_{i,X',Y'}$ if at least one of $X\neq X'$ or $Y\neq Y'$. Also these gains are independent of $h_{i,X,Y}$, $\forall X, Y\in \BSprocess\cup \UEprocess$. Here, the probability mass functions (PMF) of $\Psi_{t,t}$ and $\Psi_{t_1,t_2}$ are given in Table~\ref{tab:antenna}, $g_{(.)}$ and $\Delta_{(.)}$ represent the side-lobe gain and 3-dB beam width.
\begin{table}
\centering
\begin{tabular}{|c|c|c|}
\hline 
Parameter & Value & Probability\\ \hline
\multirow{4}{*}{$\Psi_{t_1,t_2}$} & $G_{t_1} G_{t_2} $ & $\frac{\Delta_{t_1} \Delta_{t_2}}{4\pi^2}$ \\ 
& $G_{t_1} g_{t_2}$ & $ \frac{\Delta_{t_1}(2\pi-\Delta_{t_2})}{4\pi^2}$\\ 
& $g_{t_1} G_{t_2}$ & $ \frac{(2\pi-\Delta_{t_1})\Delta_{t_2}}{4\pi^2}$ \\ 
& $g_{t_1} g_{t_2} $ & $\frac{(2\pi-\Delta_{t_1})(2\pi-\Delta_{t_2})}{4\pi^2}$ \\ 
\hline
\multirow{3}{*}{$\Psi_{t,t}$} & $G^2_t$ & $\frac{\Delta^2_t}{4\pi^2}$ \\ 
& $G_t g_t$ & $ \frac{2\Delta_t(2\pi-\Delta_t)}{4\pi^2}$\\ 
& $g^2_t$ & $ \frac{(2\pi-\Delta_t)^2}{4\pi^2}$ \\ \hline
\end{tabular}
\caption{Antenna gain distributions}
\label{tab:antenna}
\end{table}
\subsection{User and SBS association}
\label{sec:sysassoc}
Each user associates with either an MBS or SBS. Each SBS connects to an MBS. A typical user at $Z\in\UEprocess$ associates to BS at $X^*(Z)\in\BSprocess$ iff
$
X^*(Z) = \arg\max_{Y\in \Phi_t, t\in\{m,s\}} \power{t}L(Y,Z)^{-1} G_t B_t,
$
where $B_t$ denotes a bias value multiplied to the received signal power from a BS of tier $t\in\{m,s\}$. Since the association criterion maps every point in $\UEprocess$ to a unique point in $\BSprocess$ almost surely, the mean number of users connected to a typical MBS is $\lambda_u\mathcal{A}_m/\lambda_m$, and that to a typical SBS is $\lambda_u\mathcal{A}_s/\lambda_s$\cite{SinDhiAnd13,SinJSAC14}.  Here, $\mathcal{A}_m$ is the probability of associating with a MBS and $\mathcal{A}_s=  1-\mathcal{A}_m$. The derivation of $\mathcal{A}_m$ can be found in Appendix~\ref{sec:App_assocprob}. A SBS at $Z\in\SBSs$ connects to a MBS at $X^*(Z)\in\MBSs$ iff $X^*(Z) = \arg\min_{Y\in \Phi_m} L(Y,Z)$. 
Thus, the mean number of SBSs connected to a typical MBS is $\lambda_s/\lambda_m$. 

\subsection{Load distribution}
\label{sec:loaddist}
Characterizing the load distribution with PPP BSs and UEs even under the simplest setting of nearest BS association is a long-standing open problem\cite{Ferenc2007}. Several papers have assumed an independent load model for tractability \cite{SinDhiAnd13,ElSawy16,YuKim13,Renzo16,SinJSAC14}. Using a similar model, every $X\in\MBSs$ is associated with independent marks $N_{s,X},N_{u,X},N_{d,X}$ representing number of SBSs, UL UEs and DL UEs connected to the MBS. Similarly, every $X\in\SBSs$ is associated with independent marks $N_{u,X},N_{d,X}$. Their distributional assumptions are given as follows \cite{YuKim13,SinDhiAnd13}. 
\begin{assumption}
\label{assumption:1}
Let $\epsilon$ be the mean number of devices (users or SBSs) connected to a typical BS in $\Phi_t\in\{\MBSs,\SBSs\}$. The marginal probability mass function (PMF) of number of devices connected to a tagged and typical BS in $\Phi_b$ is given by $\kappa^*(n)$ and $\kappa(n)$ respectively. 
\begin{equation}
\label{eq:tagged}
\kappa^{*}(n) = \frac{3.5^{3.5}\Gamma(n+3.5)\epsilon^{n-1} \left(3.5+\epsilon\right)^{-n-3.5}}{(n-1)!\Gamma(3.5)}, \text{    for $n\geq 1$}
\end{equation}
\begin{equation}
\kappa(n) = \frac{3.5^{3.5}\Gamma(n+3.5)\epsilon^{n} \left(3.5+\epsilon\right)^{-n-3.5}}{n!\Gamma(3.5)}, \text{     for $n\geq 0$.}
\end{equation}
Thus, the marginal PMF of $N_{s,X},N_{u,X},N_{d,X}$ is denoted as $\kappa_{s,t},\kappa_{u,t},\kappa_{d,t}$ for typical BS $X\in\Phi_t$ and with a superscript $*$ for tagged BS $X$. $\epsilon$ for each of these is given by $\frac{\lambda_s}{\lambda_m}$, $\frac{(1-\eta)\lambda_u\mathcal{A}_t}{\lambda_t}$ and $\frac{\eta\lambda_u\mathcal{A}_t}{\lambda_t}$, respectively.
\end{assumption}

\begin{assumption}
\label{assumption:2}
Let $\epsilon = \lambda_u\mathcal{A}_t/\lambda_t$ be the mean number of users connected to a typical BS in $\Phi_t\in\{\MBSs,\SBSs\}$. The joint PMF of number of UL and DL users connected to a typical BS in $\Phi_t$ is given by $\Upsilon_t(n_1,n_2,3.5)$ for $n_1,n_2\geq 0$, where 
$$\Upsilon_t(n_1,n_2,k)=\frac{3.5^{3.5}}{\Gamma(3.5)}\frac{\eta^{n_2} (1-\eta)^{n_1}}{n_1! n_2!}\frac{\Gamma(n_1+n_2+k)}{\epsilon^{k} \left(1+\frac{3.5}{\epsilon}\right)^{n_1+n_2+k}}.$$
Consider a BS serving the user at origin, then the joint PMF of number of UL and DL users connected to the BS {\em apart from the user at origin} is given by $\Upsilon_t(n_1,n_2,4.5)$ for $n_1,n_2\geq 0$. 
\end{assumption}
A summary of key notation is given in Table~\ref{tab:notation} and Fig.~\ref{fig:tddframe}.

\section{Uplink SINR and rate} 
\label{sec:ULanalysis}
As shown by Fig.~\ref{fig:tddframe}, the SINR distribution will be dependent on the time slot $1\leq i\leq \F$ and the scheduling strategies. Our goal is to compute the mean end-to-end rate of a typical user (UL or DL) at the origin under the various scheduling strategies described before. We analyze the marginal SINR distribution for access and backhaul links as two separate cases. Before going into the details, we first characterize the PMF of the number of DL access slots as follows.
\begin{lem}
\label{lem:Fad}
The PMF of $\F_{ad,w,X}$ $\stackrel{d}{=}\F_{ad,w}$, for  a typical $X\in\Phi_t$ given $\mathcal{F}$ is computed as follows. 
\begin{enumerate}
\item For static TDD, that is $w=S$,
\begin{multline}
\label{eq:FadS}
\mathbb{P}\left(\F_{ad,S,X} = n\big| \F_a\right) = \tilde{\F}_{ad} \mathds{1}\left(\lceil \gamma_{a} \F_a\rceil = n\right)\\+ (1-\tilde{\F}_{ad}) \mathds{1}\left(\lfloor \gamma_{a} \F_a\rfloor=n\right),
\end{multline} 
where $\tilde{\F}_{ad} = \gamma_{a} \F_a- \lfloor \gamma_{a} \F_a\rfloor$.
\item For dynamic TDD, that is $w=D$,
\begin{multline}
\label{eq:FadD}
\mathbb{P}\left(\F_{ad,D,X} = n\big| \F_a\right) \\= \int_{0}^{1} (p_1(n+r-1)-p_2(n+1-r)) \mathrm{d}r,
\end{multline}
where
\begin{multline*}
p_1(r) = \mathds{1}(r> 0)\sum_{n_2=1}^{\infty}\sum_{n_1=0}^{\lceil\frac{n_2 (\F_a-r)}{r}\rceil-1} \Upsilon_t(n_1,n_2,3.5)+\\\mathds{1}(r\leq 0) - \mathds{1}(r = 0)\left(1+\frac{\mathcal{A}_t \lambda_u\eta }{3.5\lambda_t}\right)^{-3.5},
\end{multline*}
\begin{multline*}
p_2(r) = 
\mathds{1}(r> 0) \sum_{n_2=1}^{\infty}\sum_{n_1=0}^{\lfloor\frac{n_2 (\F_a-r)}{r}\rfloor} \Upsilon_t(n_1,n_2,3.5)\\+
\mathds{1}(r\leq 0).
\end{multline*}
\end{enumerate}
\end{lem}
\begin{proof}
See Appendix~\ref{appendix:subframe}.
\end{proof}
Small tail probabilities of the PMFs in Assumptions~\ref{assumption:1} and \ref{assumption:2} for load values larger than the $\sim6\times$ the mean allows us to compute the infinite sums as finite sums with first $\lfloor\frac{6\mathcal{A}_t\lambda_u}{\lambda_t}\rfloor$ terms. 

\subsection{SINR model for access links}
Access links can be active in both access and backhaul subframes if the BSs operate in UAB. The SINR of a receiving BS at $X^*\in\Phi_t$, where $t\in\{m,s\}$, serving the UL user at origin is given as
$$
\SINR^{ul}_{i,a,w} = \frac{\mathrm{C}_0 \power{u} h_{i,X^*,0} G_u G_t L(X^*,0)^{-1}}{I_{i,m,w}(X^*) + I_{i,s,w}(X^*) + I_{i,u,w}(X^*) + \sigma^2},
$$
where $w\in\{S,D\}$ denotes static and dynamic TDD if $i\leq \mathrm{F}_a$ and $w\in\{\text{SAB}, \text{UAB}\}$ if $i>\mathrm{F}_a$. $I_{i,\nu,w}(Z)$ is the interference power at location $Z\in\BSprocess\cup\UEprocess$ from all active devices of type $\nu\in\{m,s,u\}$ in the $i^\text{th}$ slot and $\sigma^2$ is the noise power. Here, for $\nu \in\{m,s\}$ and $i\leq \F_a$
\begin{multline}
\label{eq:Inu_access}
I_{i,\nu,w}(Z) = \sum_{Y\in\Phi_\nu\backslash\{X^*\}} \mathds{1}(i\leq \F_{ad,w,Y}) \mathds{1}(N_{d,Y}>0) \mathrm{C}_0 \power{\nu} \\\times h_{i,Z,Y} G_{i,Z,Y} L(Z,Y)^{-1}.
\end{multline}
Note that $\Phi_\nu \backslash\{X^*\} = \Phi_\nu$ if $X^*\notin \Phi_\nu$. Similarly, for $i\leq \F_a$
\begin{multline}
\label{eq:Iu_access}
I_{i,u,w}(Z) = \sum_{Y\in \BSprocess\backslash\{X^*\}} \mathds{1}(\F_{ad,w,Y}<i\leq \F_{a}) \mathds{1}(N_{u,Y}>0) \\\times \mathrm{C}_0 \power{u}h_{i,Z,Y'} G_{i,Z,Y'} L(Z,Y')^{-1},
\end{multline}
where $Y'$ is the UL UE scheduled by BS at $Y$. If $i>\F_a$, then 
\begin{multline}
\label{eq:Im_backhaul}
I_{i,m,w}(Z) =\sum_{Y\in\MBSs\backslash\{X^{**}\}} \mathds{1}(\F_a <i\leq \F_{a}+\F_{bd})   \\\times \mathds{1}\left(N_{s,d,Y}>0\right)\mathrm{C}_0\power{m}h_{i,Z,Y} G_{i,Z,Y} L(Z,Y)^{-1},
\end{multline}
where $X^{**}$ is the location of MBS serving $X^*\in\SBSs$ and $N_{s,d,Y}$ is the number of SBS with atleast one DL UE. Similarly, if $N_{s,u,Y}$ is the number of SBS connected to $Y\in\MBSs$ with at least one UL UE,
\begin{multline}
\label{eq:Is_backhaul}
I_{i,s,w}(Z) = \sum_{Y\in\MBSs} \mathds{1}(\F_{a}+\F_{bd}<i\leq \F) \mathds{1}\left(N_{s,u,Y}>0\right)  \\\times \mathrm{C}_0\power{s} h_{i,Z,Y'} G_{i,Z,Y'} L(Z,Y')^{-1} +\mathds{1}(w=\text{UAB})
 \\\times\sum_{Y\in\SBSs\backslash\{X^*\}}\mathds{1}(\F_a <i\leq \F_{a}+\F_{bd}) \mathds{1}\left(N_{d,Y}>0\right) \\\times \xi_{Y} \zeta_{Y} \mathrm{C}_0 \power{s} h_{i,Z,Y} G_{i,Z,Y} L(Z,Y)^{-1},
\end{multline}
where $Y'$ is the SBS scheduled by MBS at $Y\in\Phi_m$. Here, $\zeta_{Y}$ is a Bernoulli random variable (independent across all $Y$) with success probability $p_{dl}\mathds{1}\left(\F_a <i\leq \F_{a}+\F_{bd}\right)+ p_{ul}\mathds{1}\left(\F_a+\F_{bd}< i\leq \F\right)$ and $\xi_{Y}$ is also an indicator random variable denoting whether the SBS is not scheduled by its serving MBS for backhauling in slot $i$ of the typical frame under consideration. Also, 
\begin{multline}
\label{eq:Iu_backhaul}
I_{i,u,w}(Z) = \mathds{1}(w=\text{UAB})\sum_{Y\in \SBSs \backslash\{X^*\}} \mathds{1}(\F_a+\F_{bd}< i\leq \F)\\ \mathds{1}(N_{u,Y}>0) \xi_{Y} \zeta_{Y} \mathrm{C}_0 \power{u} h_{i,Z,Y'} G_{i,Z,Y'} L(Z,Y')^{-1},
\end{multline}
where $Y'\in\ULUEprocess$ is the UL user scheduled by the BS at $Y$.

\Cref{eq:Inu_access,eq:Iu_access,eq:Im_backhaul,eq:Is_backhaul,eq:Iu_backhaul} are applicable for evaluating the UL backhaul, DL access, and DL backhaul SINR distribution as well, although the receiving location $Z$ will be different under each case and is summarized in Table~\ref{tab:links}. Note that an UL access link will be active in a backhaul subframe only in $\F_a+ \F_{bd}\leq i\leq \F$ and $w=$UAB scenario. Thus, to compute UL access SINR, \eqref{eq:Is_backhaul} would have only the first summation term, and \eqref{eq:Im_backhaul} would be zero.  

\begin{rem}[A note on the interfering point processes in \eqref{eq:Inu_access} to \eqref{eq:Iu_backhaul}]
\label{rem:int}
Computing the Laplace transform of interference is a key step in evaluating SINR distribution. Exact expressions are available in literature for interferers generated from a PPP, Poisson cluster process, some special repulsive point processes \cite{andganbac11,Suryaprakash15,Li15}. Note that \eqref{eq:Inu_access} and \eqref{eq:Im_backhaul} have PPP interferers, and thus computing exact Laplace transform is possible. However, \eqref{eq:Iu_access}, \eqref{eq:Is_backhaul} and \eqref{eq:Iu_backhaul} have non-Poisson interfering processes, for which it is highly non-trivial to characterize the Laplace transform. Several approximate PPP models have been proposed in literature for computing Laplace functional of the interfering point processes in \eqref{eq:Iu_access} and first term in \eqref{eq:Is_backhaul}, for example \cite{SinZhaAnd15,ElSawy14,Lee14,Hae16}. We follow a theme of PPP approximations for the same inspired from these works. To compute an approximate Laplace transform of \eqref{eq:Iu_backhaul} and second term in \eqref{eq:Is_backhaul} we propose novel PPP approximations on the same lines as \cite{SinZhaAnd15} and validate these approximations with Monte-Carlo simulations.
\end{rem}

\begin{table}
\centering
\begin{tabular}{|c|c|c|}
\hline
\bf Link & \bf Receiver & \bf Transmitter\\\hline
UL access & $X^*$ & 0\\\hline
UL backhaul & $X^{**}$ & $X^*$\\\hline
DL access & 0 & $X^*$\\\hline
DL backhaul & $X^*$ & $X^{**}$\\
\hline
\end{tabular}
\caption{Transmitter-receiver pairs for computing end-to-end rate of a typical user at origin.}
\label{tab:links}
\end{table}

\subsection{SINR distribution for access links}
\begin{definition}
Conditioned on $\mathcal{F}$, the SINR coverage of a typical UL access link is defined as $\mathsf{S}^{ul}_{i,a,w}(\tau) = \mathbb{P}\left(\SINR^{ul}_{i,a,w}>\tau\big| \mathcal{F}\right)$, if slot $i\leq \F_a$ and $w\in\{S,D\}$. If $i>\F_a$, typical UL UE is scheduled only if  $w=\mathrm{UAB}$ and it connects to a SBS. Thus, the SINR coverage for $i>\F_a$ is given by $\mathsf{S}^{ul,t}_{i,a,w}(\tau) = \mathbb{P}\left(\SINR^{ul}_{i,a,w}>\tau \big| X^*\in\Phi_{t},\mathcal{F}\right)$ for $t=s$.
\end{definition} 
\begin{definition}
The Laplace transform of the interference at a typical UL access receiver at $X^*$ conditioned on the event that the receiving BS is at a distance $R$ and belongs to $\Phi_{t,\mu}$, which is the point process of LOS/NLOS BSs in $\Phi_t$ looking from origin, is given as follows for $\mu\in\{l,n\}$.
$$L^{ul,a,t,\mu}_{i,w}(\mathsf{s},R) =\mathbb{E}\left[\exp\left(-\mathsf{s} I \right)\big| X^*\in\Phi_{t,\mu},||X^*||=R,\mathcal{F}\right],$$ where $I = I_{i,m,w}(X^*)+I_{i,s,w}(X^*)+I_{i,u,w}(X^*)$. 
\end{definition}
\begin{lem}
\label{lem:LaplaceULaccess}
For $i\leq \F_a$, the Laplace transform $L^{ul,a,t,\mu}_{i,w}(\mathsf{s},R)\approx L_m L_s L_u$, where 
\begin{itemize}
\item For $\nu\in\{m,s\}$, $L_\nu  = 1$ if $w = S$ and is given as follows if $w = D$, $$L_\nu \geq \exp\left(-\int_{0}^{\infty}\mathbb{E}\left[\frac{1}{1+\frac{r}{\mathsf{s}\mathrm{C}_0\power{\nu}\Psi_{t,\nu}}}\right]p_{i,D,\nu}\Lambda_\nu(\mathrm{d}r)\right).$$ Exact expression for $L_\nu$ is given in \eqref{eq:longeqLnu}.
\begin{figure*}[!t]
\begin{align}
\label{eq:longeqLnu}
L_\nu  &= \prod_{\mu_1,\mu_2\in\{l,n\}} \exp\left(- \int^{\infty}_{\left(\frac{R^{\alpha_\mu}\power{\nu} G_{\nu} B_{\nu}}{\power{t} G_t B_t}\right)^{1/\alpha_{\mu_1}}}\int^{2\pi}_{0}\mathbb{E}\left[\frac{p_{i,D,\nu} \hat{\lambda}_{\nu,\mu_1,\mu_2}(r,\theta) r}{1+\frac{\left(r^2+R^2-2rR\cos(\theta)\right)^{\alpha_{\mu_2}/2}}{\s \mathrm{C}_0\power{\nu} \Psi_{t, \nu}}}\right]  \mathrm{d}r\mathrm{d}\theta\right).
\end{align}
\hrulefill
\end{figure*}
where  $\hat{\lambda}_{\nu,\mu_1,\mu_2}(r,\theta)$ is equal to $$\lambda_\nu p_{\mu_1}(r)p_{\mu_2}\left(\sqrt{r^2+R^2-2rR\cos\theta}\right),$$
$p_{i,D,\nu} = \sum\limits_{n=i}^{\F_a}\mathbb{P}\left(\F_{ad,D}=n\big| \mathcal{F}\right)$, and $\Lambda_\nu(\mathrm{d}\tau)$ is given in \eqref{eq:LambdadT}. The expectation is with respect to the antenna gains $\Psi_{(.)}$ given in Table~\ref{tab:antenna}.
\item For $w\in\{S,D\}$,
$$
L_u  = \exp\left(-\int_{0}^{\infty}\mathbb{E}\left[\frac{1}{1+ \frac{r}{\s\mathrm{C}_0\power{u}\Psi_{t,u}}}\right]\Lambda(t,\mathrm{d}r)\right),
$$
where the expectation is with respect to the antenna gains $\Psi_{(.)}$ given in Table~\ref{tab:antenna}, $\Lambda(t,\mathrm{d}r) = \sum_{k\in\{m,s\}}p_{i,w,k}\times$
\begin{equation*}
 \left(1-\exp\left(-\Lambda_k\left(r\frac{\power{k}B_k G_k}{\power{t}B_t G_t}\right)\right)\right) \Lambda_k(\mathrm{d}r),
\end{equation*}
with $\Lambda_k(r)$ given in \eqref{eq:fcap_s}. 
\begin{multline*}
p_{i,S,k} = \left(1-\left(1+\frac{\lambda_u\mathcal{A}_k (1-\eta)}{3.5\lambda_k}\right)^{-3.5} \right)\\\times\mathds{1}\left(\F_{ad}<i\leq \F_a \right),
\end{multline*} 
and 
\begin{multline*}
p_{i,D,k} = \mathbb{P}\left(\F_{ad,D}<i\leq \F_a\big| \mathcal{F}\right) \\- \left(1+\frac{\lambda_u\mathcal{A}_k (1-\eta)}{3.5\lambda_k}\right)^{-3.5},
\end{multline*}
which is computed using distribution of $\F_{ad,D}$ given in Lemma~\ref{lem:Fad}. 
\end{itemize}
\end{lem}
\begin{lem}
\label{lem:LaplaceULaccess2}
For $i>\F_a$ and $w = \mathrm{UAB}$, the Laplace transform $L^{ul,a,t,\mu}_{i,w}(\mathsf{s},R)\approx L_s L_u$, where 
\begin{multline*}
L_s =  \exp\left(-\int_{0}^{\infty}\mathbb{E}\left[\frac{1}{1+ \frac{r}{\s\mathrm{C}_0\power{s}\Psi_{t,s}}}\right]\times\right. \\\Bigg.\left( p_{void}\left(1-\exp\left(-\Lambda_m(r)\right)\right)+\exp\left(-\Lambda_m(r)\right)\right)\Lambda_m(\mathrm{d}r)\Bigg),
\end{multline*}
\begin{multline*}
L_u  = \exp\left(-\int_{0}^{\infty}\mathbb{E}\left[\frac{1}{1+ \frac{r}{\s\mathrm{C}_0\power{u}\Psi_{t,u}}}\right]\frac{\hat{\lambda}}{\lambda_s}\times\right. \\\Bigg.\left(1-\exp\left(-\Lambda_s(r)\right)\right)\Lambda_s(\mathrm{d}r)\Bigg).
\end{multline*}
Here, the expectation is with respect to the antenna gains $\Psi_{(.)}$ given in Table~\ref{tab:antenna}, $$p_{void} = 1-\left(1+\frac{\lambda_{s,u}}{3.5\lambda_m}\right)^{-3.5},$$ 
with $$\lambda_{s,u} = \lambda_s\left(1-\left(1+\frac{\mathcal{A}_s \lambda_u (1-\eta)}{3.5\lambda_s}\right)^{-3.5} \right),$$
\begin{multline*}
\hat{\lambda} = p_{ul}\left(\lambda_s-\left(1-\left(1+\frac{\lambda_s }{3.5\lambda_m}\right)^{-3.5} \right)\lambda_m\right)^+ \times\\
\mathds{1}\left(\F_a+\F_{bd}<i\leq \F\right)\left(1-\left(1+\frac{\lambda_u (1-\eta)\mathcal{A}_s
}{3.5\lambda_s}\right)^{-3.5} \right).
\end{multline*}
\end{lem}
{\em Proof.} See Appendix~\ref{sec:App_UL_Laplace} for proofs of Lemma~\ref{lem:LaplaceULaccess} and Lemma~\ref{lem:LaplaceULaccess2}.
\begin{thm}
\label{thm:ULSINRaccess}
For $i\leq \F_a$, the SINR coverage of a typical UL user is given by $\mathbb{E}\left[\mathsf{S}^{ul}_{i,a,w}(\tau)\right]$ where the expectation is over $\mathcal{F}$. For $i>\F_a$ and $w=\mathrm{UAB}$, the SINR coverage is given by $\mathbb{E}\left[\mathsf{S}^{ul,s}_{i,a,\mathrm{UAB}}(\tau)\right]$.  Here,  $\mathsf{S}^{ul}_{i,a,w}(\tau) =\mathcal{A}_s \mathsf{S}^{ul,s}_{i,a,w}(\tau)+\mathcal{A}_s \mathsf{S}^{ul,s}_{i,a,w}(\tau)$, where $\mathsf{S}^{ul,t}_{i,a,w}(\tau)=$
\begin{multline}
\label{eq:ULSINRdist_access}
\sum_{\mu\in\{l,n\}}\int\limits_{0}^{\infty}\exp \left( \frac{-\tau R^{\alpha_\mu} \sigma^2}{\mathrm{C}_0 \power{u} G_u G_t}\right) L^{ul,a,t,\mu}_{i,w}\left(\frac{\tau R^{\alpha_\mu}}{\mathrm{C}_0 \power{u} G_u G_t }, R\right) \\
\times\prod_{\substack{t'\in\{m,s\},\\\mu'\in\{l,n\},\\ t'\neq t \text{ or }
  \mu'\neq \mu}} F_{t',\mu'}\left(\left(\frac{P_{t'} G_{t'} B_{t'} R^{\alpha_{\mu}}}{P_t B_t G_t}\right)^{\frac{1}{\alpha_{\mu'}}}\right) \frac{f_{t,\mu}(R)}{\mathcal{A}_t}  \mathrm{d}R, 
\end{multline}
where $L^{ul,a,t,\mu}_{i,w}(.)$ is given in Lemma~\ref{lem:LaplaceULaccess} and \ref{lem:LaplaceULaccess2},
\begin{align*}
F_{t,n}(R) &= \mathrm{exp}\left({-\pi\lambda_t \left(R^2 - p_{\mathrm{LOS}}\min(R,\dlos)^2\right)}\right),\\
F_{t,l}(R) &= \mathrm{exp}\left({-\pi\lambda_t p_{\mathrm{LOS}}\min(R,\dlos)^2}\right),\\
f_{t,l}(R) &= 
2\pi\lambda_t R p_{\mathrm{LOS}}\mathds{1}(R\leq\dlos)\\&\hspace{0.5cm}\times\exp\left(-\pi\lambda_t p_\mathrm{LOS} \min(R,\dlos)^2\right), \\
f_{t,n}(R) &= 2\pi\lambda_t R \left(1-p_{\mathrm{LOS}}\mathds{1}(R\leq\dlos)\right)\\&\hspace{0.5cm}\times\exp\left(-\pi\lambda_t \left(R^2-p_\mathrm{LOS} \min\left(R, \dlos\right)^2\right)\right).
\end{align*}
\end{thm}
\begin{proof}
The $\SINR$ coverage of a typical UL user scheduled in the $i^\text{th}$ slot ($i\leq \F_a$), is given by \footnote{Note that conditioning on $\mathcal{F}$ is not explicitly written in the following equations for convenience. }
\begin{align*}
&\mathsf{S}^{ul}_{i,a,w}(\tau) = \mathbb{P}\left(\SINR^{ul}_{i,a,w}>\tau\right)\\&=\sum_{t\in\{s,m\},\, \mu\in\{l,n\}}\mathbb{P}\left(\SINR^{ul}_{i,a,w}>\tau , X^*\in\Phi_{t,\mu}\right)\\ 
&=\sum_{t\in\{s,m\},\, \mu\in\{l,n\}}\int_{0}^{\infty}\mathbb{P}\left(\SINR^{ul}_{i,a,w}>\tau, \right.\\& \left.X^*\in\Phi_{t,\mu}\Big| ||X^*_{t,\mu}||=R\right)f_{t,\mu}(R)\mathrm{d}R ,
\end{align*}
where $f_{t,\mu}(R)$ is the probability that there exists $X^*_{t,\mu}$, which is the BS nearest to origin of tier $t$ and link type $\mu\in\{l,n\}$, and its distance from origin is $R$. It is given as 
\begin{align*}
f_{t,\mu}(R)&= -\frac{\mathrm{d}}{\mathrm{d}R}\mathbb{P}\left(\Phi_{t,\mu}\left(\mathcal{B}(0,R)=0\right),  \Phi_{t,\mu}\left(\mathbb{R}^2>0\right)\right) \\
&= 2\pi\lambda_t R p_{\mu}(R)\exp\left(-2\pi\lambda_t\int_{0}^{R}p_{\mu}(r)r\mathrm{d}r\right).
\end{align*}
The SINR coverage expression is simplified further as shown in \eqref{eq:SINRproof}. 
\begin{figure*}[!t]
\begin{align}
\label{eq:SINRproof}
\mathsf{S}^{ul}_{i,a,w}(\tau)&=\sum_{t\in\{s,m\},\, \mu\in\{l,n\}}\int_{0}^{\infty}\mathbb{P}\left(\SINR^{ul}_{i,a,w}>\tau \Big| X^*\in\Phi_{t,\mu}, ||X^*||=R\right)
\mathbb{P}\left(X^*\in\Phi_{t,\mu}\big| ||X^*||=R\right)  f_{t,\mu}(R)\mathrm{d}R\nonumber
\\
&= \sum_{\substack{t\in\{s,m\},\\ \mu\in\{l,n\}}}\int_{0}^{\infty}\mathbb{E}\left[\exp \left( \frac{-\tau R^{\alpha_\mu} (I_{i,m,w}(X^*) + I_{i,s,w}(X^*) + I_{i,u,w}(X^*) + \sigma^2)}{\mathrm{C}_0 \power{u} G_u G_t }\right)\Big| X^*\in\Phi_{t,\mu}, ||X^*||=R\right]\nonumber\\&\times \prod_{\substack{t'\in\{m,s\},\mu'\in\{l,n\},t'\neq t \text{ or }
  \mu'\neq \mu}} F_{t',\mu'}\left(\left(\frac{P_{t'} G_{t'} B_{t'} R^{\alpha_{\mu}}}{P_t B_t G_t}\right)^{1/\alpha_{\mu'}}\right) f_{t,\mu}(R)
   \mathrm{d}R.
\end{align}
\hrulefill
\end{figure*}
where $F_{t,\mu}(R) = \mathbb{P}\left(\Phi_{t,\mu}\left(\mathcal{B}(0,R)=0\right)\right)$ and
$\mathcal{B}(0,R)$ is the ball of radius R centered at the origin.   
\end{proof}
\subsection{SINR distribution for backhaul links}
SINR model for backhaul links is given by 
$$
\SINR^{ul}_{i,b,w} = \frac{\mathrm{C}_0 \power{s} h_{i,X^{*},X^{**}} G_m G_s L(X^{*},X^{**})^{-1}}{I_{i,s,w}(X^{**}) + I_{i,u,w}(X^{**}) + \sigma^2},
$$
where $w\in\{\text{UAB},\text{SAB}\}$, $\F_a+\F_{bd}<i\leq \F$, $I_{i,u,\text{SAB}}=0$.  $I_{i,s,w}(.)$ and  $I_{i,u,\text{UAB}}(.)$ are same as \eqref{eq:Is_backhaul} and \eqref{eq:Iu_backhaul}, respectively, except that here the receiver is $X^{**}$, which is the MBS serving the tagged SBS. 

For the backhaul links, we are interested to find $\mathbb{P}\left(\SINR^{ul}_{i,b,w}>\tau\big| X^*\in\SBSs\right)$ where the probability is under the Palm of the user process. The reason is that for computing the end-to-end rate of a typical user at origin, we are interested in the distribution of backhaul SINR distribution only in scenarios when the user at origin connects to a SBS. However, to compute even serving distance distribution of backhaul link under the Palm of user process is highly non-trivial. In \cite{Sharma16}, such distribution was computed in the case when there were no blockage effects. Although in principle, such computations can be done with blockage effects there will be a total 12 cases that will arise -- condition LOS/NLOS links for typical UE at origin to $X^*$ and the backhaul links between $X^*$ and $X^{**}$, and 3 sub-cases for each of these that account for different exclusion regions as shown in \cite{Sharma16}. Computing the SINR CCDF under the Palm of the SBS process is much easier as follows and we will approximate the SINR coverage of a typical backhaul link to be equal to that of the tagged link for rate computations, as also done previously in \cite{SinJSAC14,Lu15,Tabassum16}. Validation of this is done in Figure~\ref{fig:plot6}. Similar to UL access, the following can be derived. 
\begin{cor}
\label{cor:ULSINRdist_backhaul}
CCDF of a typical backhaul UL SINR link for $i>\F_a$ is given as
\begin{multline*}
\mathsf{S}^{ul}_{i,b,w}(\tau) = 
\sum_{\mu\in\{l,n\}}\int_{0}^{\infty}\exp \left( \frac{-\tau R^{\alpha_\mu} \sigma^2_n}{\mathrm{C}_0 \power{s} G_s G_m}\right)\times\\ L^{ul,b}_{i,w}\left(\frac{\tau R^{\alpha_\mu}}{\mathrm{C}_0 \power{s} G_s G_m }\right) F_{m,\mu'}\left(R^{\alpha_{\mu}/\alpha_{\mu'}}\right)  f_{m,\mu}(R) \mathrm{d}R
\end{multline*}
where $L^{ul,b}_{i,w}(\mathsf{s}) = \mathbb{E}\left[\exp\left(-\mathsf{s}(I_{i,s,w}(X^{**})+I_{i,u,w}(X^{**}))\right)\right]\approx L_s L_u $ with
\begin{multline*}
L_s =  \exp\left(-\int_{0}^{\infty}\mathbb{E}\left[\frac{1}{1+ \frac{r}{\s\mathrm{C}_0\power{s}\Psi_{m,s}}}\right]\times\right.\\\left. \left(1-\left(1+\frac{\lambda_{s,u}}{3.5\lambda_m}\right)^{-3.5}\right)\left(1-\exp\left(-\Lambda_m(r)\right)\right)\Lambda_m(\mathrm{d}r)\right),
\end{multline*}
where $\lambda_{s,u} = \lambda_s\left(1-\left(1+\frac{\lambda_u (1-\eta)\mathcal{A}_s
}{3.5\lambda_s}\right)^{-3.5} \right)$.
\begin{multline*}
L_u =  \mathds{1}(w = \mathrm{SAB})+
\mathds{1}(w = \mathrm{UAB})\times\\ \exp\left(-\int_{0}^{\infty}\mathbb{E}\left[\frac{\left(1-\exp\left(-\Lambda_s(r)\right)\right)\frac{\bar{\lambda}_u}{\lambda_s}\Lambda_s(\mathrm{d}r)}{1+ \frac{r}{\s\mathrm{C}_0\power{u}\Psi_{m,u}}}\right] \right),
\end{multline*}
with \begin{multline*}
\bar{\lambda}_u = p_{ul}\left(1-\left(1+\frac{\lambda_u (1-\eta)\mathcal{A}_s
}{3.5\lambda_s}\right)^{-3.5} \right)\times\\ \left(\lambda_s-\left(1-\left(1+\frac{\lambda_s }{3.5\lambda_m}\right)^{-3.5} \right)\lambda_m\right)^+ .
\end{multline*}
The expectation in the expressions for $L_u$ and $L_s$ is with respect to the antenna gains $\Psi_{(.)}$ given in Table~\ref{tab:antenna}.
\end{cor}

\subsection{Mean rate analysis}
Let $\mathcal{E}_m$ and $\mathcal{E}_s$ denote the events when the typical UE connects to a MBS and SBS, respectively. 

\noindent{\bf Typical UE connected to MBS.} Data transmitted by a typical UL user in a frame is given by 
$\mathrm{D}_{ul,m,w_a}=\BW \mathrm{T}\times$ $$  \sum_{i=1+\F_{ad}}^{\F_{a}}\mathds{1}\left(\text{UE scheduled in $i^\text{th}$ slot}\right)\log_2\left(1+\mathsf{SINR}^{ul}_{i,a,w_a}\right).$$
Here, $w_a\in\{S,D\}$ representing static and dynamic TDD. As time progresses, in every frame the data transmitted by the UL UE is distributed according to the above equation. Thus, the data rate of the user averaged over time is given by $\mathbb{E}\left[\mathrm{D}_{ul,m,w_a}\Big| \BSprocess, \UEprocess, \mathcal{E}_m\right]/\mathrm{T}\F$, where expectation is over temporally varying random variables (all the randomness except that from $\BSprocess$ and $\UEprocess$).  Spatial averaging over the user and BS point processes gives data rate of the typical user at origin as 
$
\mathrm{R}_{ul,m,w_a} =\frac{\mathbb{E}\left[\mathrm{D}_{ul,m,w_a}\big| \mathcal{E}_m\right]}{\mathrm{T}\F}.
$

\noindent{\bf Typical UE connected to SBS.} Data transmitted by a typical UL user in access and backhaul slots of a typical frame is given by $\mathrm{D}_{ul,s,a,w_a}$ and $\mathrm{D}_{ul,s,b,w_b}$, respectively, in \eqref{eq:longeqRate1}. 
\begin{figure*}[!t]
\begin{align}
\label{eq:longeqRate1}
&\mathrm{D}_{ul,s,a,w_a} = \BW\mathrm{T} \sum_{i=1+\F_{ad}}^{\F}\mathds{1}\left(\text{UE scheduled in $i^\text{th}$ slot}\right)\log_2\left(1+\mathsf{SINR}^{ul}_{i,a,w_a}\right),\nonumber\\
&\mathrm{D}_{ul,s,b,w_b} = \BW\mathrm{T} \sum_{i=1+\F_{a}+\F_{bd}}^{\F}\log_2\left(1+\mathsf{SINR}^{ul}_{i,b,w_b}\right)\mathds{1}\left(\text{tagged SBS scheduled in $i^\text{th}$ slot and tx the UE's data}\right) .
\end{align}
\hrule
\end{figure*}
Here, $w_a\in\{S,D\}$ for access links and $w_b\in\{\mathrm{UAB},\mathrm{SAB}\}$ for backhaul links. The data rate of the UE averaging over temporally varying random variables is given by $\tilde{R} = $ $$\frac{\min\left(\mathbb{E}\left[\mathrm{D}_{ul,s,a,w_a}\Big| \BSprocess, \UEprocess, \mathcal{E}_s\right], \mathbb{E}\left[\mathrm{D}_{ul,s,b,w_b}\Big| \BSprocess, \UEprocess, \mathcal{E}_s\right]\right)}{\mathrm{T}\F}.$$ The data rate after spatial averaging is given by expectation of the aforementioned rate over $\BSprocess$ and $\UEprocess$ and is given by $\mathrm{R}_{ul,s,w_a,w_b} = \mathbb{E}\left[\tilde{R}\big| \mathcal{E}_s\right]\leq \min\left(\mathbb{E}\left[\mathrm{D}_{ul,s,a,w_a}\Big| \mathcal{E}_s\right], \mathbb{E}\left[\mathrm{D}_{ul,s,b,w_b}\Big| \mathcal{E}_s\right]\right)/\mathrm{T}\F.$ We will use this upper bound as an approximation to our mean rate estimates. We observe that the upper bound is very close to actual mean rate for $\delta$ larger or smaller than the optimal $\delta\in\{0.1,0.2,\ldots,0.9\}$, which is intuitive since the network is either highly access or backhaul limited in these scenarios and thus the minimum of expectation is roughly equal to expectation of minimum. For $\delta$ close to the optimal, there is some gap and in future it will be desirable to close this gap with a better approximation. However, the estimates for optimal $\delta$ were observed to be roughly same with the upper bound and the actual ergodic mean rate\cite{KulCode17}. 
\begin{thm}
\label{thm:ULrate}
Approximate mean rate of a typical UL user in the network is given by $\mathrm{R}_{ul,w_a,w_b} = \mathcal{A}_m \mathrm{R}_{ul,m,w_a} + \mathcal{A}_s \mathrm{R}_{ul,s,w_a,w_b}$, where 
\begin{multline*}
\mathrm{R}_{ul,m,w_a} = \mathbb{E}_{\mathcal{F}}\frac{\BW}{\F}  \sum_{n_1=0}^{\infty}\sum_{n_2=0}^{\infty}\frac{\Upsilon_m(n_1,n_2,4.5)}{n_1+1}\times\\
\int_{0}^{\infty}\frac{\sum_{i=1+\F_{ad,w_a,X^*}}^{\F_{a}}\mathsf{S}^{ul,m}_{i,a,w_a}(\tau)}{1+\tau}\mathrm{d}\tau.
\end{multline*}
$\mathbb{E}_\mathcal{F}$ denotes expectation is over $\mathcal{F}$. Also note that given the UL and DL loads $n_1$ and $n_2$, $\F_{ad,w_a,X^*}$ is computed as per Section~\ref{sec:scheduling}. 
\begin{equation*}
\mathrm{R}_{ul,s,w_a,w_b}= \mathbb{E}_{\mathcal{F}}\frac{\min\left(\mathrm{R}_{a,ul,s,w_a,w_b}, \mathrm{R}_{b,ul,s,w_a,w_b}\right)}{\F},
\end{equation*}
where $\mathrm{R}_{a,ul,s,w_a,w_b}$ is given in \eqref{eq:longeqRate2} and 
\begin{figure*}
\begin{multline}
\label{eq:longeqRate2}
\mathrm{R}_{a,ul,s,w_a,w_b} = \bwidth   \sum_{n_1=0}^{\infty}
\sum_{n_2=0}^{\infty}\frac{\Upsilon_s(n_1,n_2,4.5)}{n_1+1}\int_{0}^{\infty}\frac{\sum_{i=1+\F_{ad,w_a,X^*}}^{\F_{a}}\mathsf{S}^{ul,s}_{i,a,w_a}(\tau)}{1+\tau}\mathrm{d}\tau\\
+\mathds{1}(w_b=\mathrm{UAB}) \bwidth  (1-\mathbb{E}\left[1/N_{s,u}\right])p_{ul}\sum_{n=1}^{\infty}\frac{\kappa^*_{u,s}(n)}{n}   \int_{0}^{\infty}\frac{\sum_{i=1+\F_{a}+\F_{bd}}^{\F}\mathsf{S}^{ul,s}_{i,a,w_a}(\tau)}{1+\tau}\mathrm{d}\tau ,
\end{multline}
\hrulefill
\end{figure*}
\begin{multline*}
\mathrm{R}_{b,ul,s,w_a,w_b} =  \bwidth \mathrm{T} \mathbb{E}\left[1/N_{s,u}\right]\times\\ \sum_{n=1}^{\infty}\frac{\kappa^*_{u,s}(n)}{n}   \int_{0}^{\infty}\frac{\sum_{i=1+\F_{a}+\F_{bd}}^{\F}\mathsf{S}^{ul}_{i,b,w_b}(\tau)}{1+\tau}\mathrm{d}\tau,
\end{multline*}
where $w_a\in\{S,D\}$, $w_b\in\{\mathrm{SAB},\mathrm{UAB}\}$, $N_{s,u}$ has distribution in \eqref{eq:tagged} with $\epsilon = \lambda_s \left(1-\left(1+\frac{\mathcal{A}_s \lambda_u (1-\eta)}{3.5\lambda_s}\right)^{-3.5} \right)/\lambda_m$. Also, $\kappa^{*}_{u,s}$, $\Upsilon_m(.)$ and $\Upsilon_s(.)$ are given in Section~\ref{sec:loaddist}. Further, the notation $\sum\limits_{x}^{y}$ implicitly assumes that the sum is zero if $y<x$.  
\end{thm} 
\begin{proof}
See Appendix~\ref{App_ULrate}. 
\end{proof}
\begin{rem}
The infinite summations in Theorem~\ref{thm:ULrate} correspond to averaging some load distribution, as inferred from Appendix~\ref{App_ULrate}. These can be computed accurately as finite sums with roughly $6x$ terms if the mean load for the particular summation is $x$. 
\end{rem}

\section{Downlink SINR and rate}
\label{sec:DLanalysis}
Analyzing DL SNR distribution is very similar to UL, and the key difference lies in the interference distribution which results due to the receiver position now being at the origin instead of $X^*$ or $X^{**}$ as in the UL case. This leads to different exclusion regions that need to be considered while computing shot noise of the interfering points as will be clear in Appendix~\ref{sec:App_DL_Laplace}. For rate computations, another major difference arises due to different probability of being scheduled in $i^\text{th}$ slot for DL and UL UEs, that depends on the DL subframe length distribution in access and backhaul subframes as a function of $\eta$. 

\noindent{\bf SINR distribution for access links.}
DL SINR of a typical UE at the origin being served by a BS at $X^*\in\Phi_t$, where $t\in\{m,s\}$, is given as follows.
$$
\SINR^{dl}_{i,a,w} = \frac{\mathrm{C}_0 \power{t} h_{i,0,X^*} G_u G_t L(0,X^*)^{-1}}{I_{i,m,w}(0) + I_{i,s,w}(0) + I_{i,u,w}(0) + \sigma^2},
$$
where $w\in\{S,D\}$ if $i\leq \mathrm{F}_a$ and $w\in\{\text{SAB}, \text{UAB}\}$ if $i>\mathrm{F}_a$. $I_{i,\nu,w}(0)$ is the interference power at origin from all active devices of type $\nu\in\{m,s,u\}$ in the $i^\text{th}$ slot as given in \eqref{eq:Inu_access}-\eqref{eq:Iu_backhaul}. Note that here the DL access link will be active only when $\F_a<i\leq \F_a+\F_{bd}$ in the backhaul subframe and thus, the second sum in \eqref{eq:Is_backhaul} would be non-zero but the first summation would be zero. 

The SINR distribution is given similar to \eqref{eq:ULSINRdist_access} and is given as follows, $\mathsf{S}^{dl,t}_{i,a,w}(\tau) =$
\begin{multline}
\label{eq:DLSINRdist_access}  
\hspace{-0.25cm}\sum_{\substack{t\in\{s,m\},\\\mu\in\{l,n\}}}\int_{0}^{\infty}\exp \left( \frac{-\tau R^{\alpha_\mu} \sigma^2}{\mathrm{C}_0 \power{t} G_u G_t}\right) L^{dl,a,t,\mu}_{i,w}\left(\frac{\tau R^{\alpha_\mu}}{\mathrm{C}_0 \power{t} G_u G_t }, R\right)\\
\times\prod_{\substack{t'\in\{m,s\},\\\mu'\in\{l,n\},\\ t'\neq t \text{ or }
  \mu'\neq \mu}} F_{t',\mu'}\left(\left(\frac{P_{t'} G_{t'} B_{t'} R^{\alpha_{\mu}}}{P_t B_t G_t}\right)^{\frac{1}{\alpha_{\mu'}}}\right) \frac{f_{t,\mu}(R)}{\mathcal{A}_t} \mathrm{d}R.
\end{multline}
Note the different transmit power here and also that $L^{dl,a,t,\mu}_{i,w}(\s,R)$, derived in Appendix~\ref{sec:App_DL_Laplace}, is different from the UL Laplace transform of interference given in Lemmas~\ref{lem:LaplaceULaccess} and \ref{lem:LaplaceULaccess2}. 

\noindent{\bf SINR distribution for backhaul links.}
For DL backhaul link, considering a typical SBS at origin being served by a MBS at $X^{**}$,
$$
\SINR^{dl}_{i,b,w} = \frac{\mathrm{C}_0 \power{m} h_{i,X^{**},0} G_m G_s L(X^{**},0)^{-1}}{I_{i,s,w}(0) + I_{i,m,w}(0) + \sigma^2},
$$
where $w\in\{\text{UAB},\text{SAB}\}$, $\F_a <i\leq \F_a+\F_{bd}$, $I_{i,s,\text{SAB}}=0$.  $I_{i,m,w}(0)$ and  $I_{i,s,\text{UAB}}(0)$ can be obtained from \eqref{eq:Im_backhaul} and \eqref{eq:Is_backhaul}, respectively. $\mathsf{S}^{dl}_{i,b,w}$ is same as Corollary~\ref{cor:ULSINRdist_backhaul} with  $L^{ul,b}_{i,w}$ replaced by 
$L^{dl,b}_{i,w}(\mathrm{s},\rho) \approx L_m L_s,$
where 
\begin{equation*}
L_s = \exp\left(-\int_{0}^{\infty}\mathbb{E}\left[\frac{\bar{\lambda}_d\Lambda_s(\mathrm{d}r)/\lambda_s}{1+ \frac{r}{\s\mathrm{C}_0\power{u}\Psi_{m,u}}}\right] \right),
\end{equation*}
if $w=\mathrm{UAB}$ and $L_s = 1$ if $w = \mathrm{SAB}$. Here, 
\begin{multline*}
\bar{\lambda}_d =  \left(\lambda_s-\left(1-\left(1+\frac{\lambda_s }{3.5\lambda_m}\right)^{-3.5} \right)\lambda_m\right)^+\times\\ \left(1-\left(1+\frac{\lambda_u \eta\mathcal{A}_s
}{3.5\lambda_s}\right)^{-3.5} \right)p_{dl}.
\end{multline*}
and $L_m =  \exp\left(-\theta\right)$, where
\begin{equation*}
\theta = \int_{\rho^{\alpha_l}}^{\infty} \mathbb{E}\left[\frac{\left(1-\left(1+\frac{\lambda_{s,d}}{3.5\lambda_m}\right)^{-3.5}\right)\Lambda_m(\mathrm{d}r)}{1+ \frac{r}{\s\mathrm{C}_0\power{m}\Psi_{m,s}}}\right],
\end{equation*}
and $\lambda_{s,d} = \lambda_s\left(1-\left(1+\frac{\lambda_u \eta\mathcal{A}_s
}{3.5\lambda_s}\right)^{-3.5} \right).$ The expectation in the expression for $L_s$ and $\theta$ is with respect to the antenna gains $\Psi_{(.)}$ given in Table~\ref{tab:antenna}. 
\begin{thm}
\label{thm:DLrate}
The mean rate of a typical DL user in the network is given by $\mathrm{R}_{dl,w_a,w_b} = \mathcal{A}_m \mathrm{R}_{dl,m,w_a} + \mathcal{A}_s \mathrm{R}_{dl,s,w_a,w_b}$, where 
\begin{multline*}
\mathrm{R}_{dl,m,w_a} = \mathbb{E}_{\mathcal{F}}\frac{\BW}{\F}  \sum_{n_1=0}^{\infty}\sum_{n_2=0}^{\infty}\frac{\Upsilon_m(n_1,n_2,4.5)}{n_2+1}\\\int_{0}^{\infty}\frac{\sum_{i=1}^{\F_{ad,w_a,X^*}}\mathsf{S}^{dl,m}_{i,a,w_a}(\tau)}{1+\tau}\mathrm{d}\tau,
\end{multline*}
\begin{equation*}
\mathrm{R}_{dl,s,w_a,w_b}= \mathbb{E}_{\mathcal{F}}\frac{\min\left(\mathrm{R}_{a,dl,s,w_a,w_b}, \mathrm{R}_{b,dl,s,w_a,w_b}\right)}{\mathrm{T}\F},
\end{equation*}
\begin{multline*}
\mathrm{R}_{a,dl,s,w_a,w_b} = \bwidth \mathrm{T}  \sum_{n_1=0}^{\infty}\sum_{n_2=0}^{\infty}\frac{\Upsilon_s(n_1,n_2,4.5)}{n_2+1}\times\\
\int_{0}^{\infty}\frac{\sum_{i=1}^{\F_{ad,w_a,X^*}}\mathsf{S}^{dl,s}_{i,a,w_a}(\tau)}{1+\tau}\mathrm{d}\tau
+\mathds{1}(w_b=\mathrm{UAB}) \bwidth \mathrm{T}\times \\(1-\mathbb{E}\left[1/N_{s,d}\right])p_{dl}\sum_{n=1}^{\infty}\frac{\kappa^*_{d,s}(n)}{n}   \int_{0}^{\infty}\frac{\sum_{i=1+\F_{a}}^{\F_{bd}}\mathsf{S}^{dl,s}_{i,a,w_a}(\tau)}{1+\tau}\mathrm{d}\tau ,
\end{multline*}
\begin{multline*}
\mathrm{R}_{b,dl,s,w_a,w_b} =  \bwidth \mathrm{T} \mathbb{E}\left[1/N_{s,d}\right]\times\\ 
\sum_{n=1}^{\infty}\frac{\kappa^*_{d,s}(n)}{n}   \int_{0}^{\infty}\frac{\sum_{i=1+\F_{a}}^{\F_{bd}}\mathsf{S}^{dl}_{i,b,w_b}(\tau)}{1+\tau}\mathrm{d}\tau,
\end{multline*}
where $w_a\in\{S,D\}$ and $w_b\in\{\mathrm{SAB},\mathrm{UAB}\}$. Here, $N_{s,d}$ has distribution  as in \eqref{eq:tagged} with $\epsilon = \frac{\lambda_s \left(1-\left(1+\frac{\mathcal{A}_s \lambda_u \eta}{3.5\lambda_s}\right)^{-3.5} \right)}{\lambda_m}$. Also, $\kappa^{*}_{d,s}$, $\Upsilon_m(.)$ and $\Upsilon_s(.)$ are given in Section~\ref{sec:loaddist}.
\end{thm} 
\begin{proof}
Follows Appendix~\ref{App_ULrate}. Note the subtle differences in the limits of summations inside the integrals here compared to Theorem~\ref{thm:ULrate}. This is due to the different subframes in which an UL or DL UE or SBS is scheduled. 
\end{proof}
\begin{figure*}
  \centering
\subfloat[SINR validation. Dotted lines are with Monte Carlo simulations. ]
{\label{fig:plot1}{\includegraphics[width= \columnwidth]{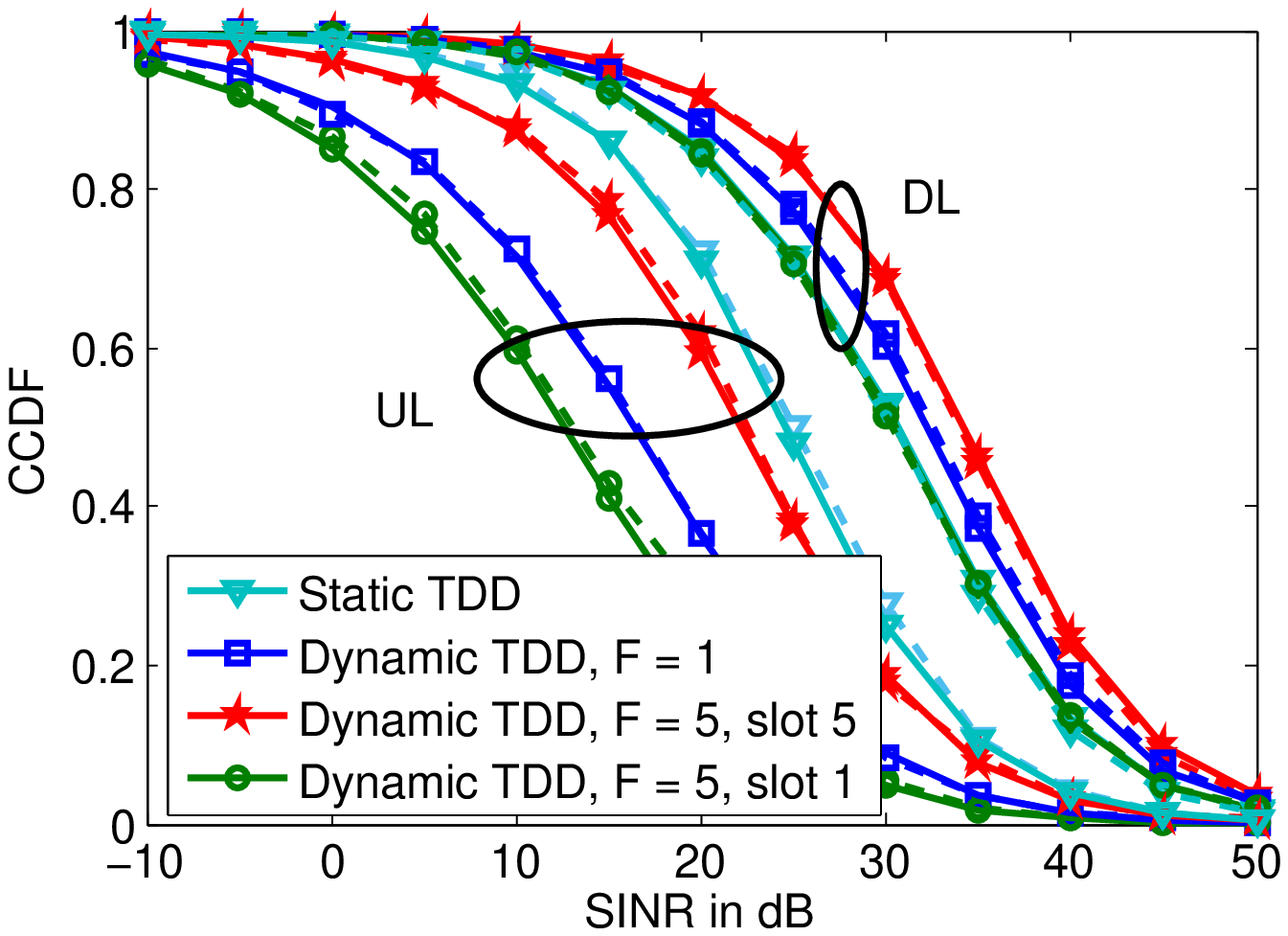}}}
\subfloat[UL and DL Mean Rates]
{\label{fig:plot2}{\includegraphics[width=\columnwidth]{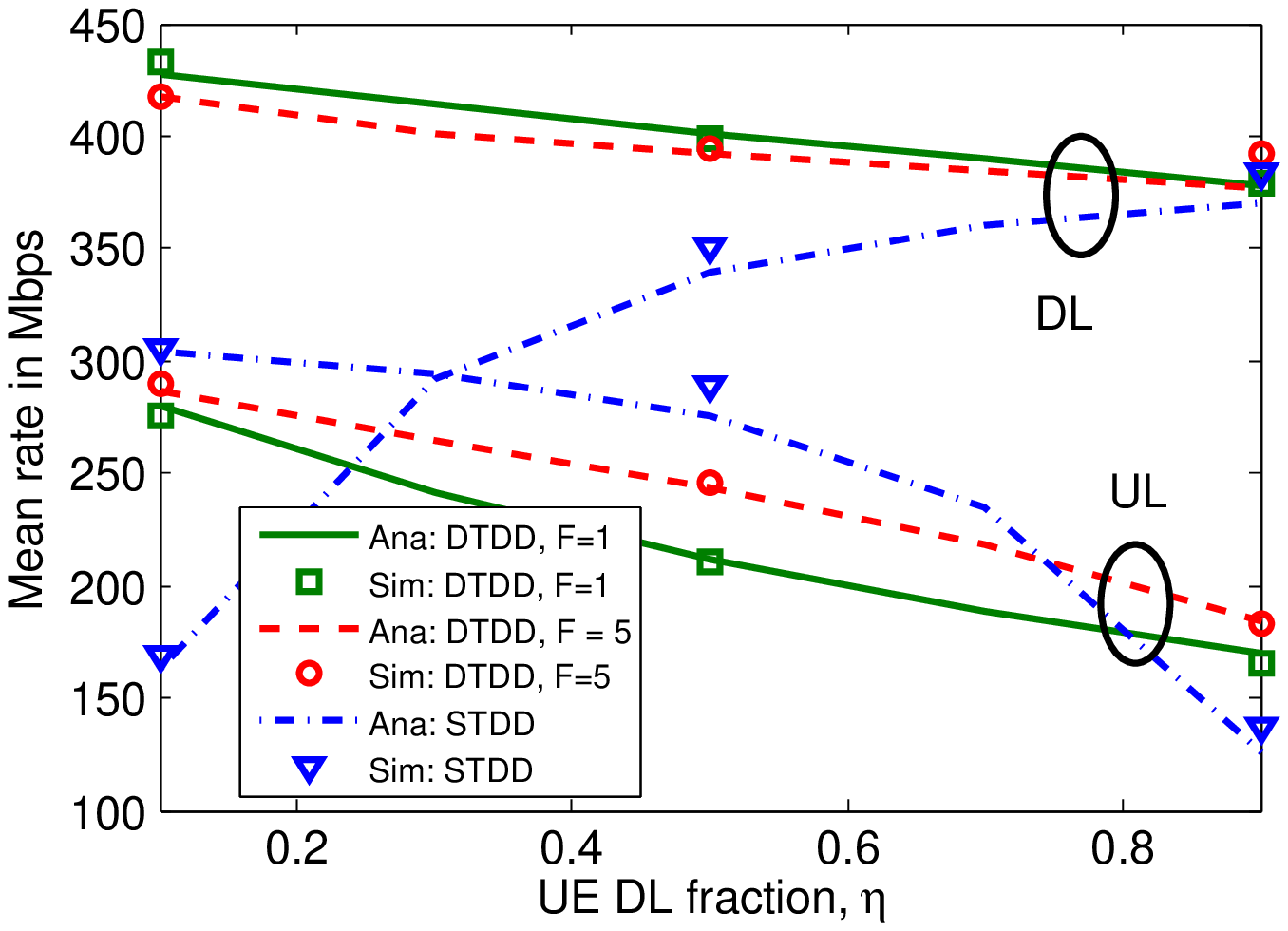}}}
\caption{Impact of frame size and analysis validation.}
 \label{fig:plot12}
 \end{figure*}
\begin{cor}
The mean rate of a typical user is given by 
$
\mathrm{R}_{w_a,w_b} = \eta \mathrm{R}_{dl,w_a,w_b} + (1-\eta) \mathrm{R}_{ul,w_a,w_b}.
$
\end{cor}
\begin{proof}
The typical point at origin is DL with probability $\eta$ and UL with probability $1-\eta$. 
\end{proof}
\begin{rem}
We recommend to first evaluate the SINR coverage for different thresholds, and then use the stored values to compute the numerical integrals involved in the mean rate formulae. Our codes can be accessed at \cite{KulCode17}.  
\end{rem}

\section{Numerical results}
\label{sec:results}
First we study static vs dynamic TDD when all BSs are MBSs. Then we introduce wirelessly backhauled SBSs into the network and study the comparison of TDD schemes. 
 \begin{figure*}
  \centering
\subfloat[$f_c = 28$ GHz, $\res = 200$ MHz]
{\label{fig:plot3}{\includegraphics[width= \columnwidth]{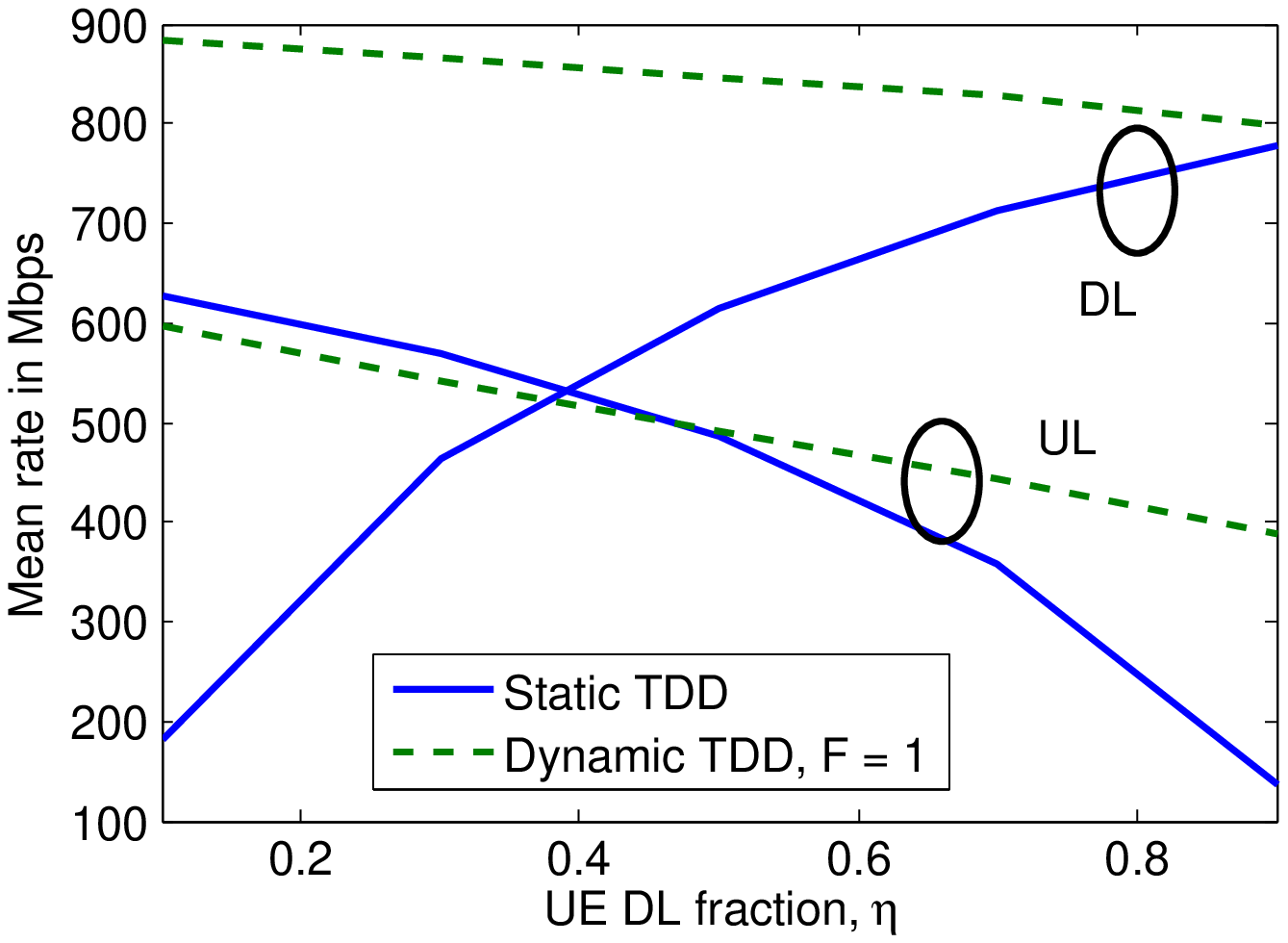}}}
\subfloat[$f_c = 73$ GHz, $\res = 2$ GHz]
{\label{fig:plot4}{\includegraphics[width=\columnwidth]{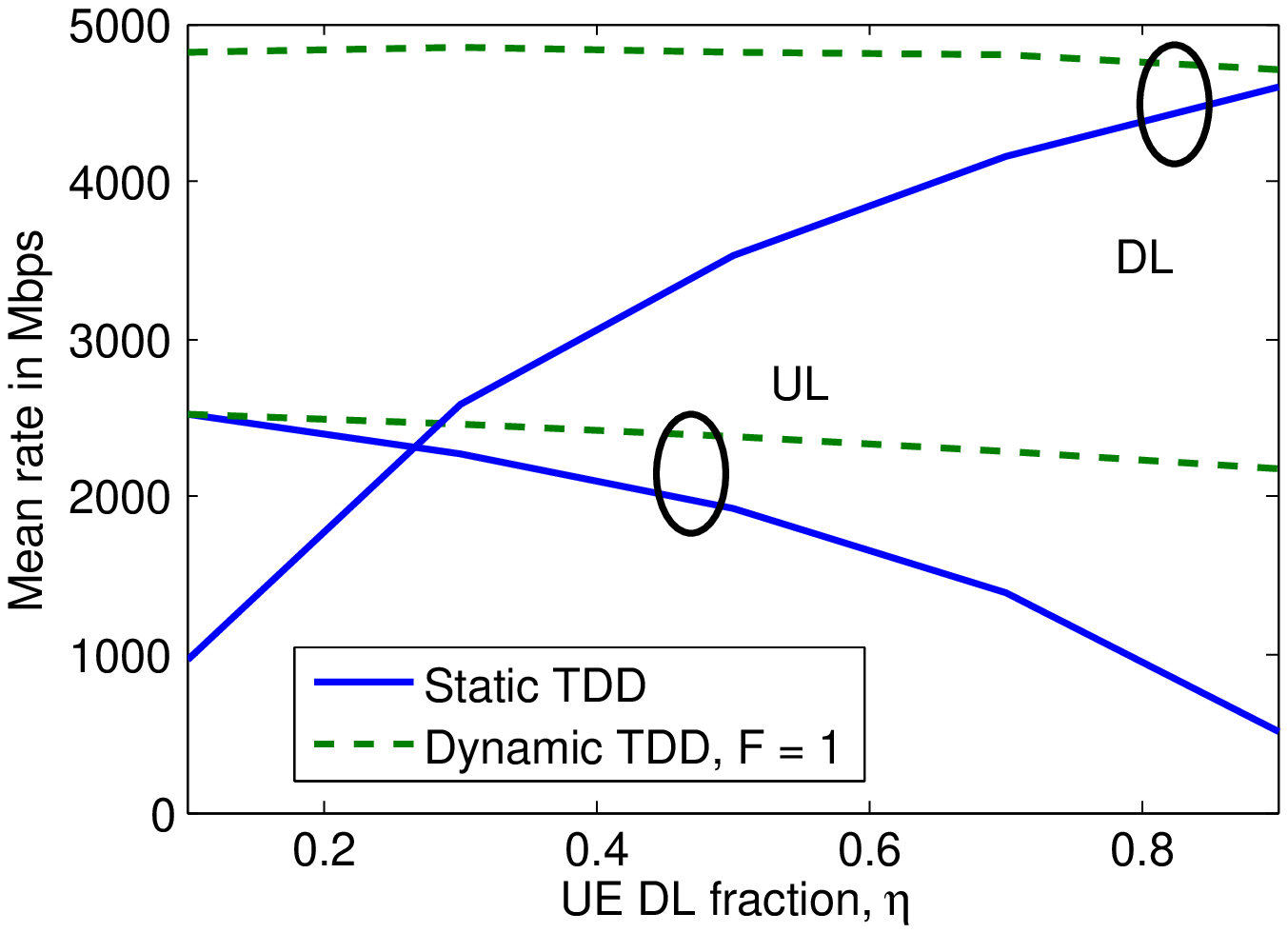}}}
\caption{Dynamic vs Static TDD, $\lambda_u = 200/$km$^2$. Dynamic TDD helps the ``rare" UEs in the network perform much better.}
 \label{fig:plot34}
\end{figure*}
\begin{figure*}
  \centering
\subfloat[Uplink access. Dotted lines with Monte Carlo simulations. ]
{\label{fig:plot5}{\includegraphics[width=\columnwidth]{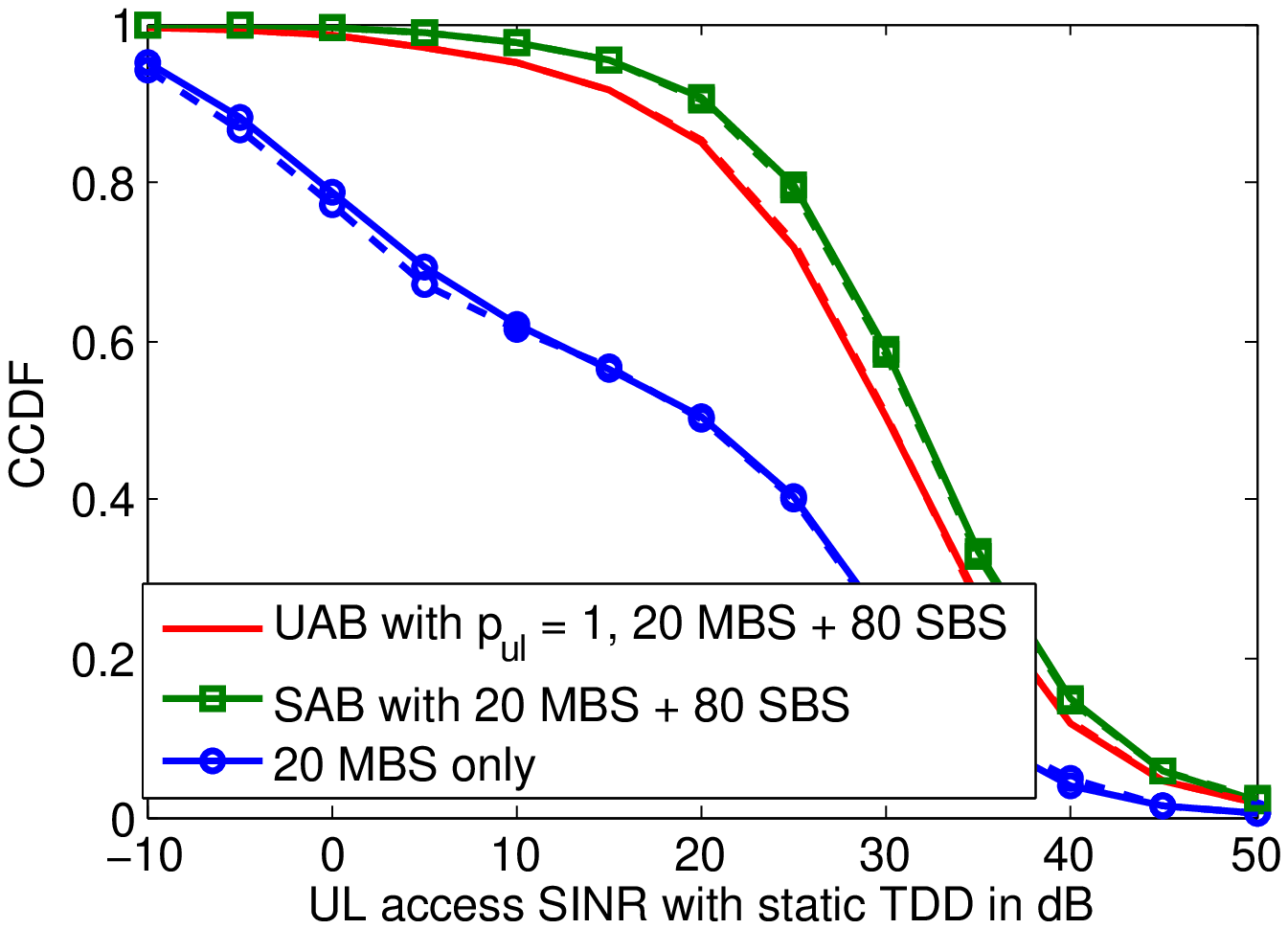}}}
\subfloat[Backhaul]
{\label{fig:plot6}{\includegraphics[width= \columnwidth]{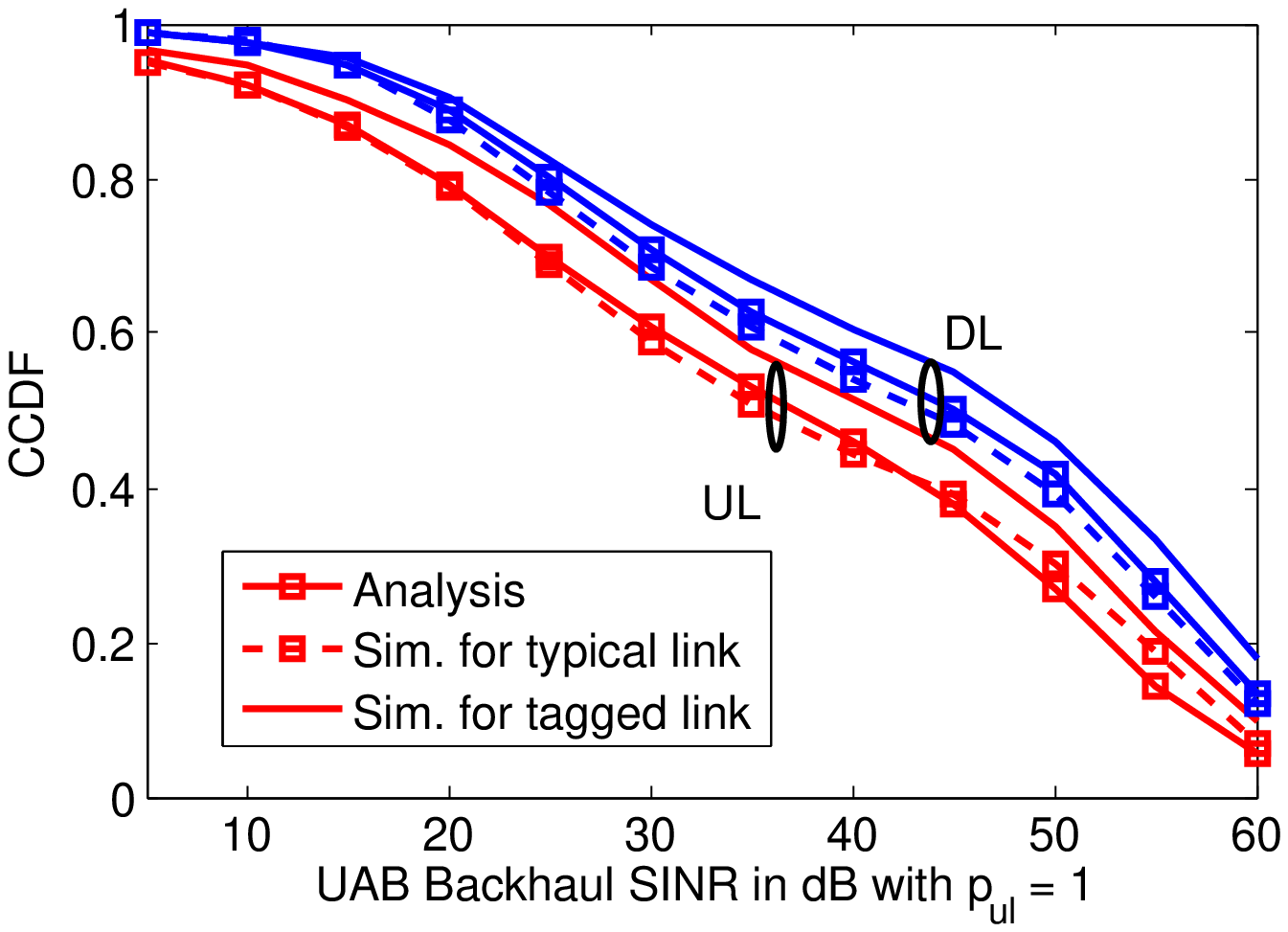}}}
\caption{SINR validation with self-backhauling for $\eta = 0.5$. Also shows self-backhauling is a good coverage solution.}
 \label{fig:plot56}
\end{figure*}
\subsection{Dynamic vs static TDD when all BSs are MBSs}
\label{sec:resultsMBSonly}
{\bf Validation of analysis and impact of frame size.}
Fig.~\ref{fig:plot1} validates UL and DL SINR distribution with static and dynamic TDD for frame-size $\F=1$ and $\F=5$ with $\eta = 0.5$ and $\lambda_u = 500/$km$^2$. The Monte Carlo simulations match the analytical results very well. Fig.~\ref{fig:plot1} also shows that UL SINR coverage with static TDD is better than with dynamic TDD but DL SINR coverage with dynamic TDD is better than with static TDD. This is primarily because of the transmit power disparity between UL and DL. For a moderately loaded system, as considered in this setup, the average number of interferers seen by a typical UL user is roughly the same with static and dynamic TDD. However, with dynamic TDD some of these interferers now have 10 dB more transmit power, which increases the interference and thus lowers SINR coverage. Note that the location of the interferers with static and dynamic TDD are different and thus a theoretical result like stochastic dominance of UL SINR with static TDD over dynamic TDD cannot be stated. 

Fig.~\ref{fig:plot1} further shows that the UL SINR coverage with dynamic TDD for slot 5 with $\F=5$ is better by about $10$ dB than $\F=1$ and by $15$ dB for slot 1 with $\F = 5$, which is significant. This can be explained as follows. For $\F=1$, the probability that an interferer is DL is 0.5, whereas for $\F=5$ the probability rises to 0.95 (computed using the formula in Lemma \ref{lem:Fad}) for slot 1 and decreases to 0.04 for slot 5. Since DL transmit power is much higher than UL, the UL SINR coverage for $\F=1$ falls between the two curves for $\F=5$. {\em Thus, there is an inherent UL interference mitigation with larger frame size} since UL UE has more chances on being scheduled towards end of the frame than at the beginning, as can be seen in Fig.~\ref{fig:plot2}. Similar observations can be made for DL but are less pronounced since DL to DL interference is less significant than DL to UL due to low UL transmit power. 

{\bf Dynamic TDD not desirable in high load interference-limited scenarios but desirable in low load and asymmetric traffic scenarios.} Fig.~\ref{fig:plot2} plots the UL and DL mean rates with static and dynamic TDD for different values of $\eta$. First, note that the analytical formula gives a close match with the Monte Carlo simulations. Dynamic TDD essentially helps boost the rates of the ``rare" UEs in the network. For example, the DL rates double when $\eta = 0.1$ with dynamic TDD. In this scenario, there is about $5.6\%$ loss in UL rate with dynamic TDD. Similarly, note the $1.5\times$ gain for UL when $\eta = 0.9$. This indicates that dynamic TDD can be beneficial in asymmetric traffic scenarios but the gains are not very significant for $\eta$ close to 0.5, in fact there is $15\%$ gain for DL but $11\%$ loss for UL. {\em Thus, in high load interference-limited scenarios it is beneficial to switch to load aware static TDD}. The comparison is more persuasive for dynamic TDD in a low load scenario as shown in Fig.~\ref{fig:plot3} and even more for noise-limited 73 GHz network with 2 GHz bandwidth as shown in Fig.~\ref{fig:plot3}. For example, Fig.~\ref{fig:plot3} shows that the mean rates with DL (UL) are $5\times$ with dynamic TDD for $\eta = 0.1(0.9)$. Even for $\eta = 0.5$, there is a gain of $23\%$ for UL and $37\%$ for DL. To summarize the observations for MBS only scenario: low load, asymmetric traffic, and noise-limitedness benefit dynamic TDD. 
\subsection{Impact of self-backhauling}
{\bf Validation of analysis.}
\begin{figure}
\centering
\includegraphics[width = \columnwidth]{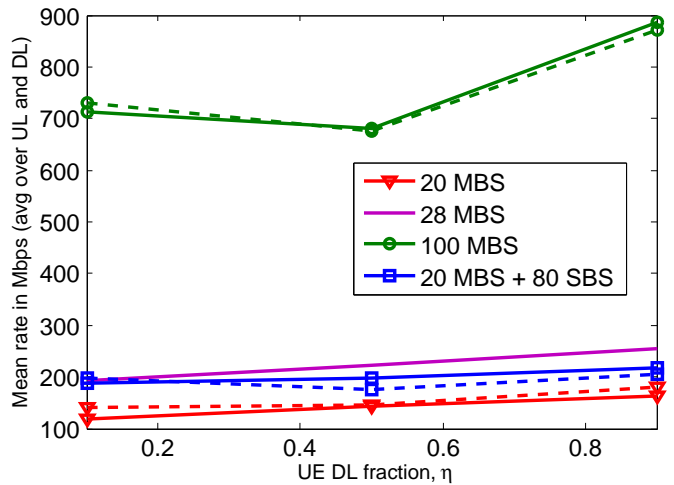}
\caption{Self-backhauling is a poor substitute for wired backhauling. Dotted lines with Monte Carlo simulations. }
\label{fig:plot7}
\end{figure}
Fig.~\ref{fig:plot56} validates the SINR coverage for access and backhaul links for the 28GHz network under consideration, and a very close match is seen between analysis and Monte Carlo simulations. In Fig.~\ref{fig:plot6} it can be seen that assuming typical SBS SINR instead of tagged SBS SINR can give an error of about 2-3dB, which is reasonable for analyzing mean rates as seen in Fig.~\ref{fig:plot7}. 
\begin{figure*}
  \centering
\subfloat[28 GHz, 200 MHz]
{\label{fig:plot13}{\includegraphics[width=\columnwidth]{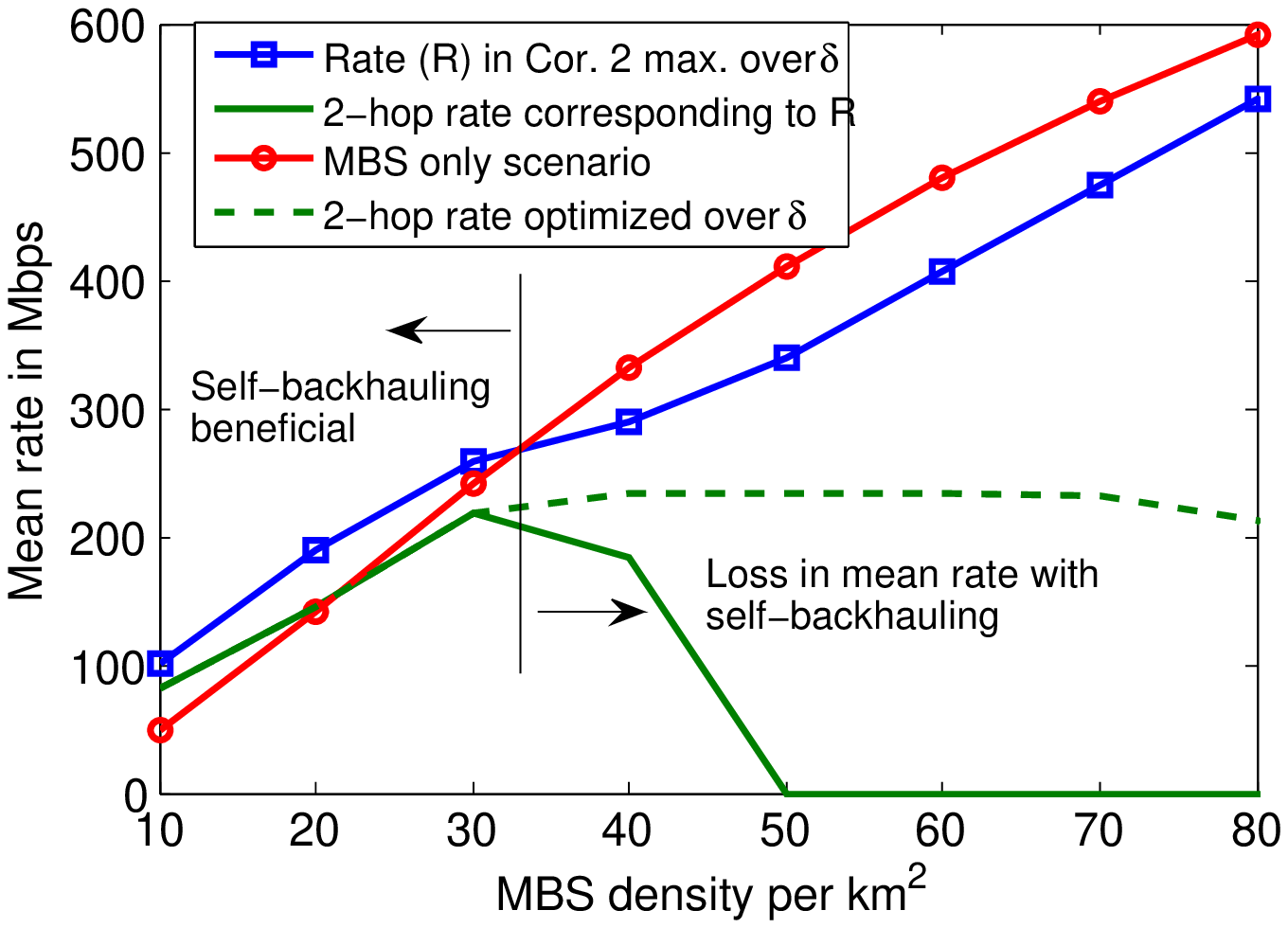}}}
\subfloat[73 GHz, 2 GHz]
{\label{fig:plot14}{\includegraphics[width= \columnwidth]{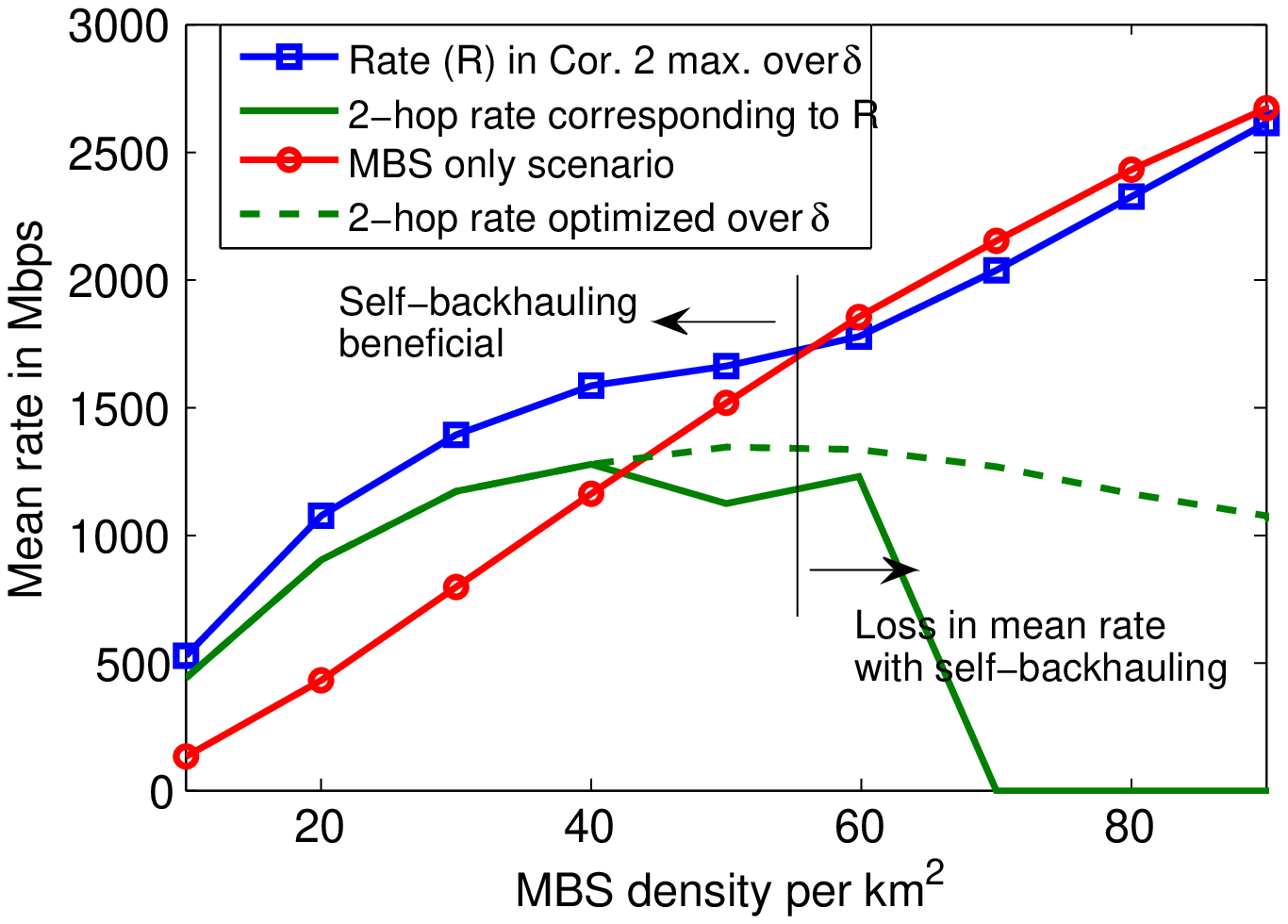}}}
\caption{Fix $\lambda_b = 100/$km$^2$ and vary $\lambda_m$. Optimization over $\delta$ is done by choosing the best from $\{0.1,0.2,\ldots,1\}$.}
 \label{fig:plot1314}
\end{figure*}

{\bf Low cost coverage solution but not for boosting mean rate.}
Fig.~\ref{fig:plot6} also shows that the $95^\text{th}$ percentile SINR increases by almost 20 dB when 80 additional SBSs are introduced to a baseline MBS only network. This clearly shows the coverage improvement with self-backhauling that translates into significant gain in cell edge rates. For example, here the cell edge rates go from $4.7\times10^6$ to $2.5\times 10^7$ for $\eta=0.5$. However, as can be seen from Fig.~\ref{fig:plot7} the mean rates increase by only $33\%-57\%$ across different $\eta$ after addition of 80 SBSs. 
This is equivalent to adding only 8 MBSs in terms of mean rate, although the 20 dB coverage improvement will not be seen in that case. Note that the mean rate values for the self-backhauling case in  Fig.~\ref{fig:plot7} are for static TDD with SAB and $\delta$ is chosen to be the maximizer of mean rates. If 80 MBSs were added instead of 80 SBSs, the rates increase by more than 7$\times$ compared to baseline scenario. Thus, self-backhauling is a low cost coverage solution and not for increasing data rates. 

{\bf Trends with network densification.}
Fig.~\ref{fig:plot1314} compares the mean rate of self-backhauled networks with $\lambda_b$ fixed at $100/$km$^2$ and varying $\lambda_m/\lambda_s$ and MBS only networks with $\lambda_m=100/$km$^2$. {\em One would expect that adding SBSs on top of MBSs would always increase the rate. However, counter-intuitively this does not occur.} When MBS density is low, as expected adding SBSs such that total density is $100$/km$^2$ increases data rates. The rates shown in the Figure correspond to the access backhaul split that maximizes rate. When MBS density$\geq 50/$km$^2$ in Fig.~\ref{fig:plot13} and $\geq 70/$km$^2$ in Fig.~\ref{fig:plot14}, the 2 hop rates corresponding to optimal $\delta$ go to zero implying $\delta = 1$. This occurs because the 2 hop rates are much lower than the single hop rates (the dotted line in the figure shows this wherein $\delta$ was chosen to maximize the 2 hop rate) and maximizing over mean rate kills the 2 hop rates to zero, giving as many resources to direct links. This indicates that {\em when there are enough MBSs, adding just a few SBSs may not be beneficial as the slight benefit in coverage is overshadowed by the loss due to 2 hops}. The losses can be converted to no-loss by biasing UEs towards MBS. Fig.~\ref{fig:plot14} corresponds to a noise-limited scenario and also in this case the DL access transmit power is reduced to 20dBm keeping backhaul transmit power as 30dBm as an example of a network which is less backhaul-limited. In this case the ``beneficial" regime with self-backhauling is pushed further towards $\lambda_b$. 
\begin{figure*}
  \centering
\subfloat[SBS saturation density increases for higher MBS density.]
{\label{fig:plot15}{\includegraphics[width= \columnwidth]{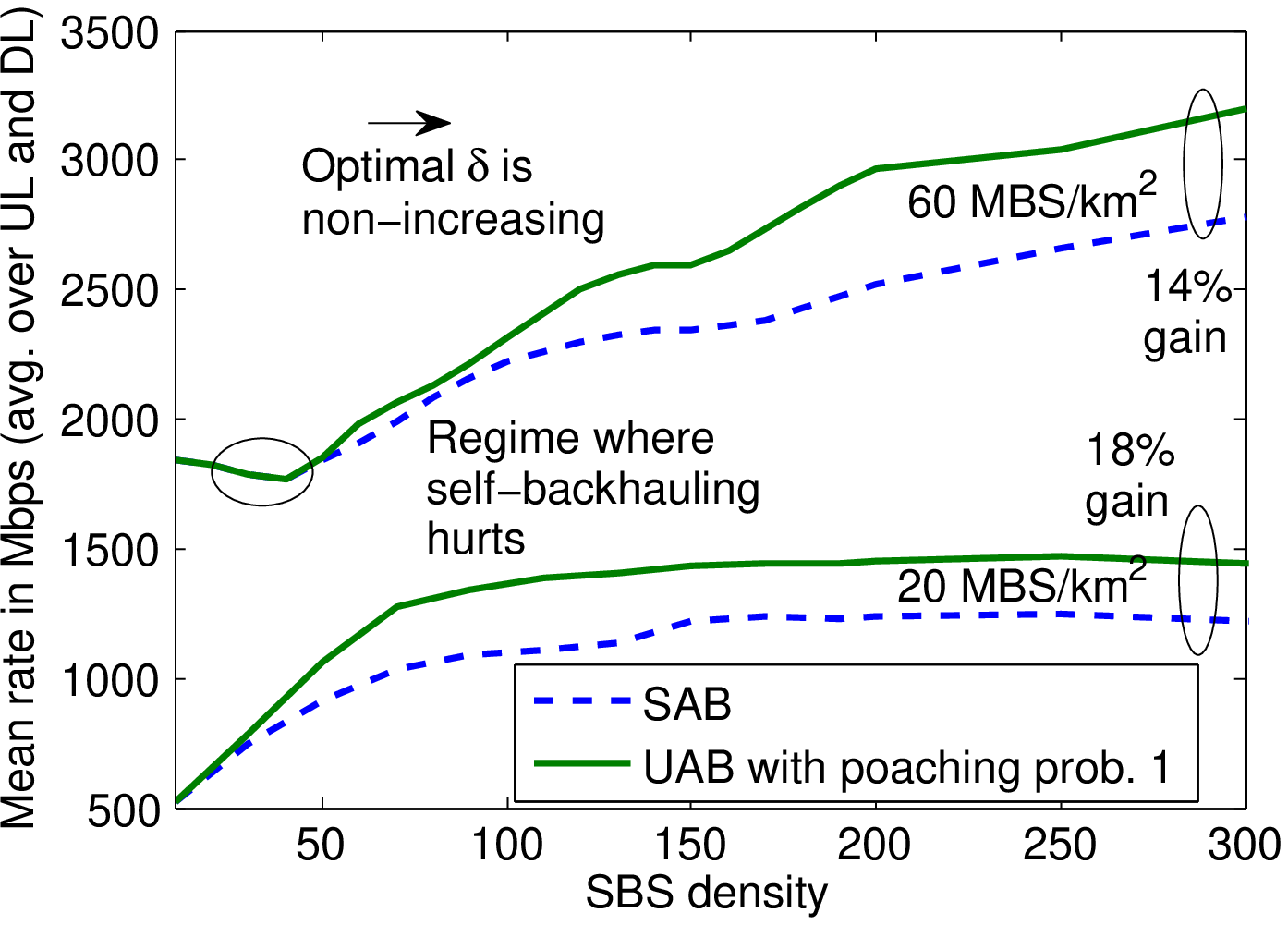}}}
\subfloat[Network becomes backhaul limited with increasing SBS density.]
{\label{fig:plot16}{\includegraphics[width=\columnwidth]{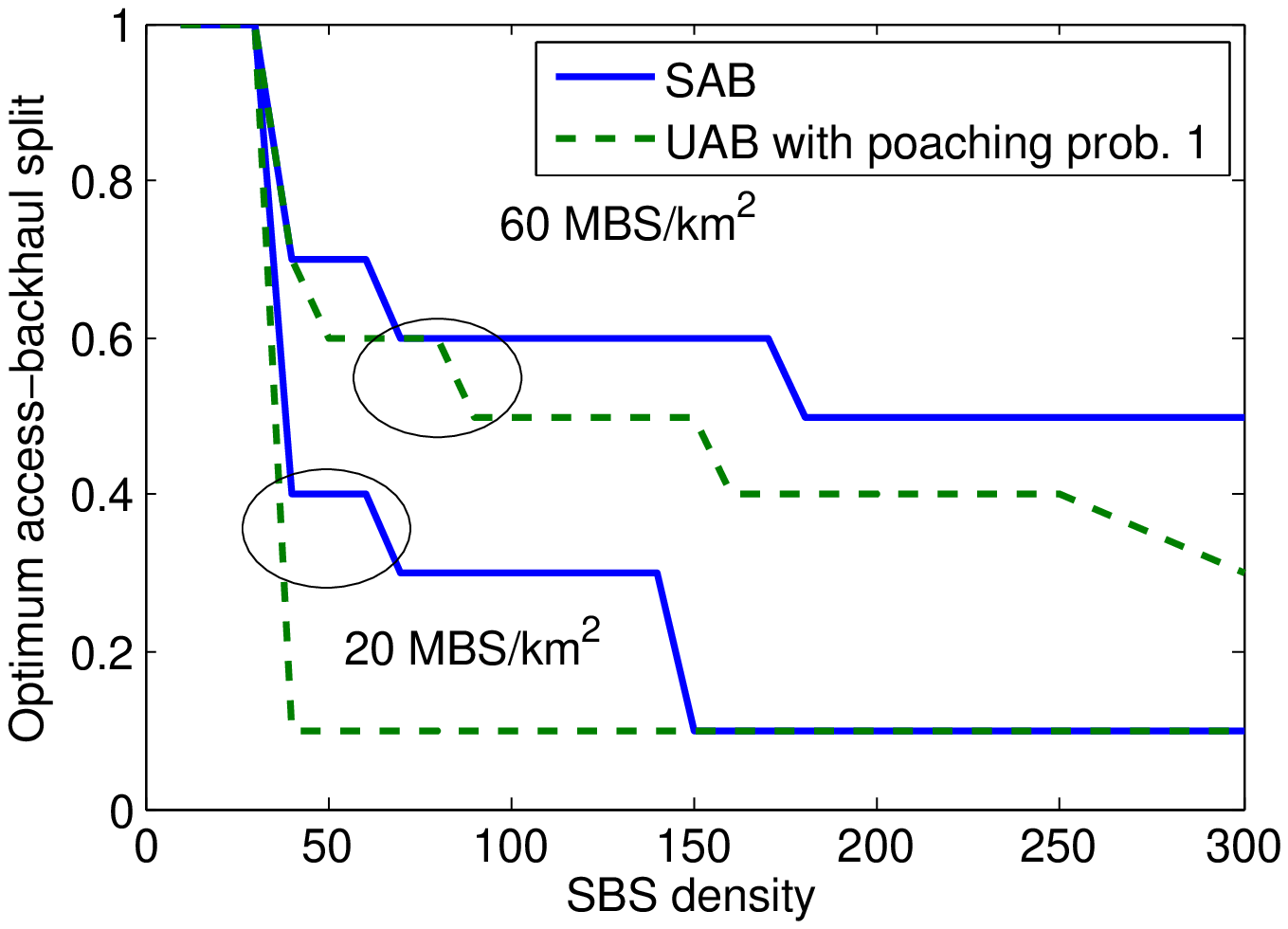}}}
\caption{Fix $\lambda_m$ and vary $\lambda_s$. Here, $\eta=0.5$.}
 \label{fig:plot1516}
 \end{figure*}

In Fig.~\ref{fig:plot15}, for a fixed $\lambda_m$, the value of $\lambda_s$ is increased. For each self-backhauling configuration an optimum $\delta$ is chosen from the set $\{0.1,0.2,\ldots,1\}$ and is shown in Fig.~\ref{fig:plot16}. {\em The optimum $\delta$ is non-increasing with SBS density and UAB as is expected}. Since more UEs connect with SBSs, we need more backhaul slots in a frame. There are another couple of observations to be made in Fig.~\ref{fig:plot15}. Firstly, note that UAB gives about $10-20\%$ gain over SAB. The gain is negligible or none at lower SBS densities wherein there are not many backhaul slots to be poached. Also note that {\em the rates saturate sooner in the 20 MBS case than the 60 MBS case}. As SBS density becomes large, the network becomes {\em backhaul limited} as indicated by the decreasing optimum $\delta$ in Fig.~\ref{fig:plot16}. Similar observations can be noted for the 28 GHz network, although the gains with UAB are negligible in that case due to increasing interference. 
  \begin{figure*}
    \centering
  \subfloat[Downlink]
  {\label{fig:plot8}{\includegraphics[width= \columnwidth]{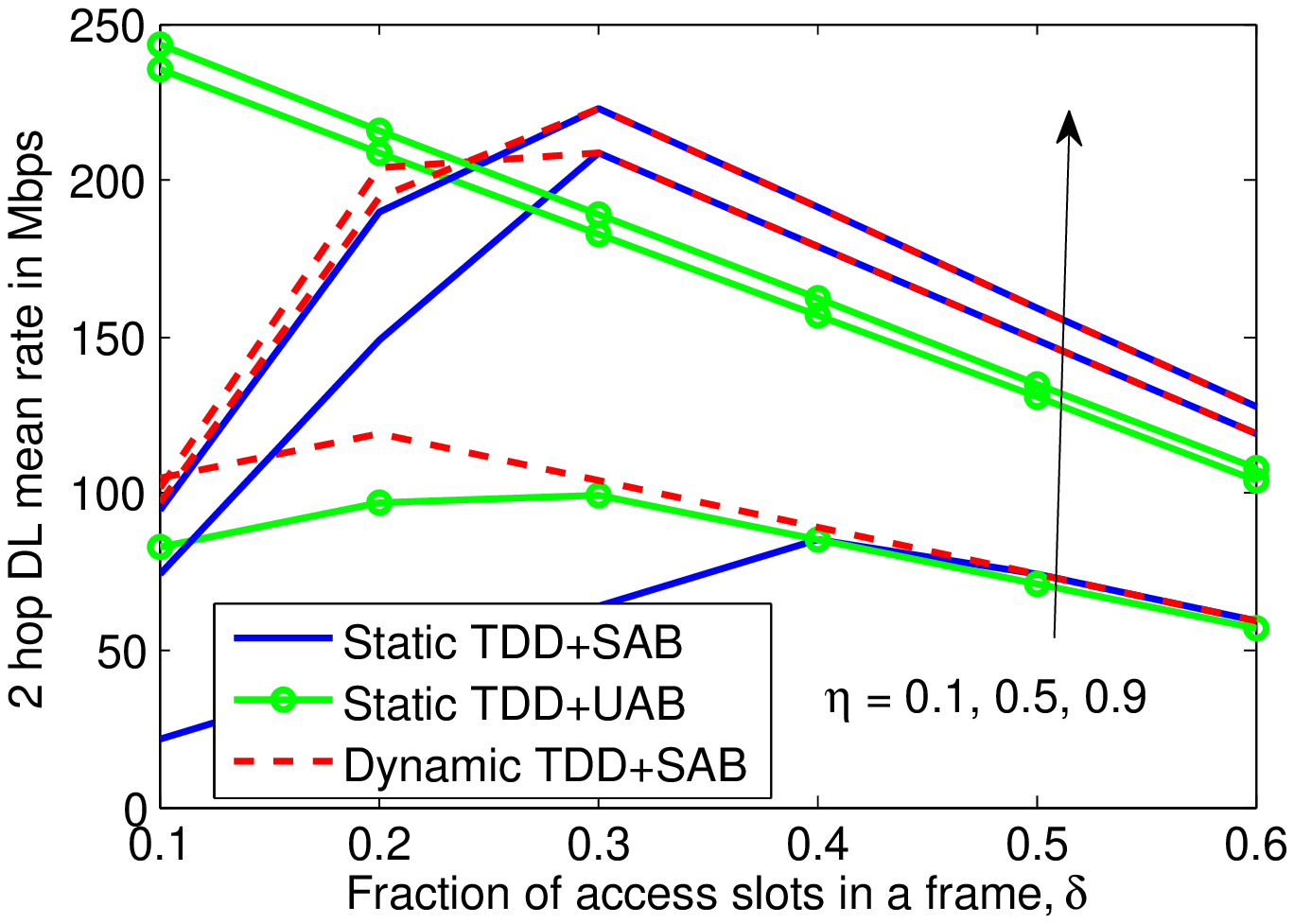}}}
  \subfloat[Uplink]
  {\label{fig:plot9}{\includegraphics[width=\columnwidth]{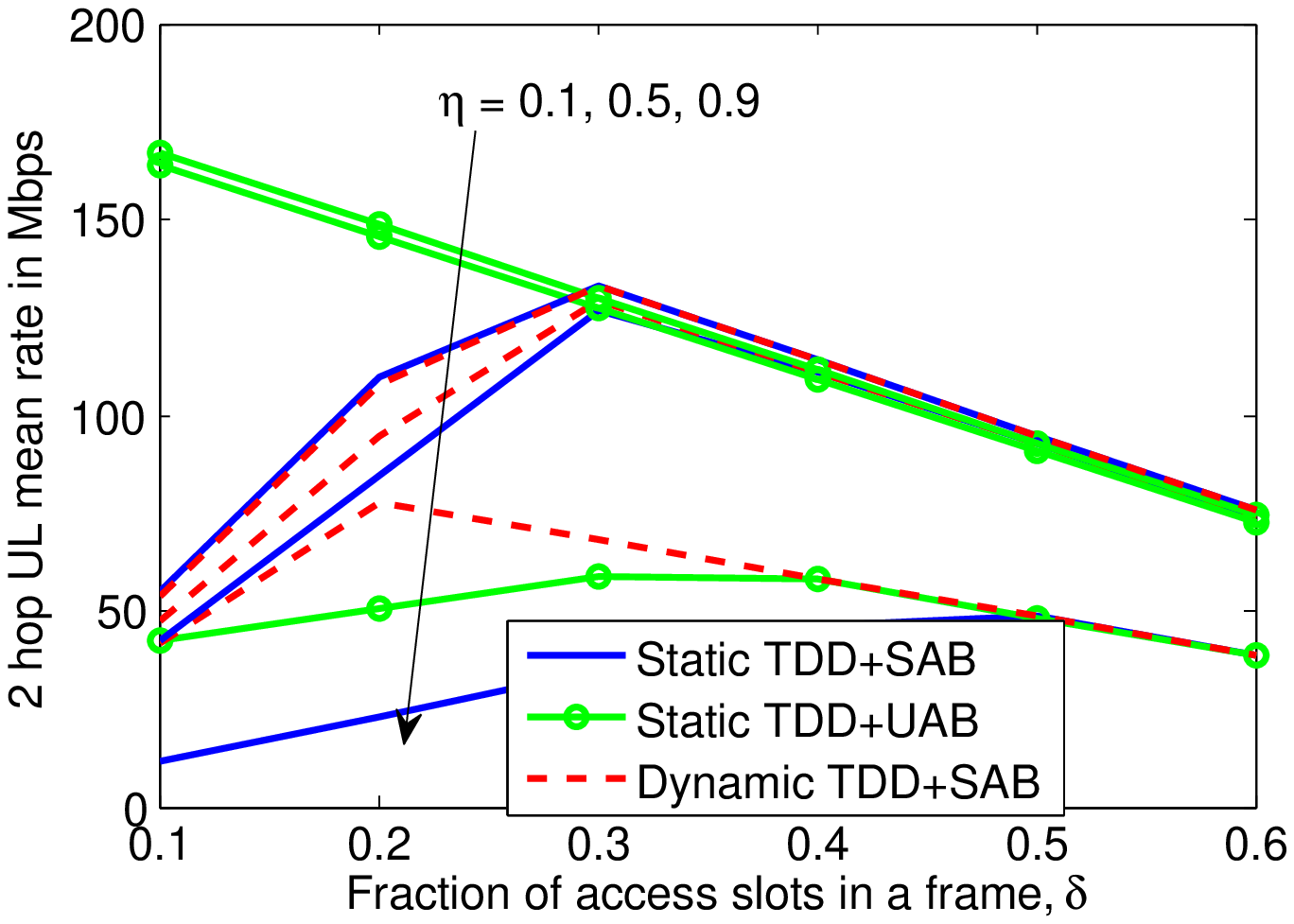}}}
  \caption{Comparison of TDD schemes across different $\delta$ and $\eta$, and impact on optimal $\delta$.}
   \label{fig:plot89}
  \end{figure*}
\subsection{Comparison of TDD schemes}
{\bf Gains from Dynamic TDD and UAB held back by weak backhaul links.} Fig.~\ref{fig:plot89} shows the comparison of 2 hop rates with different TDD schemes. As expected from our observations in Section~\ref{sec:resultsMBSonly}, for $\eta=0.1,0.9$ dynamic TDD provides about $1.5\times$ gains for DL/UL compared to static TDD for an optimal $\delta$ chosen for each scheme. For $\eta = 0.5$, the gains with dynamic TDD are completely overshadowed by weak backhaul links for the optimum choice of $\delta$ but  $20-30\%$ gains are visible for non-optimal $\delta$ lower than the optimal. Note that choosing a $\delta$ higher than optimum gives same rate as static TDD since the network is backhaul limited and this is the backhaul rate on the 2 hop link. This is clearer looking at the access and backhaul rates separately for DL UEs operating on 2 hops, as shown in Fig.~\ref{fig:plot17}. Another observation from Fig.~\ref{fig:plot89} is that {\em the optimal $\delta$ with dynamic TDD and UAB is lower or the same as compared to static TDD with SAB}. The reason is that both dynamic TDD and UAB boost access rates for a fixed $\delta$ (see Fig.~\ref{fig:plot17}) and thus can allow providing more backhaul slots in a frame still being able to achieve higher 2 hop rates. Fig.~\ref{fig:plot17} also shows a potential of up to $2-5\times$ gains in DL rates with UAB for $\eta = 0.5$ and different $\delta$, but the gains are held back by weak backhaul links. 

{\bf UAB gains are not limited to asymmetric traffic.} Fig.~\ref{fig:plot89} shows that with UAB, unlike dynamic TDD, about $30\%$ gains are still observed in UL 2 hop rates for $\eta=0.5$. The gains with DL are only $10\%$ since due to increasing interference, $p_{dl}=1$ is not optimal as seen from Fig.~\ref{fig:plot10}. Also, since DL access rates are closer to backhaul rates due to higher transmit power compared to UL, the network is even more backhaul-limited from DL UE perspective.

{\bf Consistent $30\%$ gains in mean rates across all traffic scenarios with dynamic TDD + UAB in a noise-limited scenario.}
Finally, shifting our focus back to the 73 GHz network mentioned before, which had stronger backhaul links, we can see in Fig.~\ref{fig:plot_1112} that employing dynamic TDD with UAB can offer a uniform $30\%$ gain in  UL/DL mean rate over static TDD with SAB for all traffic scenarios captured by $\eta$. With no UE antenna gain, these gains are expected to be even higher as the access links become much weaker than backhaul. {\em In conclusion, one can harness the gains from dynamic TDD and UAB only if backhaul links are strong enough.} In the future, it would be desirable to develop analytical models that allow different antenna gains and path loss for backhaul links which would likely make dynamic TDD and UAB appear more favourable. 
\begin{figure*}
  \centering
\subfloat[Optimal $p_{dl}$ is lower for higher $\eta$.]
{\label{fig:plot10}{\includegraphics[width=\columnwidth]{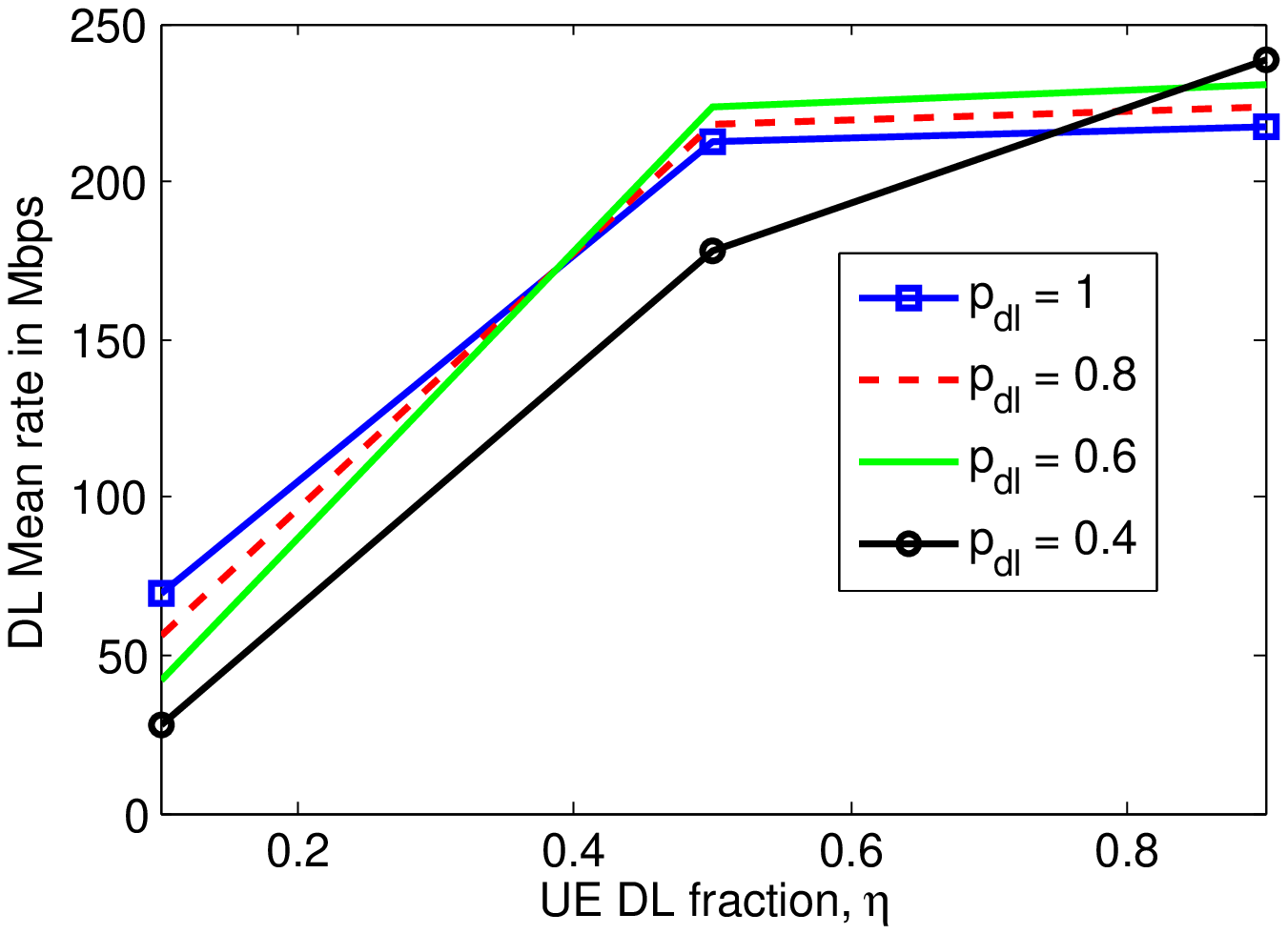}}}
\subfloat[`A' for access and `B' for backhaul. $p_{dl}=1$, $\eta = 0.5$.]
{\label{fig:plot17}{\includegraphics[width= \columnwidth]{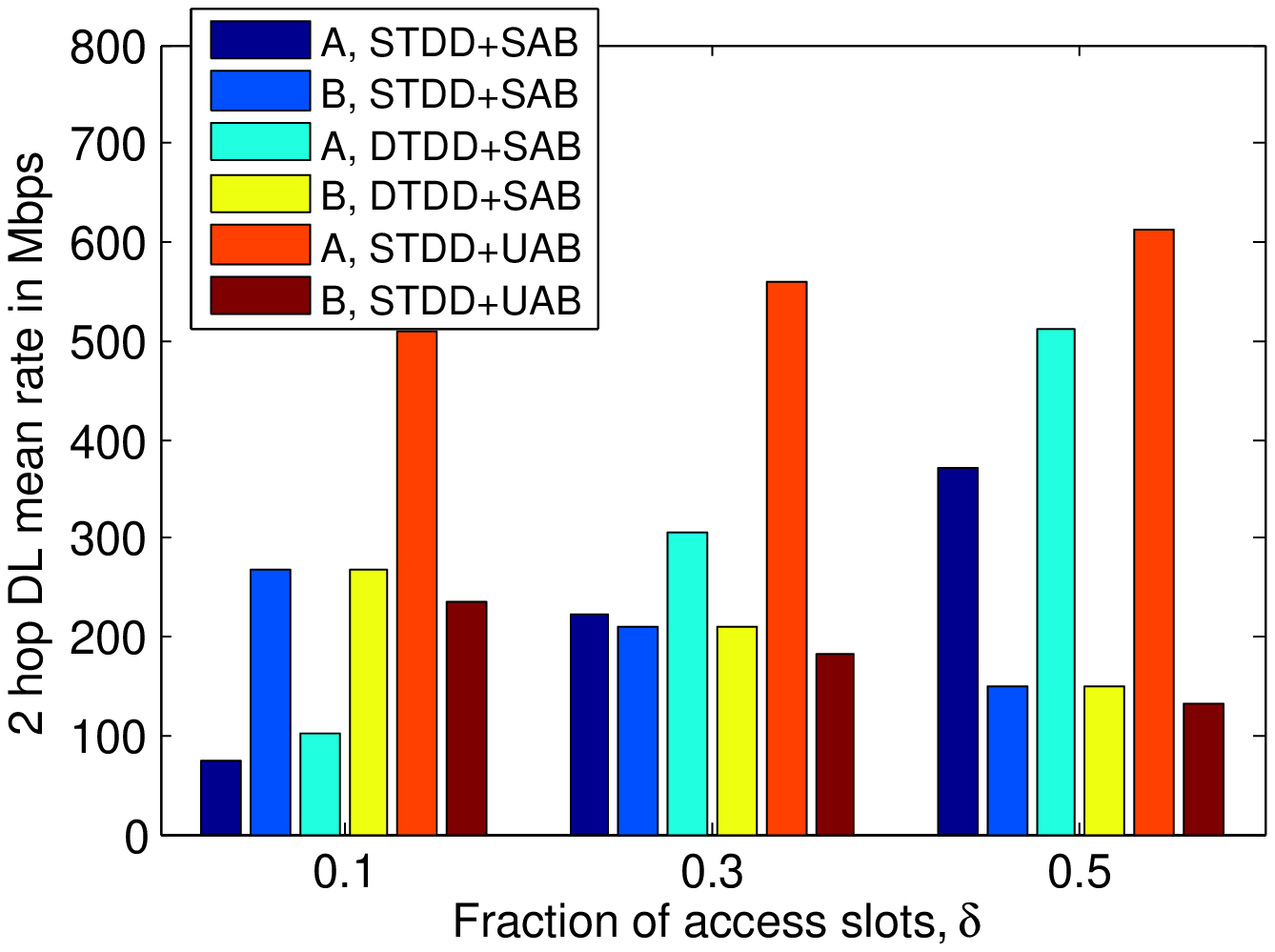}}}
\caption{DL mean rates conditioned that UE connects to SBS.}
 \label{fig:plot1017}
\end{figure*}
\begin{figure*}
  \centering
\subfloat[Downlink]
{\label{fig:plot_11}{\includegraphics[width=\columnwidth]{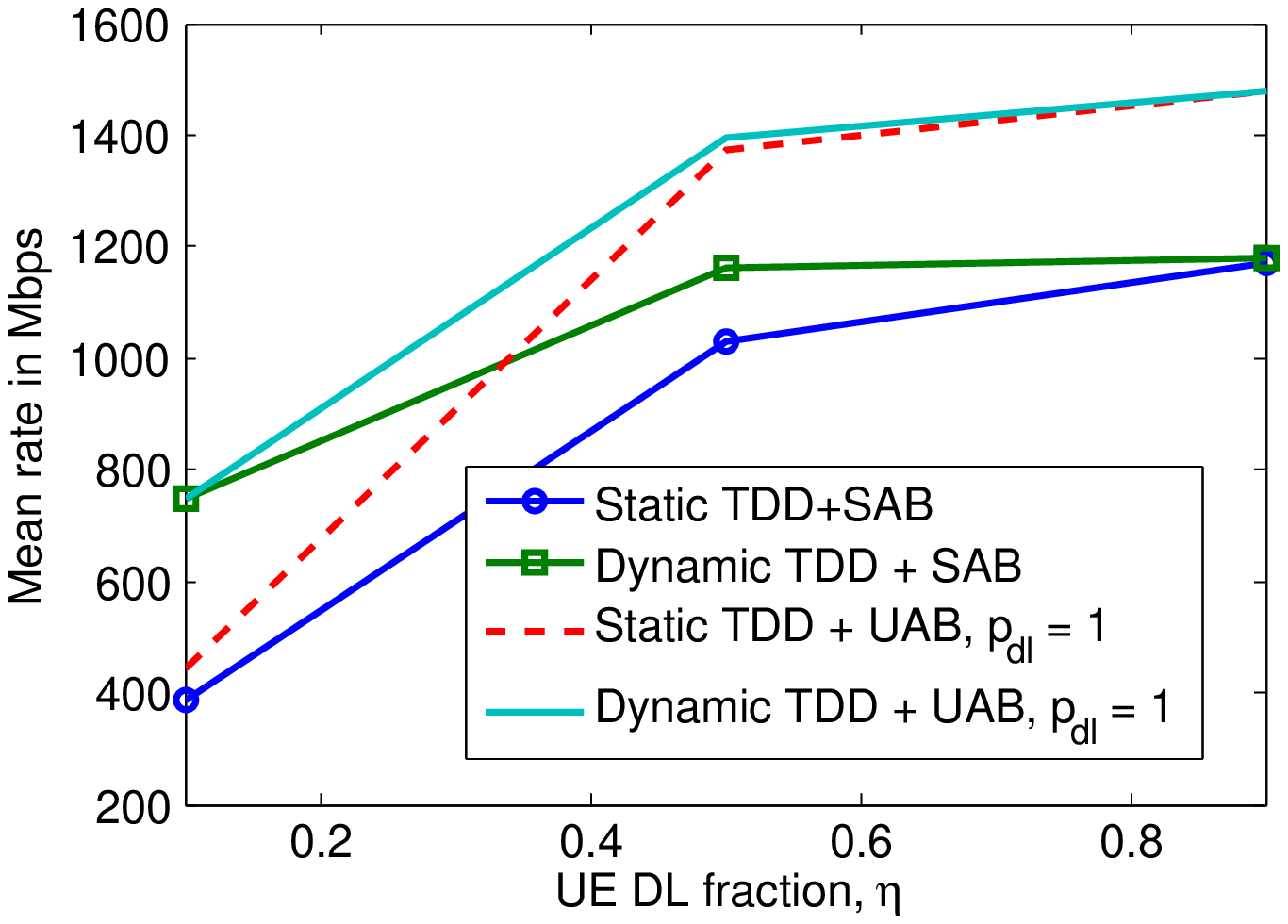}}}
\subfloat[Uplink]
{\label{fig:plot_12}{\includegraphics[width= \columnwidth]{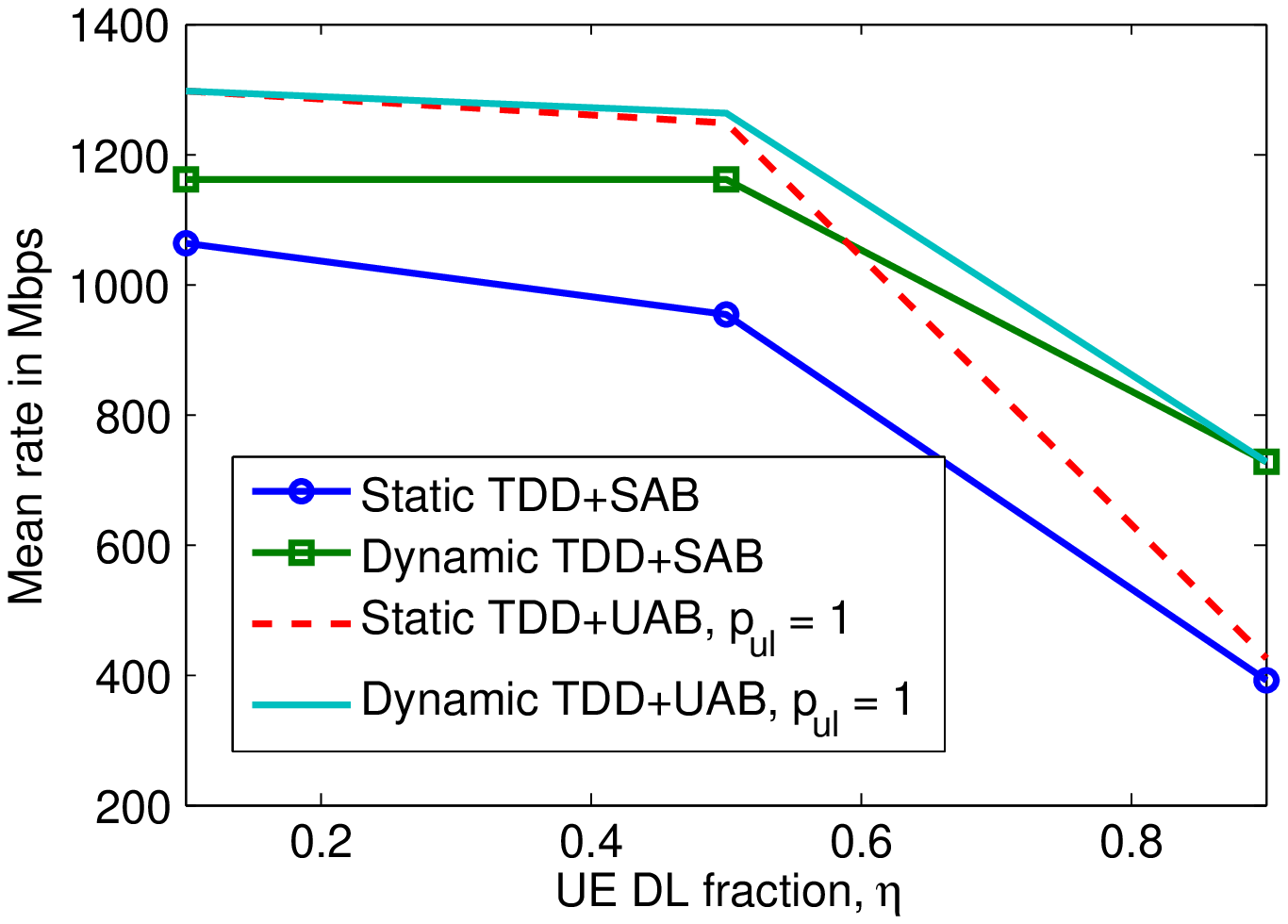}}}
\caption{Dynamic TDD with UAB gives $30\%$ gains over Static TDD with SAB in the noise-limited scenario at 73 GHz.}
 \label{fig:plot_1112}
\end{figure*}

\section{Conclusions}
\label{sec:conclusion}
This is the first comprehensive study of UL-DL SINR distribution and mean rates in dynamic TDD enabled mmWave cellular networks. A key analytical takeaway is how to explicitly incorporate TDD frame structures for resource allocation studies in self-backhauled cellular networks using stochastic geometry. Computing approximate yet fairly accurate Laplace transform of new types of interference that arise while studying dynamic TDD and UAB is another takeaway with variety of applications. It can be useful to study co-existence of device-to-device/Internet-of-Things applications with cellular networks, wherein unscheduled UEs operate on the same band but for non-cellular purposes.

From a system insights viewpoint, the key takeaways lie in the comparison of different TDD schemes as a function of different access-backhaul splits, UL/DL traffic asymmetry and the density of BSs. Dynamic TDD and UAB are intriguing as they address some key fallacies with conventional static TDD and SAB implementations, as highlighted in this work, and it is worth noting that these are in fact a class of scheduling policies. We expose the pros and cons of our heuristic implementations using the derived formulae under various network settings, and the observations arouse interest in their further investigation with more sophisticated traffic models, implementation of self-backhauling with much {\em stronger} backhaul links than the access links and more realistic deployment and propagation assumptions. In the future, several variations of the class of scheduling policies considered in this work can be studied. Extending the analysis to more than 2 hops is non trivial but desirable considering the recent interest to enable self-backhauling with as low an MBS density as possible. 

\begin{appendix}
\section{Appendices}
\subsection{Association probabilities}
\label{sec:App_assocprob}
From Lemma 1 in \cite{Shaer16}, for $t\in\{m,s\}$ the CCDF of $\min_{X\in\Phi_t} L(X,0)$ is given by $V_t(\tau) = \mathbb{P}\left(\min_{X\in\Phi_t} L(X,0)> \tau\right) =\exp\left(-\Lambda_t(\tau)\right)$, where 
\begin{multline}
\label{eq:fcap_s}
\Lambda_t(\tau)=\pi \lambda_t \Big( \left(\plos \tau^{\frac{2}{\alpha_l}} + (1-\plos)\tau^{\frac{2}{\alpha_n}}\right)\\\mathds{1}(\tau < \dlos^{\alpha_l})+ \tau^{\frac{2}{\alpha_n}} \mathds{1} (\tau > \dlos^{\alpha_n}) +  \\\left(\plos \dlos^2 + (1-\plos) \tau^{\frac{2}{\alpha_n}} \right)\mathds{1}( \dlos^{\alpha_l} \leq \tau \leq \dlos^{\alpha_n}) \Big).
\end{multline}
Here, $\Lambda_t(\tau)$ is the intensity of the propagation process $\{L(X,0): X\in\Phi_t\}$. The PDF of $\min_{X\in\Phi_t} L(X,0)$ is given by 
$v_t(\tau) = \frac{\mathrm{d}\Lambda_t(\tau)}{\mathrm{d}\tau} \exp\left(-\Lambda_t(\tau)\right),$
where $\frac{\mathrm{d}\Lambda_t(\tau)}{\mathrm{d}\tau} =$ 
\begin{multline}
\label{eq:LambdadT}
\frac{2 \pi \lambda_t \tau^{\frac{2}{\alpha_n}- 1}}{\alpha_n} \Bigg(\bigg(\frac{\alpha_n \plos \tau^{\frac{2}{\alpha_l} - \frac{2}{\alpha_n}}}{\alpha_l}  + 1-\plos\bigg)\mathds{1}(\tau < \dlos^{\alpha_l})\\
+(1-\plos)\mathds{1}( \dlos^{\alpha_l} \leq \tau \leq \dlos^{\alpha_n})+\mathds{1} (\tau > \dlos^{\alpha_n})\Bigg).
\end{multline} 
Define, $\Lambda_t(\mathrm{d}\tau)= \frac{\mathrm{d}\Lambda_t(T)}{\mathrm{d}T}\big|_{T=\tau} \mathrm{d}\tau$ which will be useful in the Appendix~\ref{sec:App_UL_Laplace}. The probability that a typical user at origin associates with a MBS is given by $\mathcal{A}_m  = $
\begin{align*}
&\mathbb{P}\left(\max_{X\in\MBSs} \power{m}
L(X,0)^{-1} G_{m} B_{m} > \max_{Y\in\SBSs} \power{s}
L(Y,0)^{-1} G_{s} B_{s}\right)\\
& = \int_{0}^{\infty} V_s\left(\frac{\power{s} G_{s} B_{s}\tau}{\power{m} G_{m} B_{m}}\right) v_m(\tau)\mathrm{d}\tau.
\end{align*}
If $\power{s} G_{s} B_{s}=\power{m} G_{m} B_{m}$,  $\mathcal{A}_m = \lambda_m/\lambda_b$. 
\subsection{Proof of Lemma~\ref{lem:Fad}}
\label{appendix:subframe}
The CDF of $\gamma_{a,D,X}\F_a$ is derived as follows.  $\mathbb{P}\left(\gamma_{a,D,X}\F_a>r\big| \F_a\right) = \mathbb{P}\left(\mathds{1}(N_{d,X}>0)\frac{N_{d,X}\F_a}{N_{d,X}+N_{u,X}}>r\big| \F_a\right)=p_1(r),$ which is computed using Assumption~\ref{assumption:2}. Similarly, $\mathbb{P}\left(\gamma_{a,D,X}\F_a\geq r\big| \F_a\right)=p_2(r)$ is derived. Now, let us denote $\tilde{\gamma}_{a,D,X} = \gamma_{a,D,X} \F_a - \lfloor\gamma_{a,D,X} \F_a\rfloor$. Thus,  
\begin{align*}
&\mathbb{P}\left(\F_{ad,D,X} = n\big| \F_a\right)= \mathbb{E}\left[\tilde{\gamma}_{a,D,X}\mathds{1}(\lceil\gamma_{a,D,X} \F_a\rceil = n)\right. \\
&\left. + \left(1-\tilde{\gamma}_{a,D,X}\right)\mathds{1}(\lfloor\gamma_{a,D,X} \F_a\rfloor = n)\big| \F_a\right] \\
& = \mathbb{E}\left[\tilde{\gamma}_{a,D,X}\mathds{1}(n-1<\gamma_{a,D,X}\F_a\leq n) \right.\\
&\left.+\left(1-\tilde{\gamma}_{a,D,X}\right) \mathds{1}(n\leq \gamma_{a,D,X}\F_a<n+1)\big| \F_a\right] \triangleq \mathbb{E}\left[\Xi\big| \F_a\right].
\end{align*} 
Since, $1\geq \Xi\geq 0$ the expectation can be computed as 
$\mathbb{E}\left[\Xi\big|\F_a\right] = \int_{0}^{1}\mathbb{P}\left(\Xi>r\big|\F_a\right)\mathrm{d}r.$ For $r = 1$, the probability inside the integral is zero and for $r<1$,
\begin{align*}
&\mathbb{P}\left(\Xi>r\big|\F_a\right) = \mathbb{P}\left(n+r-1<\gamma_{a,D,X}\F_a<n+1-r\big|\F_a\right)\\& = p_1(n+r-1)-p_2(n+1-r).
\end{align*}

\subsection{Laplace functional of interference for computing access UL SINR}
\label{sec:App_UL_Laplace}
{\bf Approximation 1:} Interference from MBS, SBS and UE is assumed independent of each other. Thus,  $L^{ul,a,t,\mu}_{i,w}(\s,R)\approx$
\begin{align*}
&\prod_{\nu\in\{m,s,u\}} \mathbb{E}\left[\exp\left(-\s I_{i,\nu,w}(X^*)\right)\Big| X^* \in \Phi_{t,\mu}, ||X^*|| = R, \mathcal{F}\right] \\&= L_m L_s L_u. 
\end{align*}
\subsubsection{$i\leq\F_a$}
\paragraph{Interference from MBSs and SBSs} This is non-zero only with dynamic TDD for access subframe. For $\nu\in\{m,s\}$ the Laplace transform can be simplified as follows. By superposition of PPPs, $\Phi_\nu = \Phi_{\nu, l}+ \Phi_{\nu, n}$, wherein both the child processes are independent non-homogeneous PPPs with intensities $\lambda_\nu \plos\mathds{1}(x\leq\dlos) $ and $\lambda_\nu (1-\plos\mathds{1}(x\leq\dlos)) $.  Further, by strong Markov property of PPPs, replacing the shot noise of interference by that from independent copies of the PPPs, \eqref{eq:longequationappendix1} is derived, 
\begin{figure*}[!t]
\begin{align}
\label{eq:longequationappendix1}
\nonumber L_{\nu}
&=\mathbb{E}\left[\exp\left(-\s \sum_{\mu_1\in\{l,n\}}\sum_{Y\in\Phi_{\nu,\mu_1}}  \mathds{1}(i\leq \F_{ad,w,Y}, N_{d,Y}>0) \mathds{1}\left(||Y||^{\alpha_{\mu_1}}>R^{\alpha_\mu}\frac{\power{\nu} G_{\nu} B_{\nu}}{\power{t} G_t B_t}\right) \mathrm{C}_0 \power{\nu} h_{i,X^*,Y} G_{i,X^*,Y} \right.\right.\\
\nonumber&\left.\Bigg.\hspace{1cm}L(X^*,Y)^{-1}\Bigg)\Bigg| X^*\in\Phi_{t,\mu}, ||X^*|| = R, \mathcal{F}\right]\\
\nonumber&=\mathbb{E}\left[\exp\left(-\s \sum_{\mu_1\in\{l,n\}}\sum_{\mu_2\in\{l,n\}}\sum_{Y\in {\Phi}_{\nu,\mu_1,\mu_2}} \mathds{1}(i\leq \F_{ad,w,Y}) \mathds{1}\left(||Y||^{\alpha_{\mu_1}}>R^{\alpha_\mu}\frac{\power{\nu} G_{\nu} B_{\nu}}{\power{t} G_t B_t}\right) \mathds{1}(N_{d,Y}>0) \mathrm{C}_0 \power{\nu} h_{i,X^*,Y}\right.\right.\\
&\hspace{1cm}\left.\Bigg. G_{i,X^*,Y} ||X^*-Y||^{-\alpha_{\mu_2}}\Bigg)\Bigg| X^*\in\Phi_{t,\mu}, ||X^*|| = R, \mathcal{F}\right].
\end{align}
\hrulefill
\end{figure*}
where $\Phi_{\nu,\mu_1,\mu_2}$ are BSs of tier $\nu$ which have type $\mu_1\in\{l,n\}$ links to the origin and type $\mu_2\in\{l,n\}$ links to $X^*$. Given, $||X^*||=R$, $\Phi_{\nu,\mu_1,\mu_2}$ is a PPP with density  
$\hat{\lambda}_{\nu,\mu_1,\mu_2}(r,\theta)=\lambda_\nu p_{\mu_1}(r)p_{\mu_2}\left(\sqrt{r^2+R^2-2rR\cos(\theta)}\right)$. Further simplifying, the above expression is equal to 
\begin{align*}
& \prod_{\mu_1,\mu_2} \exp\left(- \int^{\infty}_{\left(\frac{R^{\alpha_\mu}\power{\nu} G_{\nu} B_{\nu}}{\power{t} G_t B_t}\right)^{1/\alpha_{\mu_1}}}\right.\\&\left.
\int^{2\pi}_{0}\mathbb{E}\left[\frac{ p_{i,w,\nu} \hat{\lambda}_{\nu,\mu_1,\mu_2}(r,\theta) r}{1+\frac{\left(r^2+R^2-2rR\cos(\theta)\right)^{\alpha_{\mu_2}/2}}{\s \mathrm{C}_0\power{\nu} \Psi_{t, \nu}}}\right] \mathrm{d}r\mathrm{d}\theta\right),
\end{align*}
where 
\begin{equation}
\label{eq:piwd}
p_{i,w,\nu} = \mathbb{P}\left(N_{d}>0, i\leq \F_{ad,w}\big| \mathcal{F}\right) = \sum_{n=i}^{\F_a}\mathbb{P}\left(\F_{ad,D}=n\big| \mathcal{F}\right),
\end{equation}
which can be computed using Lemma~\ref{lem:Fad}. 

Note that the lower limit of integral on $r$ is exactly the value of $\s$ from \eqref{eq:ULSINRdist_access}. Thus, rewriting the equation with change of variables  $\rho~=~r\left(\frac{R^{\alpha_\mu}\power{\nu} G_{\nu} B_{\nu}}{\power{t} G_t B_t}\right)^{-1/\alpha_{\mu_1}} $ 
is easier to implement on MATLAB. An even easier implementation, which is in fact a lower bound to the Laplace functional, can be obtained by neglecting the $\mathds{1}(||Y||^{\alpha_{\mu_1}}>(.))$ term in the above derivation, which gives lower bound in Lemma~\ref{lem:LaplaceULaccess}. 

\paragraph{Interference from UEs} $\mathbb{E}[\exp\left(-\s I_{i,u,w}\right)\big| X^*\in\Phi_{t,\mu}, ||X^*|| = R,  \mathcal{F}]$ can be computed using a non-homogeneous PPP approximation inspired from \cite{SinZhaAnd15}.  \setcounter{approximation}{1}
\begin{approximation} 
Laplace transform of interference from scheduled device process ($\Phi_1$) connected to a PPP BS process ($\Phi_2$) to a receiver under consideration is approximated by that generated from an independent PPP device process $\Phi_3$ with same intensity as $\Phi_2$. Further thinning is done $\Phi_3$ to approximate the pair correlation function by taking into consideration the association of points in $\Phi_3$ to those in $\Phi_2$\cite{SinZhaAnd15}. 
\end{approximation}
Thus, conditioned on the event that the tagged BS $X^*$ is of tier $t$, the propagation process of interfering UEs is approximately equal in distribution to an independent non-homogeneous PPP on $\mathbb{R}^+$ with intensity \begin{multline}
\label{eq:piwu}
\Lambda(t,\mathrm{d}r) = \sum_{k\in\{m,s\}} \left(1-\exp\left(-\Lambda_k\left(r\frac{\power{k}B_k G_k}{\power{t}B_t G_t}\right)\right)\right)\\\times p_{i,w,k} \Lambda_k(\mathrm{d}r),
\end{multline}
with $\Lambda_k(\mathrm{d}r) = \frac{\mathrm{d}\Lambda_k(x)}{\mathrm{d}x}\Big|_r \mathrm{d}r$, and  $p_{i,w,k} = \mathbb{P}\left(N_u>0, \F_{ad,w}<i\leq \F_a\big| \mathcal{F}\right)$. 
Note that $p_{i,w,k}$ is captures the active probability of interferer in the $i^\text{th}$ slot and the non-idle probability of parent BS process. The $1-\exp\left(.\right)$ term ensures that the biased received power from at least one of the points in $\Phi_k$ is better than that from the BS at $X^*$\cite{SinZhaAnd15}.
Thus,
\begin{equation}
\label{eq:ulintue}
L_u \approx\exp\left(-\int_{0}^{\infty}\mathbb{E}\left[\frac{1}{1+ \frac{r}{\s\mathrm{C}_0\power{u}\Psi_{t,u}}}\right]\Lambda(t,\mathrm{d}r)\right).
\end{equation}
Here, $p_{i,S,k} = \mathbb{P}\left(N_{u,X}>0, \F_{ad,S,X}<i\leq \F_a\big| \mathcal{F}\right) = \left(1-\left(1+\frac{\lambda_u \mathcal{A}_k(1-\eta)}{3.5\lambda_k}\right)^{-3.5} \right)\mathds{1}\left(\F_{ad}<i\leq \F_a \right).$ Since, an UL UE is only scheduled in access subframe for $\F_{ad}<i\leq \F_a$ with static TDD, the indicator in previous expression will always be 1 for feasible UL access SINR distributions. Similarly, $p_{i,D,k}$
\begin{align*}
& =\mathbb{P}\left(\F_{ad,w}<i\leq \F_a\big| \mathcal{F}\right) - \mathbb{P}(N_u = 0, \F_{ad,w}<i\leq \F_a| \mathcal{F}) \\&=  \mathbb{P}\left(\F_{ad,w}<i\leq \F_a\big| \mathcal{F}\right) - \mathbb{P}(N_u = 0\big| \mathcal{F} )
\end{align*}
The first term can be found by substituting $t=k$ in Lemma~\ref{lem:Fad} and the second term is $\left(1+\frac{\lambda_u\mathcal{A}_k (1-\eta)}{3.5\lambda_k}\right)^{-3.5}$.
\subsubsection{$i>\F_a$ and $w = \text{UAB}$}
Note that if we are computing Laplace functional of interference at an UL receiver of an access link for $i>\F_a$, by definition we are operating in $w=\text{UAB}$ mode with $X^*\in\Phi_s$. 
In this case there is no interference from MBSs. 
\paragraph{Interference from SBSs}
The interference from SBSs can be computed similar to the previous case on interfering UEs with $i<\F_a$. However, we need to incorporate the fact that the MBS serving $X^*$ has an interfering SBS scheduled with probability 1 but other MBSs may not have a scheduled SBS with probability $p_{i,w,s} = \left(1+\frac{\lambda_{s,u}}{3.5\lambda_m}\right)^{-3.5}$ with $\lambda_{s,u} = \lambda_s\left(1-\left(1+\frac{\mathcal{A}_s \lambda_u (1-\eta)}{3.5\lambda_s}\right)^{-3.5} \right)$. Thus, the following version of approx. 2 is employed. The point closest to $X^*$ in the new interfering PPP is active with probability 1 and rest of the points are active with probability $p_{i,w,s}$. This gives the corresponding expression in Lemma~\ref{lem:LaplaceULaccess2}.

\paragraph{Interference from UEs}
By approximation 2, the interfering PPP process has intensity equal to $\lambda_s$. A further thinning by $\frac{\hat{\lambda}}{\lambda_s}$ is done, where $\hat{\lambda} =$
\begin{multline*}
 \mathds{1}\left(\F_a+\F_{bd}<i\leq \F\right)\left(1-\left(1+\frac{\lambda_u (1-\eta)\mathcal{A}_s
}{3.5\lambda_s}\right)^{-3.5} \right)\\ \times p_{ul}\left(\lambda_s-\left(1-\left(1+\frac{\lambda_s }{3.5\lambda_m}\right)^{-3.5} \right)\lambda_m\right)^+ , 
\end{multline*}
where $a^+ = a$ if $a>0$ and zero otherwise. This captures that there will be at most 1 scheduled UE from every SBS with poaching probability $p_{ul}$ 
except those SBSs which are scheduled by their serving MBS. Thus, the Laplace functional is same as \eqref{eq:ulintue} but with $\Lambda(t,\mathrm{d}r)$ replaced by $\frac{\hat{\lambda}}{\lambda_s}\left(1-\exp\left(-\Lambda_s(r)\right)\right)\Lambda_s(\mathrm{d}r)$, where 1-exp(.) accounts for the probability that the interfering UEs don't associate with the  SBS at $X^*$.

\subsection{Uplink mean rate}
\label{App_ULrate}
In the following derivation of UL mean rate, $w_a\in\{S,D\}$ and $w_b\in\{\mathrm{SAB},\mathrm{UAB}\}$. 
\begin{align*}
&\mathrm{R}_{ul,m,w_a} =\frac{\mathbb{E}\left[\mathrm{D}_{ul,m,w_a}\big| \mathcal{E}_m\right]}{\mathrm{T}\F}\\& = \frac{\BW}{\F}  \mathbb{E}\left[\sum_{i=1+\F_{ad,w_a,X^*}}^{\F_{a}}\mathds{1}\left(\text{UE scheduled in $i^\text{th}$ slot}\right)\right.\\&\times\left.\log_2\left(1+\mathsf{SINR}^{ul}_{i,a,w_a}\right)\Bigg| \mathcal{E}_m\right]=\frac{\BW}{\F}  \mathbb{E}\left[\sum_{i=1+\F_{ad,w_a,X^*}}^{\F_{a}}\right.\\
&\left. \frac{\mathbb{E}\left[\log_2\left(1+\mathsf{SINR}^{ul}_{i,a,w_a}\right)\Big| \F_a, N_{u,X^*},N_{d,X^*},  \mathcal{E}_m\right]}{N_{u,X^*}}\Bigg| \mathcal{E}_m\right]\\
&= \frac{\BW}{\F}  \mathbb{E}\left[\sum_{i=1+\F_{ad,w_a,X^*}}^{\F_{a}}\frac{1}{N_{u,X^*}}\int_{0}^{\infty}\frac{\mathsf{S}^{ul,m}_{i,a,w_a}(\tau)}{1+\tau}\mathrm{d}\tau\Bigg| \mathcal{E}_m\right], 
\end{align*}
where distribution of $\F_{ad,D,X^*}$ given $\gamma_{a,D,X} = \frac{n_2}{n_1+n_2+1}$ is given by \eqref{eq:FadS}. Similarly, given the constant $\gamma_a$ the distribution of $\F_{ad,S,X^*}$ can also be found from \eqref{eq:FadS}. To compute $\mathrm{R}_{ul,s,w_a,w_b}$, let us look at each of the expectations inside the minimum one by one. 
\begin{align*}
&\mathbb{E}\left[\mathrm{D}_{ul,s,a,w_a,w_b}\big| \mathcal{E}_s\right] \\
&= \bwidth \mathrm{T} \mathbb{E}\left[\sum_{i=1+\F_{ad,w_a,X ^*}}^{\F_a}\frac{1}{N_{u,X^*}}\int_{0}^{\infty}\frac{\mathsf{S}^{ul,s}_{i,a,w_a}(\tau)}{1+\tau}\mathrm{d}\tau \Bigg| \mathcal{E}_s\right]\\
&+\mathds{1}(w_b= \mathrm{UAB})\bwidth \mathrm{T}\\ &\mathbb{E}\left[\sum_{i=1+\F_a+\F_{bd}}^{\F}\left(1-\frac{1}{N_{s,X^{**}}}\right)
\frac{p_{ul}}{N_{u,X^*}}\int_{0}^{\infty}\frac{\mathsf{S}^{ul,s}_{i,a,w_b}(\tau)}{1+\tau}\mathrm{d}\tau \Bigg| \mathcal{E}_s\right], 
\end{align*}
where $N_{s,X^{**}}$ is the number of SBSs associated with $X^{**}$ with at least one UL UE. Similarly, 
\begin{align*}
&\mathbb{E}\left[\mathrm{D}_{ul,s,b,w_b}\big| \mathcal{E}_s\right]= \bwidth \mathrm{T} \mathbb{E}\left[1/N_{s,X^{**}}\right] \sum_{n=1}^{\infty}\frac{\kappa^*_{s,ul}(n)}{n} \\&\hspace{1cm}\times\mathbb{E}_{\mathcal{F}} \int_{0}^{\infty}\frac{\sum_{i=1+\F_{a}+\F_{bd}}^{\F}\mathsf{S}^{ul,s}_{i,b,w_b}(\tau)}{1+\tau}\mathrm{d}\tau.
\end{align*}

\subsection{Laplace functional of interference for access DL SINR}
\label{sec:App_DL_Laplace}
The main difference with UL case is that now the receiver is at origin instead of at $X^*$. Thus, different exclusion regions need to be considered while computing the shot noise. By approximation 1, 
\begin{align*}
&L^{dl,a,t,\mu}_{i,w}(\s,R) \\&\approx \prod_{\nu\in\{m,s,u\}} \mathbb{E}\left[\exp\left(-\s I_{i,\nu,w}(0)\right)\Big| X^* \in \Phi_{t,\mu}, ||X^*|| = R, \mathcal{F}\right].
\end{align*}
\paragraph{$i\leq \F_a$}
For $\nu\in\{m,s\}$,
\begin{align*}
&\mathbb{E}\left[\exp\left(-\s I_{i,\nu,w}(0)\right)\Big| X^* \in \Phi_{t,\mu}, ||X^*|| = R, \mathcal{F}\right]\\
&=\exp\left(-\int_{R^{\alpha_\mu}}^{\infty}\mathbb{E}\left[\frac{1}{1+\frac{r}{\s \mathrm{C}_0 \power{\nu}\Psi_{t,u}}}\right]p_{i,w,d}\Lambda_{\nu}(\mathrm{d}r)\right),
\end{align*}
where $p_{i,w,\nu}$ is given \eqref{eq:piwd} and $\Lambda_\nu(\mathrm{d}r)$ was defined in Appendix~\ref{sec:App_assocprob}.  Note that this is exact expression.

For $\nu=u$, there will non-zero interference only with dynamic TDD. By approximation 2, we compute the Laplace functional of interference from UEs is generated from two independent PPPs -- for SBS/MBS connection -- as follows. 
\begin{align*}
&\mathbb{E}\left[\exp\left(-\s I_{i,u,w}(0)\right)\Big| X^* \in \Phi_{t,\mu}, ||X^*|| = R, \mathcal{F}\right]\\
&\approx\exp\left(-\int_{0}^{\infty}\mathbb{E}\left[\frac{1}{1+\frac{r}{\s \mathrm{C}_0 \power{u}\Psi_{u,u}}}\right]\sum_{k\in\{m,s\}}p_{i,w,k}\Lambda_{k}(\mathrm{d}r)\right),
\end{align*}
where $p_{i,w,k}$ can be found just after \eqref{eq:ulintue}. 
\paragraph{$i>\F_a$}
In backhaul subframe, a DL UE is scheduled for access only if $\F_a<i\leq \F_a+\F_{bd}$, $w=$UAB and the UE connects to a SBS.  Thus, there is interference only from MBSs and SBSs. 
\begin{align*}
&\mathbb{E}\left[\exp\left(-\s I_{i,m,w}(0)\right)\Big| X^* \in \Phi_{s,\mu}, ||X^*|| = R, \mathcal{F}\right]\\
&=\exp\left(-\int_{R^{\alpha_\mu}}^{\infty}\mathbb{E}\left[\frac{1}{1+\frac{r}{\s \mathrm{C}_0 \power{m}\Psi_{m,u}}}\right]p_{i,w,m}\Lambda_{m}(\mathrm{d}r)\right),
\end{align*}
where $$p_{i,w,m} = \mathds{1}(\F_a<i\leq \F_a+\F_{bd})\left(1-\left(1+\frac{\lambda_{s,d}}{3.5\lambda_m}\right)^{-3.5}\right)$$ with $\lambda_{s,d} = \lambda_s\left(1-\left(1+\frac{\lambda_u \eta\mathcal{A}_s}{3.5\lambda_s}\right)^{-3.5}\right)$. 

To compute $\mathbb{E}\left[e^{-\s I_{i,s,w}(0)}\Big| X^* \in \Phi_{s,\mu}, ||X^*|| = R, \mathcal{F}\right]$, we make the following approximation similar to the corresponding UL case for poaching. The SBS interferers form an independent homogeneous \PPP $ $ with density given by 
\begin{multline*}
\hat{\lambda}_{d} =\left(\lambda_s-\left(1-\left(1+\frac{\lambda_s}{3.5\lambda_m}\right)^{-3.5}\right)\lambda_m\right)^+\times \\
p_{dl}\mathds{1}\left(\F_a<i\leq \F_a+\F_{bd}\right)\left(1-\left(1+\frac{\lambda_u \eta\mathcal{A}_s}{3.5\lambda_s}\right)^{-3.5}\right).
\end{multline*}
Thus, we get 
\begin{align*}
&\mathbb{E}\left[\exp\left(-\s I_{i,s,w}(0)\right)\Big| X^* \in \Phi_{s,\mu}, ||X^*|| = R, \mathcal{F}\right]\\
&\approx\exp\left(-\int_{0}^{\infty}\mathbb{E}\left[\frac{1}{1+\frac{r}{\s \mathrm{C}_0 \power{s}\Psi_{s,u}}}\right]\frac{\hat{\lambda}_{d}}{\lambda_s} \Lambda_{s}(\mathrm{d}r)\right).
\end{align*}
\end{appendix}
\section*{Acknowledgment}
The authors wish to thank R. Heath, E. Visotsky, F. Vook, M. Cudak and A. Gupta for valuable advice, encouragement, and discussions.
\bibliographystyle{IEEEtran}
\bibliography{IEEEabrv,Kulkarni}

\begin{thebibliography}{10}
\providecommand{\url}[1]{#1}
\csname url@samestyle\endcsname
\providecommand{\newblock}{\relax}
\providecommand{\bibinfo}[2]{#2}
\providecommand{\BIBentrySTDinterwordspacing}{\spaceskip=0pt\relax}
\providecommand{\BIBentryALTinterwordstretchfactor}{4}
\providecommand{\BIBentryALTinterwordspacing}{\spaceskip=\fontdimen2\font plus
\BIBentryALTinterwordstretchfactor\fontdimen3\font minus
  \fontdimen4\font\relax}
\providecommand{\BIBforeignlanguage}[2]{{%
\expandafter\ifx\csname l@#1\endcsname\relax
\typeout{** WARNING: IEEEtran.bst: No hyphenation pattern has been}%
\typeout{** loaded for the language `#1'. Using the pattern for}%
\typeout{** the default language instead.}%
\else
\language=\csname l@#1\endcsname
\fi
#2}}
\providecommand{\BIBdecl}{\relax}
\BIBdecl

\bibitem{Taori15}
R.~Taori and A.~Sridharan, ``Point-to-multipoint in-band mmwave backhaul for
  {5G} networks,'' \emph{{IEEE} Commun. Mag.}, vol.~53, no.~1, pp. 195--201,
  January 2015.

\bibitem{Ran14}
S.~Rangan, T.~S. Rappaport, and E.~Erkip, ``Millimeter wave cellular wireless
  networks: Potentials and challenges,'' \emph{Proc. {IEEE}}, vol. 102, no.~3,
  pp. 366--385, March 2014.

\bibitem{SinJSAC14}
S.~Singh, M.~N. Kulkarni, A.~Ghosh, and J.~G. Andrews, ``Tractable model for
  rate in self-backhauled millimeter wave cellular networks,'' \emph{{IEEE} J.
  Sel. Areas Commun.}, vol.~33, no.~10, pp. 2196--2211, Oct. 2015.

\bibitem{selfbackpatent06}
S.~Jin, J.~Liu, X.~Leng, and G.~Shen, ``Self-backhaul method and apparatus in
  wireless communication networks,'' U.S. Patent US20\,070\,110\,005 A1, 2007.

\bibitem{Li00}
J.~Li, S.~Farahvash, M.~Kavehrad, and R.~Valenzuela, ``Dynamic {TDD} and fixed
  cellular networks,'' \emph{{IEEE} Commun. Lett.}, vol.~4, no.~7, pp.
  218--220, July 2000.

\bibitem{Shen12}
Z.~Shen, A.~Khoryaev, E.~Eriksson, and X.~Pan, ``Dynamic uplink-downlink
  configuration and interference management in {T}{D}-{L}{T}{E},'' \emph{{IEEE}
  Commun. Mag.}, vol.~50, no.~11, pp. 51--59, Nov. 2012.

\bibitem{Rois15}
J.~Garc{\' i}a-Rois \emph{et~al.}, ``On the analysis of scheduling in dynamic
  duplex multihop mm{W}ave cellular systems,'' \emph{{IEEE} Trans. Wireless
  Commun.}, vol.~14, no.~11, pp. 6028--6042, Nov. 2015.

\bibitem{GuptaKul2016}
A.~Gupta \emph{et~al.}, ``Rate analysis and feasibility of dynamic {TDD} in
  5{G} cellular systems,'' \emph{in Proc. IEEE ICC}, pp. 1--6, 2016.

\bibitem{NokiaTDD}
Nokia. (2015) {5G} radio access system design aspects. Available at:
  https://goo.gl/99POAV.

\bibitem{Eric16}
Ericsson. (2016) 5{G} radio access- capabilities and technologies. Available at
  https://goo.gl/3vsqVy.

\bibitem{Yu15}
B.~Yu, L.~Yang, H.~Ishii, and S.~Mukherjee, ``Dynamic {TDD} support in
  macrocell-assisted small cell architecture,'' \emph{{IEEE} J. Sel. Areas
  Commun.}, vol.~33, no.~6, pp. 1201--1213, June 2015.

\bibitem{Sun15}
H.~Sun, M.~Wildemeersch, M.~Sheng, and T.~Q.~S. Quek, ``{D2D} enhanced
  heterogeneous cellular networks with dynamic {TDD},'' \emph{{IEEE} Trans.
  Wireless Commun.}, vol.~14, no.~8, pp. 4204--4218, Aug 2015.

\bibitem{Dahlman14}
E.~Dahlman, S.~Parkvall, and J.~Sk{\"o}ld, \emph{4{G} {LTE/LTE-A}dvanced for
  Mobile Broadband}, 2nd~ed.\hskip 1em plus 0.5em minus 0.4em\relax Academic
  Press, 2014.

\bibitem{Ghosh14}
A.~Ghosh \emph{et~al.}, ``Millimeter wave enhanced local area systems: A high
  data rate approach for future wireless networks,'' \emph{{IEEE} J. Sel. Areas
  Commun.}, vol.~32, no.~6, pp. 1152--1163, June 2014.

\bibitem{BaiHea14}
T.~Bai and {R. W. Heath}{ Jr.}, ``Coverage and rate analysis for millimeter
  wave cellular networks,'' \emph{{IEEE} Trans. Wireless Commun.}, vol.~14,
  no.~2, pp. 1100--1114, Oct. 2014.

\bibitem{Viswanathan05}
H.~Viswanathan and S.~Mukherjee, ``Performance of cellular networks with relays
  and centralized scheduling,'' \emph{{IEEE} Trans. Wireless Commun.}, vol.~4,
  no.~5, pp. 2318--2328, Sept 2005.

\bibitem{Nazmul17}
M.~N. Islam, S.~Subramanian, and A.~Sampath, ``Integrated access backhaul in
  millimeter wave networks,'' in \emph{IEEE WCNC}, 2017.

\bibitem{GhoshATT16}
\BIBentryALTinterwordspacing
A.~{Ghosh (AT\&T)}, ``Designing ultra-dense networks for 5{G},'' in \emph{Texas
  Wireless Summit}, 2016. [Online]. Available:
  \url{https://youtu.be/Xp2fGEzD5JQ}
\BIBentrySTDinterwordspacing

\bibitem{Lu15}
W.~Lu and M.~D. Renzo, ``Stochastic geometry modeling and system-level analysis
  \& optimization of relay-aided downlink cellular networks,'' \emph{{IEEE}
  Trans. Commun.}, vol.~63, no.~11, pp. 4063--4085, Nov 2015.

\bibitem{Tabassum16}
H.~Tabassum, A.~H. Sakr, and E.~Hossain, ``Analysis of massive {MIMO}-enabled
  downlink wireless backhauling for full-duplex small cells,'' \emph{{IEEE}
  Trans. Commun.}, vol.~64, no.~6, pp. 2354--2369, June 2016.

\bibitem{Sharma16}
A.~Sharma, R.~K. Ganti, and J.~K. Milleth, ``Joint backhaul-access analysis of
  full duplex self-backhauling heterogeneous networks,'' \emph{IEEE Trans.
  Wireless Commun.}, vol.~16, no.~3, pp. 1727--1740, March 2017.

\bibitem{Akd14}
M.~R. Akdeniz, Y.~Liu, M.~Samimi, S.~Sun, S.~Rangan, T.~S. Rappaport, and
  E.~Erkip, ``Millimeter wave channel modeling and cellular capacity
  evaluation,'' \emph{{IEEE} J. Sel. Areas Commun.}, vol.~32, no.~6, pp.
  1164--1179, June 2014.

\bibitem{AndBai16}
J.~G. Andrews, T.~Bai, M.~N. Kulkarni, A.~Alkhateeb, A.~Gupta, and R.~W. Heath,
  ``Modeling and analyzing millimeter wave cellular systems,'' \emph{{IEEE}
  Trans. Commun.}, vol.~65, no.~1, pp. 403 -- 430, Jan. 2017.

\bibitem{Ghosh16}
A.~Ghosh, ``The 5{G} mm{W}ave radio revolution,'' \emph{Microwave Journal},
  vol.~59, no.~9, pp. 22--36, Sep. 2016.

\bibitem{Renzo15}
M.~Di~Renzo, ``Stochastic geometry modeling and analysis of multi-tier
  millimeter wave cellular networks,'' \emph{{IEEE} Trans. Wireless Commun.},
  vol.~14, no.~9, pp. 5038--5057, Sept. 2015.

\bibitem{Kul14}
M.~N. Kulkarni, S.~Singh, and J.~G. Andrews, ``Coverage and rate trends in
  dense urban millimeter wave cellular networks,'' \emph{Proc. IEEE Globecom},
  Dec. 2014.

\bibitem{SinDhiAnd13}
S.~Singh, H.~S. Dhillon, and J.~G. Andrews, ``Offloading in heterogeneous
  networks: Modeling, analysis, and design insights,'' \emph{{IEEE} Trans.
  Wireless Commun.}, vol.~12, no.~5, pp. 2484--2497, May 2013.

\bibitem{Ferenc2007}
J.-S. Ferenc and Z.~N{\'e}da, ``On the size distribution of {P}oisson {V}oronoi
  cells,'' \emph{Physica A: Statistical Mechanics and its Applications}, vol.
  385, no.~2, pp. 518 -- 526, Nov. 2007.

\bibitem{ElSawy16}
H.~ElSawy, A.~Sultan-Salem, M.~S. Alouini, and M.~Z. Win, ``Modeling and
  analysis of cellular networks using stochastic geometry: A tutorial,''
  \emph{IEEE Communication Surveys and Turorials}, vol.~19, no.~1, pp.
  167--203, 2017.

\bibitem{YuKim13}
S.~M. Yu and S.-L. Kim, ``Downlink capacity and base station density in
  cellular networks,'' in \emph{Workshop in Spatial Stochastic Models for
  Wireless Networks}, May 2013.

\bibitem{Renzo16}
M.~D. Renzo, W.~Lu, and P.~Guan, ``The intensity matching approach: A tractable
  stochastic geometry approximation to system-level analysis of cellular
  networks,'' \emph{{IEEE} Trans. Wireless Commun.}, vol.~15, no.~9, pp.
  5963--5983, Sept 2016.

\bibitem{andganbac11}
J.~G. Andrews, F.~Baccelli, and R.~K. Ganti, ``A tractable approach to coverage
  and rate in cellular networks,'' \emph{{IEEE} Trans. Commun.}, vol.~59,
  no.~11, pp. 3122--3134, Nov. 2011.

\bibitem{Suryaprakash15}
V.~Suryaprakash, J.~M{\o}ller, and G.~Fettweis, ``On the modeling and analysis
  of heterogeneous radio access networks using a {P}oisson cluster process,''
  \emph{{IEEE} Trans. Wireless Commun.}, vol.~14, no.~2, pp. 1035--1047, Feb
  2015.

\bibitem{Li15}
Y.~Li, F.~Baccelli, H.~S. Dhillon, and J.~G. Andrews, ``Statistical modeling
  and probabilistic analysis of cellular networks with determinantal point
  processes,'' \emph{{IEEE} Trans. Commun.}, vol.~63, no.~9, pp. 3405--3422,
  Sept 2015.

\bibitem{SinZhaAnd15}
S.~Singh, X.~Zhang, and J.~G. Andrews, ``Joint rate and {SINR} coverage
  analysis for decoupled uplink-downlink biased cell associations in
  {H}et{N}ets,'' \emph{{IEEE} Trans. Wireless Commun.}, vol.~14, no.~10, pp.
  5360--5373, Oct 2015.

\bibitem{ElSawy14}
H.~ElSawy and E.~Hossain, ``On stochastic geometry modeling of cellular uplink
  transmission with truncated channel inversion power control,'' \emph{{IEEE}
  Trans. Wireless Commun.}, vol.~13, no.~8, pp. 4454--4469, Aug 2014.

\bibitem{Lee14}
H.~Lee, Y.~Sang, and K.~Kim, ``On the uplink {SIR} distributions in
  heterogeneous cellular networks,'' \emph{IEEE Commun. Let.,}, vol.~18,
  no.~12, pp. 2145--2148, Dec. 2014.

\bibitem{Hae16}
M.~Haenggi, ``User point processes in cellular networks,'' \emph{IEEE Wireless
  Commun. Lett.}, vol.~6, no.~2, April 2017.

\bibitem{KulCode17}
M.~N. Kulkarni. (2017, Jan.) {MATLAB} codes for self-backhauled mm{W}ave
  cellular networks. Available at: https://goo.gl/VoAjWm.

\bibitem{Shaer16}
H.~E. Shaer \emph{et~al.}, ``Downlink and uplink cell association with
  traditional macrocells and millimeter wave small cells,'' \emph{IEEE Trans.
  Wireless Commun.}, vol.~15, no.~9, pp. 6244--6258, Sept. 2016.

\end{thebibliography}
\begin{IEEEbiography} [{\includegraphics[width=1in,height=1.25in,clip,keepaspectratio]{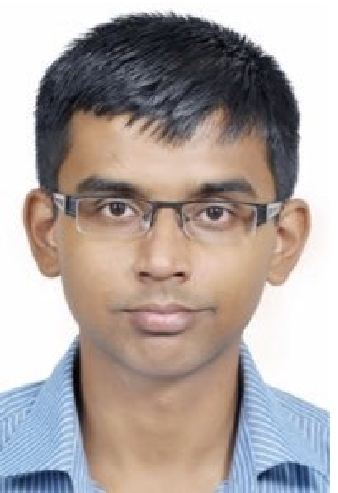}}]
{Mandar~N.~Kulkarni} (S'13) received M.S. in Electrical Engineering from the University of Texas at Austin in 2015 and B.Tech in Electronics and Communications Engineering from the Indian Institute of Technology (IIT) Guwahati in 2013. He is currently working towards Ph.D. in Electrical Engineering at the University of Texas at Austin. His research interests are broadly in the field of wireless communication, with current focus on investigating system design issues in dense urban millimeter wave cellular networks. He has held internship positions at Nokia Bell Labs, Murray Hill, NJ, U.S.A. (2016); Nokia Networks, Arlington Heights, IL, U.S.A.(2014, 2015, 2017); Technical University, Berlin, Germany (2012); and Indian Institute of Science, Bangalore, India (2011). He received the President of India Gold Medal in 2013. 
\end{IEEEbiography}

\begin{IEEEbiography} [{\includegraphics[width=1in,height=1.25in,clip,keepaspectratio]{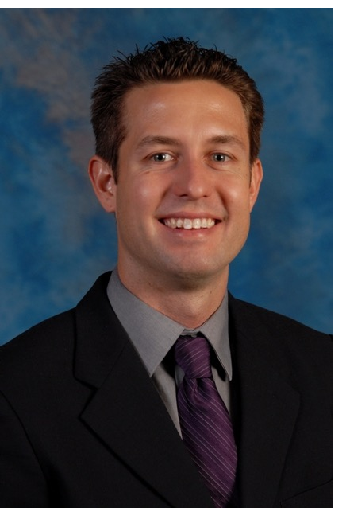}}]
{Jeffrey~G.~Andrews} (S'98$-–$M'02$–-$SM'06$–-$F'13) received the B.S. in Engineering with High Distinction from Harvey Mudd College, and the M.S. and Ph.D. in Electrical Engineering from Stanford University.  He is the Cullen Trust Endowed Professor ($\#$1) of ECE at the University of Texas at Austin.  He was the Editor-in-Chief of the IEEE Transactions on Wireless Communications from 2014-2016. He developed Code Division Multiple Access systems at Qualcomm from 1995-97, and has consulted for entities including Apple, Samsung, Verizon, AT$\&$T, the WiMAX Forum, Intel, Microsoft, Clearwire, Sprint, and NASA.  He is a member of the Technical Advisory Board of Artemis Networks, GenXComm, and Fastback Networks, and co-author of the books Fundamentals of WiMAX (Prentice-Hall, 2007) and Fundamentals of LTE (Prentice-Hall, 2010).  

Dr. Andrews is an ISI Highly Cited Researcher, received the National Science Foundation CAREER award in 2007 and has been co-author of fourteen best paper award recipients including the 2016 IEEE Communications Society $\&$ Information Theory Society Joint Paper Award, the 2011 and 2016 IEEE Heinrich Hertz Prize, the 2014 IEEE Stephen O. Rice Prize, and the 2014 IEEE Leonard G. Abraham Prize.  He received the 2015 Terman Award, is an IEEE Fellow, and is an elected member of the Board of Governors of the IEEE Information Theory Society. 
\end{IEEEbiography}

\begin{IEEEbiography} [{\includegraphics[width=1in,height=1.25in,clip,keepaspectratio]{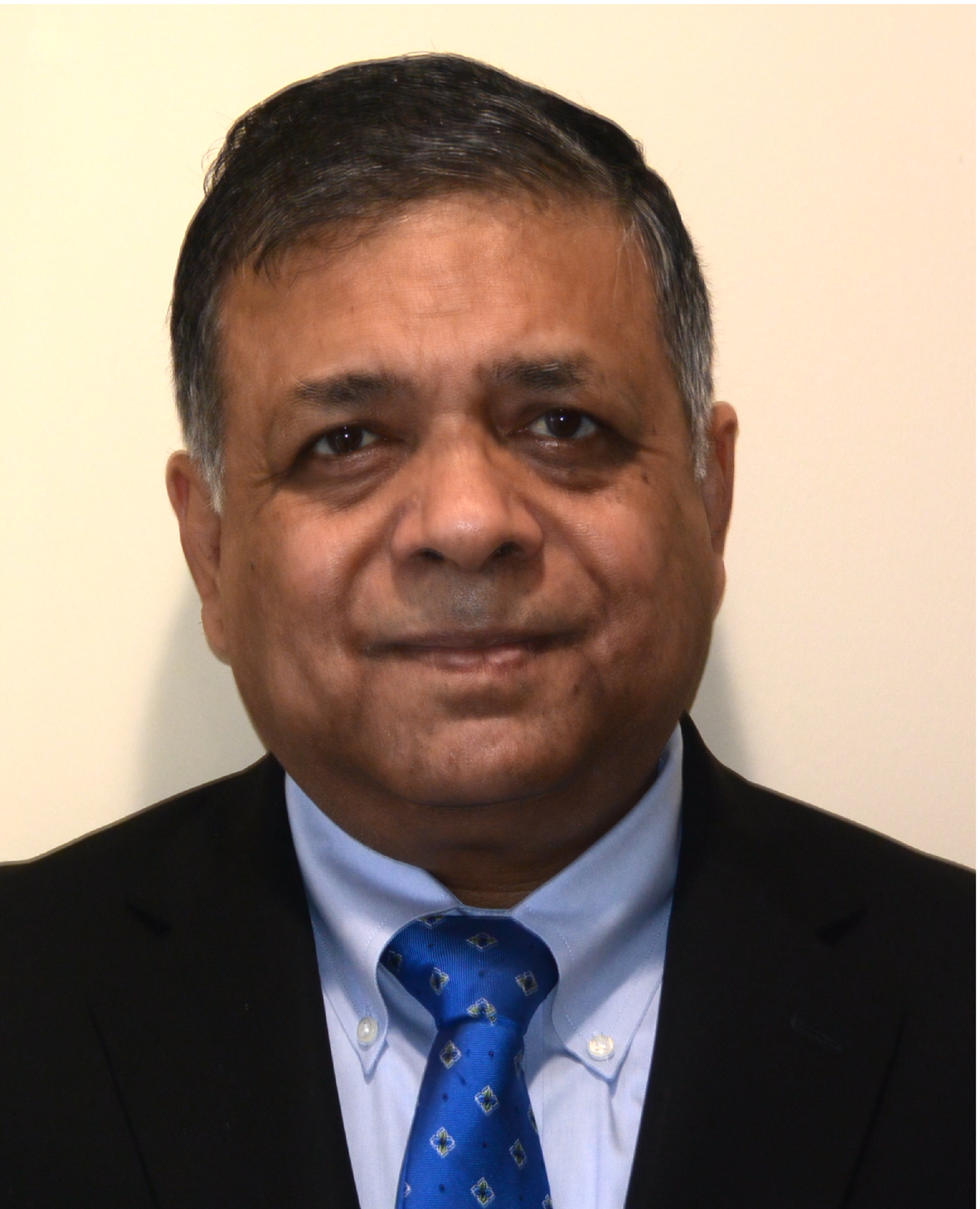}}]
{Amitabha (Amitava) Ghosh} (M'86$–-$SM'06$–-$F'15) is a Nokia Fellow and Head, Small Cell Research at Nokia Bell Labs. He joined Motorola in 1990 after receiving his Ph.D in Electrical Engineering from Southern Methodist University, Dallas.  Since joining Motorola he worked on multiple wireless technologies starting from IS-95, cdma-2000, 1xEV-DV/1XTREME, 1xEV-DO, UMTS, HSPA, 802.16e/WiMAX and 3GPP LTE. Dr. Ghosh has 60 issued patents, has written multiple book chapters and has authored numerous external and internal technical papers. He is currently working on 3GPP LTE-Advanced and 5G technologies. His research interests are in the area of digital communications, signal processing and wireless communications. He is a Fellow of IEEE, recipient of 2016 IEEE Stephen O. Rice prize and co-author of the book titled “Essentials of LTE and LTE-A”. 
\end{IEEEbiography}

\end{document}